\documentclass[journal]{IEEEtran}

\usepackage{graphicx}
\usepackage{amsfonts}
\usepackage{amsmath}
\usepackage{amssymb}
\usepackage{color}

\newtheorem{thm}{Theorem}[section]
\newtheorem{cor}[thm]{Corollary}
\newtheorem{lem}[thm]{Lemma}
\newtheorem{proof}[thm]{proof}

\newtheorem{defn}[thm]{Definition}
\newtheorem{rem}[thm]{Remark}
\newtheorem{exam}[thm]{Example}

\numberwithin{equation}{section}

\renewcommand{\baselinestretch}{1.5}

\ifCLASSINFOpdf \else \fi 
\begin{document}
\large \onecolumn

\title{On the Solvability of $3$s$/n$t Sum-Network --- A Region Decomposition
and Weak Decentralized Code Method}

\author{Wentu Song, ~  Kai Cai, ~ Chau Yuen, ~ and  ~ Rongquan Feng
\thanks{Wentu Song and Chau Yuen are with Singapore University of
Technology and Design, Singapore
       (e-mails: \{wentu$\_$song, yuenchau\}@sutd.edu.sg).}
\thanks{Kai Cai is with Department of Mathematics,
University of Hong Kong.
       (e-mail: kcai@hku.hk).}
\thanks{Rongquan Feng is with the LMAM, School of Mathematical Sciences,
Peking University, China
       (e-mail: fengrq@math.pku.edu.cn).}
} \maketitle

\begin{abstract}
We study the network coding problem of sum-networks with $3$
sources and $n$ terminals ($3$s$/n$t sum-network), for an
arbitrary positive integer $n$, and derive a sufficient and
necessary condition for the solvability of a family of so-called
``terminal-separable'' sum-network. Both the condition of
``terminal-separable'' and the solvability of a terminal-separable
sum-network can be decided in polynomial time. Consequently, we
give another necessary and sufficient condition, which yields a
faster $(O(|E|)$ time$)$ algorithm than that of Shenvi and Dey
$($\cite{Shenvi10}, $(O(|E|^3)$ time$)$, to determine the
solvability of the $3$s$/3$t sum-network.

To obtain the results, we further develop the region decomposition
method in \cite{Wentu11, Wentu12} and generalize the decentralized
coding method in \cite{Fragouli06}. Our methods provide new
efficient tools for multiple source multiple sink network coding
problems.
\end{abstract}


\IEEEpeerreviewmaketitle

\section{Introduction}
A $k$-source $n$-terminal ($k$s$/n$t) sum-network is modelled as a
directed, acyclic, finite graph $G$ with a set of $k$ vertices
$\{s_1,\cdots,s_k\}$ called sources and a set of $n$ vertices
$\{t_1,\cdots,t_n\}$ called terminals (sinks), such that each
source $s_i$ generates a message $X_i\in\mathbb F$ and each
terminal wants to obtain the sum $\sum_{i=1}^kX_i$ using linear
network coding \cite{Ahlswede00, Li03}, where $\mathbb F$ is a
finite field. The problem of communicating the sum over networks
is in fact a subclass of the problem of \emph{distributed function
computation} over networks, which has been explored in different
contexts \cite{Korner79}-\cite{Kannan}, most are from information
theoretic perspective and focus on small or simple networks.

The study of network coding for sum-networks began with A.
Ramamoorthy \cite{Rama08} and was investigated from several
aspects recently \cite{Rama08}-\cite{Rai13}. In \cite{Rama08}, it
was shown that for a $k$s$/2$t or $2$s$/n$t sum-network with unit
capacity edges and independent, unit-entropy sources, it is
solvable if and only if every source-terminal pair is connected.

The main focus of the study on the sum-network is to determine the
solvability. One direction is to find out necessary and sufficient
solvability conditions for $k$s$/n$t sum-network of larger
integers $k$ and $n$. In this direction, successes to date are
sporadic and the best known result is from Shenvi and Dey
\cite{Shenvi10}, where the authors characterize the solvability of
$3$s$/3$t sum-networks using a collection of six ``connection
conditions''. Another direction is to derive bounds on the
capacity, or similarly --- sufficient or necessary conditions on
the solvability of sum-network. In \cite{Rai10}, some upper and
lower bounds on the capacity of sum-network are presented. These
bounds are observed to be loose for the case of $k>2$ and $n>3$.
In \cite{Rai12}, the authors proved that the linear solvability of
a sum-network is equivalent to the linear solvability of some
multiple-unicast network and vice versa. They also proved that for
any set of polynomials having integer coefficients, there exists a
sum-network that is scalar linear solvable over a finite field
$\mathbb F$ if and only if the polynomials have a common root in
$\mathbb F$. However, both the multiple-unicast network coding
problem and the problem of solving the common root of polynomials
with integer coefficients are NP-hard, which indicate that the
$k$s$/n$t sum-network problem for $k>2$ and $n>3$ is very
challenging.

Like the other settings of multiple-source multiple-sink network
coding problem, loosely speaking, the difficulties of the
sum-network problem lie in two aspects: one is the topological
structure of the underlying network, which can be of arbitrary
complex; the other is the distributed function computation nature
of the coding requirements. So what is of most importance than new
results, is to develop new methodologies to  handle the
difficulties mentioned above, both for analyzing of network
structure and for coding design.

The first methodology that we developed is focusing on the network
structure. The region decomposition method, which has been found
efficient for analyzing network structure and very successful in
the 2-multicast networks \cite{Wentu11} and 2-unicast networks
with non-single rates \cite{Wentu12}, is further developed in this
paper. In \cite{Wentu11}, we proved that each directed, acyclic
network $G$ has a unique ``{\em basic region graph}'', denote by
$\text{RG}(D^{**})$, and the network coding problem on $G$ can be
converted to a coding problem on $\text{RG}(D^{**})$. Generally,
$\text{RG}(D^{**})$ could have a much more ``simpler'' topological
structure than $G$. In the present paper, we further develop this
method and decompose a network $G$ by two steps. In the first
step, we perform basic region decomposition and construct the
basic region graph, just as in \cite{Wentu11}. In the second step,
we further decompose the basic region graph into mutually disjoint
parts and consider the possibility of designing valid code in each
part separately. Most important, there is a special subgraph,
denoted by $\Pi$, such that to determine solvability of $G$, we
only need to consider code design problem on $\Pi$. Specifically,
$G$ is solvable if and only if there exists a code on $\Pi$ which
satisfies some certain conditions. This significantly simplifies
the original problem of determining solvability.

Another methodology is the concept of ``weak decentralized code''
that generalize the ``decentralized code'' in \cite{Fragouli06}.
Rather than that in \cite{Fragouli06} where the coding vectors of
{\em any two} subtrees are linear independent, by weak
decentralized code, each ``equivalent family of regions'' is
assigned the same coding vector, and for any two regions, their
coding vectors are linearly independent if and only if they belong
to different equivalent families.

Combining the developed region decomposition and weak
decentralized code methods, we give a sufficient and necessary
condition for the solvability of a class of $3$s$/n$t sum-network,
termed \emph{terminal-separable region graph}, for arbitrary
$n\in\mathbb Z^+$. The condition can be verified in polynomial
time. Moreover, as two simple corollaries of our result, we prove
that: (1) A $3$s$/2$t sum-network is always solvable if each
source-terminal pair is connected; (2) A $3$s$/3$t sum-network is
unsolvable if and only if the basic region graph has some certain
fixed structure.

Following the technical line of this paper, we can also completely
characterize the solvability of $3$s$/4$t sum-network. Limited by
space of this paper, we leave the results of $3$s$/4$t sum-network
to a future paper.

This paper is organized as follows. In Section
\uppercase\expandafter{\romannumeral 2}, we introduce the network
coding model of sum-network. The method for region decomposition
and for decomposing the basic region graph is presented in
Sections \uppercase\expandafter{\romannumeral 3} and
\uppercase\expandafter{\romannumeral 4}. The method of weak
decentralized code for $3$s$/n$t sum-network is presented in
\uppercase\expandafter{\romannumeral 5}. We characterize
solvability for terminal-separable $3$s$/n$t sum-network in
Section \uppercase\expandafter{\romannumeral 6}. Finally, the
paper is concluded in Section \uppercase\expandafter{\romannumeral
7}.

%%%%%%%%%%%%%%%%%%%%%%%%%%%%%%%%%%%%%%%%%%%
\renewcommand\figurename{Fig}
\begin{figure*}[htbp]
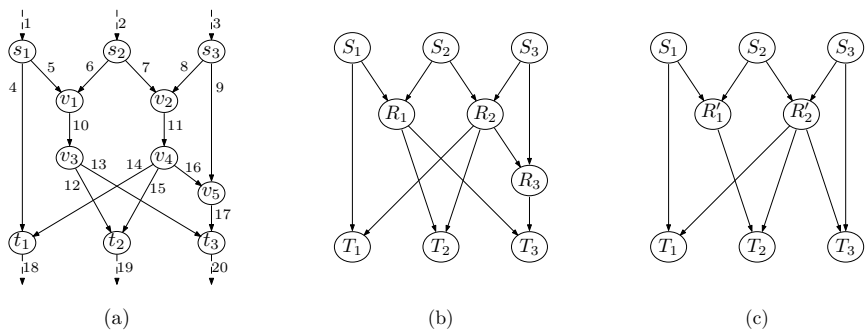

\begin{center}
\vspace{0.2cm}\includegraphics[height=4.3cm]{reg-dcm.1}
\hspace{1.2cm}\includegraphics[height=4.0cm]{reg-dcm.2}
\hspace{1.2cm}\includegraphics[height=4.0cm]{reg-dcm.3}
\end{center}
\caption{Examples of region graph: (a) is a $3$s$/3$t sum-network
$G_1$, where all links are sequentially indexed as
$1,2,\cdots,20$. Here, the imaginary links $1,2,3$ are the
$X_1,X_2,X_3$ source link, and $18,19,20$ are the terminal links
at terminal $t_1,t_2,t_3$ respectively. (b) is the region graph
$\text{RG}(D)$, where $S_1=\{1,4,5\},S_2=\{2,6,7\}, S_3=\{3,8,9\},
R_1=\{10,12,13\}, R_2=\{11,14,15,16\}, R_3=\{17\}, T_1=\{18\},
T_2=\{19\}, T_3=\{20\}$ and
$D=\{S_1,S_2,S_3,R_1,R_2,R_3,T_1,T_2,T_3\}$. (c) is the region
graph $\text{RG}(D')$, where $S_1=\{1,4,5\},S_2=\{2,6,7\},
S_3=\{3,8,9\}, R_1'=\{10,12,13\}, R_2'=\{11,14,15,16\},
T_1=\{18\}, T_2=\{19\}, T_3=\{17,20\}$ and
$D'=\{S_1,S_2,S_3,R_1',R_2',T_1,T_2,T_3\}$. In (c), although
$R_1'$ contains link $13$, which is an incoming link of link
$20\in T_3$, $R_1'$ is not a parent of $T_3$. This is because link
$20$ is not the leader of $T_3$. (Note that $\text{lead}(T_3)$ is
link $17$.)}\label{fg-ex-reg}
\end{figure*}
%%%%%%%%%%%%%%%%%%%%%%%%%%%%%%%%%%%%%%%%%%%%%%

\section{Models and Notations}
In this paper, we always denote $[n]=\{1,2,\cdots,n\}$ for any
positive integer $n$.

A $k$-source $n$-terminal ($k$s$/n$t) sum-network is a directed,
acyclic, finite graph $G=(V,E)$, where $V$ is the node (vertex)
set and $E$ is the link (edge) set. There is a set of $k$ vertices
$\{s_1, s_2, \cdots, s_k\}\subseteq V$ called sources and a set of
$n$ vertices $\{t_1, t_2, \cdots, t_n\}\subseteq V$ called
terminals (sinks) such that each source $s_i$ generates a message
$X_i \in \mathbb F$ and each terminal $t_j$ wants to get the sum
$\sum_{i=1}^kX_i$ by linear network coding, where $\mathbb F$ is a
finite field. Generally, for the sake of simplification, each link
$e$ of $G$ is further assumed to be error-free, delay-free and can
carry one symbol in each use.

For any link $e=(u,v)\in E$, $u$ is called the \emph{tail} of $e$
and $v$ is called the \emph{head} of $e$, and are denoted by
$u=\text{tail}(e)$ and $v=\text{head}(e)$, respectively. Moreover,
we call $e$ an incoming link of $v~($an outgoing link of $u)$. For
two links $e,e'\in E$, we call $e'$ an {\em incoming link} of $e$
($e$ an {\em outgoing link} of $e'$) if
$\text{tail}(e)=\text{head}(e')$. For any $e\in E$, denoted by
$\text{In}(e)$ the set of incoming links of $e$.

To aid analysis, we assume that each source $s_i$ has an imaginary
incoming link, called $X_{i}$ {\em source link} $($or \emph{source
link} for short$)$, and each terminal $t_{j}$ has an imaginary
outgoing link, called {\em terminal link}. Note that the source
links have no tail and the terminal links have no head. As a
result, the source links have no incoming link. For the sake of
convenience, if $e\in E$ is not a source link (resp. terminal
link), we call $e$ a \emph{non-source link} (resp.
\emph{non-terminal link}).

We assume that each non-source non-terminal link $e$ of $G$ is on
a path from some source to some terminal. Otherwise, $e$ has no
impact on the network coding of $G$ and can be removed from $G$.

Let $\mathbb F^k$ be the $k$-dimensional vector space over the
finite field $\mathbb F$. For any subset $A\subseteq\mathbb F^k$,
let $\langle A\rangle$ denote the subspace of $\mathbb F^k$
spanned by $A$. For $i\in\{1,\cdots,k\}$, let $\alpha_i$ denote
the vector of $\mathbb F^k$ with the $i$th component being one and
all other components being zero. Meanwhile, we denote
$$\bar{\alpha}=\sum_{i=1}^k\alpha_i=(1,1,\cdots,1)$$ i.e., the
vector with all components being one.

For any linear network coding scheme, the message transmitted
along any link $e$ is a linear combination
$M_e=\sum_{i=1}^kc_iX_i$ of the source messages, where
$c_i\in\mathbb F$. We use the vector of coefficients,
$d_e=(c_1,\cdots,c_k)$, to represent the message $M_e$ and call
$d_e$ the global encoding vector of $e$. To ensure the
computability of network coding, the outgoing message, as a
$k$-dimensional vector, must be in the span of all incoming
messages. Moreover, to ensure that all terminals receive the sum
$\sum_{i=1}^kX_i$, if $e$ is a terminal link of the sum-network,
then $d_e=\sum_{i=1}^k\alpha_i=\bar{\alpha}$. Thus, we can define
a linear network code of a $k$s$/n$t sum-network as follows:
\begin{defn}[Linear Network Code]\label{lnc}
A \emph{linear network code} (LC) of $G$ over a field $\mathbb F$
is a collection of vectors $C=\{d_{e}\in\mathbb F^k; e\in E\}$
such that
\begin{itemize}
    \item[(1)] $d_{e}=\alpha_{i}$ if $e$ is the $X_i$ source link
    $(i=1,\cdots,k)$;
    \item[(2)] $d_e\in\langle d_{e'}; e'\in\text{In}(e)\rangle$
    if $e$ is a non-source link.
\end{itemize}
The code $C=\{d_{e}\in\mathbb F^k; e\in E\}$ is said to be a
\emph{linear solution} of $G$ if $d_{e}=\bar{\alpha}$ for all
terminal link $e$.
\end{defn}

The vector $d_e$ is called the {\em global encoding vector} of
link $e$. The network $G$ is said to be {\em solvable} if it has a
linear solution over some finite field $\mathbb F$. Otherwise, it
is said {\em unsolvable}.

\section{Region Decomposition Approach for Network Coding}
In this section, we present the region decomposition approach,
which will take a key role in our discussion. The basic concepts
of region decomposition can be seen in \cite{Wentu11,Wentu12}. For
the sake of self-containment, we briefly repeat its core in the
following. It should be mentioned that the basic idea of region
decomposition is also rooted in \cite{Fragouli06}.

\subsection{Region Decomposition and Region Graph}
In the following, we consider $G=(V,E)$ as a $k$s$/n$t
sum-network.
\begin{defn}[Region Decomposition]\label{Reg}
Let $R$ be a non-empty subset of $E$. $R$ is called a region of
$G$ if there is an $e_{l}\in R$ such that for any $e\in
R\backslash\{e_l\}$, $R$ contains an incoming link of $e$. If $E$
is partitioned into mutually disjoint regions, say
$R_{1},R_{2},\cdots,R_{N}$, then we call
$D=\{R_{1},R_{2},\cdots,R_{N}\}$ a region decomposition of $G$.
\end{defn}

The edge $e_{l}$ in Definition \ref{Reg} is called the
\emph{leader} of $R$ and is denoted as $e_{l}=\text{lead}(R)$. A
region $R$ is called the $X_{i}$ \emph{source region} (or a
\emph{source region} for short) if $\text{lead}(R)$ is the $X_{i}$
source link; $R$ is called a \emph{terminal region} if $R$
contains a terminal link. If $R$ is neither a source region nor a
terminal region, then $R$ is called a {\em coding region}. If $R$
is not a source region, we call $R$ a \emph{non-source region}.

Consider the network $G_1$ in Fig. \ref{fg-ex-reg} (a). We can
easily check that the subsets of links $R=\{3,8,9,11,16,17\}$ and
$R'=\{5,10,12,13,19\}$ are two regions of $G_1$ with
$\text{lead}(R)=3$ and $\text{lead}(R')=5$. However,
$\{5,6,10,12,13\}$ is not a region because it does not contain any
incoming link of link $5$ and $6$. Similarly, the subset
$\{10,11,12,13,14\}$ is not a region.

\begin{rem}\label{num-s-t-reg}
Since the source links have no incoming link, then each source
region contains exactly one source link, i.e., its leader. But a
terminal region may contains more than one terminal links. So
there are exactly $k$ source region and at most $n$ terminal
regions for any $k$s$/n$t sum-network. Without loss of generality,
we can assume $G$ has exactly $n$ terminal regions.
\end{rem}

\textbf{Convention}: In this work, we will always denote the $k$
different source regions as $S_1,\cdots,S_k$ and the $n$ terminal
regions as $T_1,\cdots,T_n$.

\begin{defn}[Region Graph]\label{reg-g}
Let $D$ be a region decomposition of $G$. The region graph of $G$
about $D$ is a directed, simple graph with vertex set $D$ and edge
set $\mathcal E_D$, where $\mathcal E_D$ is the set of all ordered
pairs $(R',R)$ such that $R'$ contains an incoming link of
$\text{lead}(R)$.
\end{defn}

Examples of region graph are given in Fig. \ref{fg-ex-reg} (b) and
(c). In general, a network may have many region decompositions,
hence has many region graphs.

We use $\text{RG}(D)$ to denote the region graph of $G$ about $D$,
i.e., $\text{RG}(D)=(D,\mathcal E_D)$. If $(R',R)$ is an edge of
$\text{RG}(D)$, we call $R'$ a \emph{parent} of $R~(R$ a
\emph{child} of $R')$. For $R\in D$, we use $\text{In}(R)$ to
denote the set of parents of $R$ in $\text{RG}(D)$.

\begin{rem}\label{snp-arg}
Note that the leader of each source region is the corresponding
source link and the source links have no incoming link. So the
source regions have no parent. Moreover, since $G$ is acyclic,
then clearly, $\text{RG}(D)$ is acyclic.
\end{rem}

For $R,R'\in D$, a path in $\text{RG}(D)$ from $R'$ to $R$ is a
sequence of regions $\{R_0=R',R_1,\cdots,R_K=R\}$ such that
$R_{i-1}$ is a parent of $R_i$ for each $i\in\{1,\cdots,K\}$. If
there is a path from $R'$ to $R$, we say $R'$ is connected to $R$
and denote $R'\rightarrow R$. Else, we say $R'$ is not connected
to $R$ and denote $R'\nrightarrow R$. In particular, we have
$R\rightarrow R$ for all $R\in D$.

\begin{defn}[Codes on Region Graph]\label{lnc-reg-g}
A \emph{linear code} (LC) of the region graph $\text{RG}(D)$ over
the field $\mathbb F$ is a collection of vectors
$\tilde{C}=\{d_{R}\in\mathbb F^k; R\in D\}$ such that
\begin{itemize}
    \item[(1)] $d_{S_i}=\alpha_{i}$ for each $i\in\{1,\cdots,k\}$,
    where $S_i$ is the $X_i$ source region;
    \item[(2)] $d_R\in\langle d_{R'}; R'\in\text{In}(R)\rangle$
    if $R$ is a non-source region.
\end{itemize}
The code $\tilde{C}$ is said to be a \emph{linear solution} of
$\text{RG}(D)$ if $d_{T_j}=\bar{\alpha}$ for each terminal region
$T_j$. $\text{RG}(D)$ is said to be {\em feasible} if it has a
linear solution over some finite field $\mathbb F$. Otherwise, it
is said {\em infeasible}.
\end{defn}

The vector $d_R$ is called the {\em global encoding vector} of
$R$. By Definition \ref{lnc-reg-g}, for any linear solution of
$\text{RG}(D)$, it always be that $d_{S_i}=\alpha_i$ and
$d_{T_j}=\bar{\alpha}$. So in order to obtain a solution of
$\text{RG}(D)$, we only need to specify the global encoding vector
$d_R$ for each coding region $R\in D$.

Any linear solution of $\text{RG}(D)$ can be extended to a linear
solution of $G$. In fact, we have the following lemma. %whose proof
%is the same as Lemma 3.11 of \cite{Wentu11}.
\begin{lem}\label{code-ext}
Let $D$ be a region decomposition of $G$ and
$\tilde{C}=\{d_{R}\in\mathbb F^3; R\in D\}$ be a linear solution
of $\text{RG}(D)$. Let $d_e=d_R$ for each $R\in D$ and each $e\in
R$. Then $C=\{d_e; e\in E\}$ is a linear solution of $G$.
\end{lem}
\begin{proof}
For each link $e\in E$, by Definition \ref{Reg}, there is a unique
$R\in D$ such that $e\in R$. So $C$ is well defined. By Definition
\ref{lnc-reg-g}, we have $d_e=\alpha_i$ for each $X_i$ source link
$e$ and $d_e=\bar{\alpha}$ for each terminal link $e$. Moreover,
suppose $e\in E$ is a non-source link. By the same discussion as
in the proof of Lemma 3.11 of \cite{Wentu11}, we can prove that
$d_e\in\langle d_{e'}; e'\in\text{In}(e)\rangle$. Thus, $C$ is a
linear solution of $G$.
\end{proof}

By Lemma \ref{code-ext}, if $D$ is a region decomposition of $G$
and $\text{RG}(D)$ is feasible, then $G$ is solvable. But
conversely, if $G$ is solvable, it is not necessary that
$\text{RG}(D)$ is feasible.

For the region graph $\text{RG}(D)$ in Fig. \ref{fg-ex-reg} (b),
let $d_{R_1}=\alpha_1$ and $d_{R_2}=d_{R_3}=\alpha_2+\alpha_3$.
Then $ \tilde{C}=\{d_R;R\in D\}$ is a linear solution of
$\text{RG}(D)$ and we can obtain a linear solution of $G_1$ by
Lemma \ref{code-ext}. However, the region graph $\text{RG}(D')$ in
Fig. \ref{fg-ex-reg} (c) is not feasible because for any linear
code, by condition (2) of Definition \ref{lnc-reg-g}, we have
$d_{T_3}\in\langle d_{R_2'}, d_{S_3}\rangle$ and
$d_{R_2'}\in\langle d_{S_2}, d_{S_3}\rangle$, which implies that
$\alpha_1+\alpha_2+\alpha_3=\bar{\alpha}=d_{T_3}\in\langle
d_{S_2}, d_{S_3}\rangle=\langle\alpha_2, \alpha_3\rangle$, a
contradiction.

\subsection{Basic Region Graph}
In this subsection, we shall define a special region decomposition
$D^{**}$ of $G$, called the basic region decomposition of $G$,
which is unique and has the property that $G$ is solvable if and
only if the region graph $\text{RG}(D^{**})$ is feasible.

\begin{defn}[Basic Region Decomposition]\label{BRD}
Let $D^{**}$ be a region decomposition of $G$. $D^{**}$ is called
a basic region decomposition of $G$ if the following conditions
hold:
\begin{itemize}
    \item [(1)] For any $R\in D^{**}$ and any $e\in
    R\setminus\{\text{lead}(R)\}$, $\text{In}(e)\subseteq R$;
    \item [(2)] Each non-source region $R$ in $D^{**}$ has at least two
    parents in $\text{RG}(D^{**})$.
\end{itemize}
\end{defn}
Accordingly, the region graph $\text{RG}(D^{**})$ is called a
basic region graph of $G$.

For example, one can check that for the network $G_1$ in Fig.
\ref{fg-ex-reg} (a), the region graph $\text{RG}(D)$ in Fig.
\ref{fg-ex-reg} (b) is a basic region graph of $G_1$. However, the
region graph $\text{RG}(D')$ Fig. \ref{fg-ex-reg} (c) is not a
basic region graph of $G_1$ because $D'$ does not satisfies
condition (1) of Definition \ref{BRD}. In fact, for the link
$20\in T_3\in D'$, $20\neq\text{lead}(T_3)=17$ and
$\text{In}(20)=\{13,17\}\nsubseteq T_3$.

It is known that for networks with two unit-rate multicast
sessions, the basic region decomposition is unique [22, Th. 4.4]
and can be obtained in time $O(|E|)$ [22, Algorithm 5]. Note that
the notions of basic region decomposition and basic region graph
only depend on the topology of $G$ and have no relation with the
specific demand of sinks. Hence Algorithm 5 and Theorem 4.4 of
\cite{Wentu11} can be directly generalized to arbitrary directed
acyclic networks with $k$ sources and $n$ sinks (terminals),
including multicast networks and sum-networks. We list these
results for sum-network as the following theorem.

\begin{thm}\label{b-reg-unq} The $k$s$/n$t sum-network $G$ has a
unique basic region decomposition, hence has a unique basic region
graph. Moreover, the basic region decomposition and the basic
region graph can be obtained in time $O(|E|)$.
\end{thm}

\textbf{Convention}: In the following, we will always use $D^{**}$
and $\text{RG}(D^{**})$ to denote the basic region decomposition
and basic region graph of $G$.

Corresponding to Theorem 4.5 of \cite{Wentu11}, we have the
following theorem for sum-network.

\begin{thm}\label{solv-eqvlt}
$G$ is solvable if and only if $\text{RG}(D^{**})$ is feasible,
where $D^{**}$ is the basic region decomposition of $G$.
\end{thm}
\begin{proof}
Suppose $\text{RG}(D^{**})$ is feasible and $\tilde{C}$ is a
linear solution of $\text{RG}(D^{**})$. By Lemma \ref{code-ext},
we can obtain a linear solution $C$ of $G$. So $G$ is solvable.

Conversely, suppose $G$ is solvable and $C=\{d_e\in\mathbb F^k;
e\in E\}$ is a linear solution of $G$. Without loss of generality,
we can assume that $d_e\neq 0$ for all $e\in E$. By condition (2)
of Definition \ref{lnc} and condition (1) of Definition \ref{BRD},
we can easily see that for any $R\in D^{**}$ and any $e\in R$,
$d_e\in\langle d_{\text{lead}(R)}\rangle$. So we can further
assume that $d_e=d_{\text{lead}(R)}$. For each $R\in D^{**}$, let
$d_R=d_{\text{lead}(R)}$. Then by Definition \ref{reg-g} and
\ref{lnc-reg-g}, we can check that $\tilde{C}=\{d_R; R\in
D^{**}\}$ is a linear solution of $\text{RG}(D^{**})$. So
$\text{RG}(D^{**})$ is feasible.
\end{proof}

%%%%%%%%%%%%%%%%%%%%%%%%%%%%%%%%%%%%%%%%%%%
\renewcommand\figurename{Fig}
\begin{figure}[htbp]
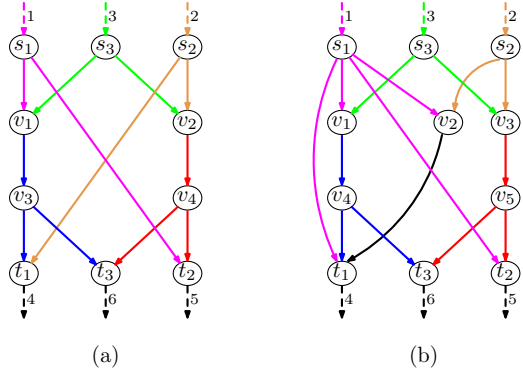

\begin{center}
%\vspace{0.2cm}
\includegraphics[height=4.9cm]{reg-dcm.4}
\hspace{1.3cm}\includegraphics[height=4.9cm]{reg-dcm.5}
\end{center}
\caption{Two examples of unsolvable $3$s$/3$t sum-network, where
$s_1,s_2,s_3$ are three sources and $t_1,t_2,t_3$ are three
terminals. The imaginary links $1,2,3$ are the $X_1,X_2,X_3$
source link, and the imaginary links $4,5,6$ are the terminal
links at terminal $t_1,t_2,t_3$ respectively.
}\label{fg-unsolv-net}
\end{figure}
%%%%%%%%%%%%%%%%%%%%%%%%%%%%%%%%%%%%%%%%%%%%%%

%%%%%%%%%%%%%%%%%%%%%%%%%%%%%%%%%%%%%%%%%%%
\renewcommand\figurename{Fig}
\begin{figure}[htbp]
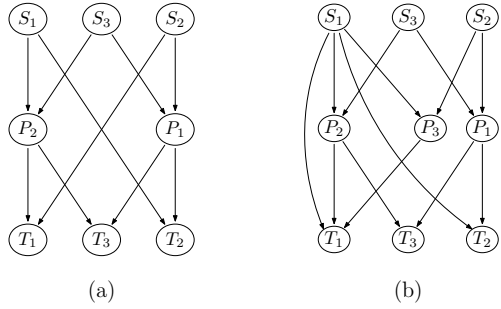

\begin{center}
%\vspace{0.2cm}
\includegraphics[height=4.0cm]{reg-dcm.6}
\hspace{1.3cm}\includegraphics[height=4.0cm]{reg-dcm.7}
\end{center}
\caption{The corresponding basic region graph of the networks in
Fig. \ref{fg-unsolv-net}. }\label{fg-unsolv-reg}
\end{figure}
%%%%%%%%%%%%%%%%%%%%%%%%%%%%%%%%%%%%%%%%%%%%%%

\begin{exam}\label{ex-Shenvi}
We consider two examples of sum-network in Fig.
\ref{fg-unsolv-net}, which can also be found in \cite{Shenvi10}.

In Fig. \ref{fg-unsolv-net} (a), let
$S_1=\{1,(s_1,v_1),(s_1,t_2)\}$, $S_2=\{2$,
$(s_2,v_2),(s_2,t_1)\}$, $S_3=\{3,(s_3,v_1),(s_3,v_2)\}$,
$P_1=\{(v_2,v_4)$, $(v_4,t_3),(v_4,t_2)\}$,
$P_2=\{(v_1,v_3),(v_3,t_1),(v_3,t_3)\}$, $T_1=\{4\}$, $T_2=\{5\}$
and $T_3=\{6\}$. Then $D^{**}=\{S_1,S_2,S_3,P_1$,
$P_2,T_1,T_2,T_3\}$ is its basic region decomposition. The basic
region graph is shown in Fig. \ref{fg-unsolv-reg} (a). If
$C=\{d_e\in\mathbb F^3; e\in E\}$ is a linear solution of $G$, the
global encoding vector of all links in the same region must be the
same. For example, $d_{(s_1,v_1)}=d_{(s_1,t_2)}=d_{1}$,
$d_{(v_3,t_1)}=d_{(v_3,t_3)}=d_{(v_1,v_3)}$, etc. So we can view
each region as a node and consider coding on the basic region
graph $\text{RG}(D^{**})$. We will show in Section
\uppercase\expandafter{\romannumeral 6}. C that
$\text{RG}(D^{**})$ is infeasible. So the original network is
unsolvable.

In Fig. \ref{fg-unsolv-net} (b), let
$S_1=\{1,(s_1,v_1),(s_1,v_2),(s_1,t_1)$, $(s_1,t_2)\}$,
$S_2=\{2,(s_2,v_2),(s_2,v_3)\}$, $S_3=\{3,(s_3,v_1)$,
$(s_3,v_3)\}$, $P_1=\{(v_3,v_5),(v_5,t_3)$, $(v_5,t_2)\}$,
$P_2=\{(v_1,v_4)$, $(v_4,t_1),(v_4,t_3)\}$, $P_3=\{(v_2,t_1)\}$,
$T_1=\{4\}$, $T_2=\{5\}$ and $T_3=\{6\}$. Then
$D^{**}=\{S_1,S_2,S_3,P_1,P_2,P_3,T_1,T_2,T_3\}$ is its basic
region decomposition. The basic region graph is shown in Fig.
\ref{fg-unsolv-reg} (b). We will also show in Section
\uppercase\expandafter{\romannumeral 6}. C that its basic region
graph is infeasible. So this network is unsolvable.
\end{exam}

\begin{lem}\label{P-C-fsb}
Suppose $\Theta\subseteq D^{**}$ and for each $j\in[n]$, there is
a $Q\in\Theta$ such that $Q\rightarrow T_j$. If the sum of source
messages $\sum_{i=1}^kX_i$ can be transmitted from
$\{S_1,\cdots,S_k\}$ to all $Q\in\Theta$ simultaneously, then
$\text{RG}(D^{**})$ is feasible.
\end{lem}

This lemma is obvious because if a region can receive the sum,
then all its down-link regions can receive the sum.

\begin{exam}\label{ex-P-C-fsb}
Consider the region graph $\text{RG}(D^{**})$ in Fig. \ref{fg-t-s}
(a). Let $\Theta=\{Q,T_1\}$. Then the sum $X_1+X_2+X_3$ can be
transmitted from $\{S_1,S_2,S_3\}$ to $\{Q,T_1\}$ by letting
$d_{R_1}=d_{R_2}=\alpha_1, d_{R_3}=\alpha_2+\alpha_3$ and
$d_{Q}=\alpha_1+\alpha_2+\alpha_3$. Moreover, by letting
$d_{T_2}=d_{T_3}=\alpha_1+\alpha_2+\alpha_3$, we can obtain a
linear solution of $\text{RG}(D^{**})$. We remind the reader that
for $k=3$, the vectors $\alpha_1=(1,0,0), \alpha_2=(0,1,0),
\alpha_3=(0,0,1)$ are the global encoding vectors of the source
messages $X_1, X_2, X_3$ respectively, and
$\bar{\alpha}=\alpha_1+\alpha_2+\alpha_3=(1,1,1)$ is the global
encoding vector of the sum $\sum_{i=1}^3X_i$.
\end{exam}

%%%%%%%%%%%%%%%%%%%%%%%%%%%%%%%%%%%%%%%%%%%
\renewcommand\figurename{Fig}
\begin{figure}[htbp]
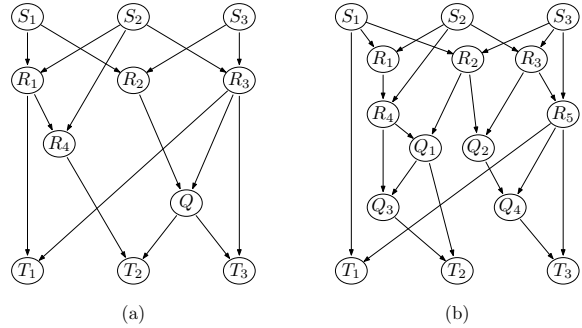

\begin{center}
\vspace{0.2cm}\includegraphics[height=4.3cm]{ex-t-s.1}
\hspace{0.9cm}\includegraphics[height=4.3cm]{ex-t-s.2}
\end{center}
\vspace{-0.2cm}\caption{Two examples of region
graph.}\label{fg-t-s}
\end{figure}
%%%%%%%%%%%%%%%%%%%%%%%%%%%%%%%%%%%%%%%%%%%%%%

If $T_j\rightarrow T_{j'}$ for some $\{j,j'\}\subseteq[n]$, then
we can reduce the number of terminal regions. In fact, without
loss of generality, assume $T_j\rightarrow T_{n}$ for some
$j\in\{1,\cdots,n-1\}$. Let $\Theta=\{T_1,T_2,\cdots,T_{n-1}\}$.
If $\text{RG}(D^{**})$ is feasible, then naturally, the sum
$\sum_{i=1}^kX_i$ can be transmitted to all $T_{j}\in\Theta$.
Conversely, if the sum can be transmitted to all $T_{j}\in\Theta$,
then by Lemma \ref{P-C-fsb}, $\text{RG}(D^{**})$ is feasible. So
we can reduce the number of terminal regions to $n-1$. For this
reason, we can assume that $T_j\nrightarrow T_{j'}$ for all pair
$\{j,j'\}\subseteq[n]$.

Since each non-source non-terminal link $e$ of $G$ is on a path
from some source to some terminal, then by Definition \ref{Reg}
and \ref{reg-g}, each region $R\in D^{**}$ is on a path from some
source region to some terminal region in $\text{RG}(D^{**})$.
Moreover, each terminal region $T_j$ has no child. Otherwise,
since $\text{RG}(D^{**})$ is acyclic, then by tracing child from
$T_j$, we can always find a path from $T_j$ to some other terminal
region $T_{j'}$, which contradicts to the assumption that
$T_j\nrightarrow T_{j'}$ for all $\{j,j'\}\subseteq[n]$. Thus, the
following assumption is reasonable.

\textbf{Assumption 1:} The terminal regions have no child and for
each region $R\in D^{**}$, $R\rightarrow T_j$ for some terminal
region $T_j$.

%\vspace{-0.4cm}
\subsection{Super Region}
As in \cite{Wentu12}, we can define super region of $G$ for the
region graph $\text{RG}(D^{**})$.

\begin{defn}[Super Region \cite{Wentu12}]\label{g-reg}
Suppose $\emptyset\neq\Theta\subseteq D^{**}$. The super region
generated by $\Theta$, denoted by $\text{reg}(\Theta)$, is a
subset of $D^{**}$ which is defined recursively as follows:
\begin{itemize}
    \item[(1)] If $R\in\Theta$, then $R\in\text{reg}(\Theta)$;
    \item[(2)] If $R\in D^{**}$ and $\text{In}(R)
    \subseteq\text{reg}(\Theta)$, then $R\in\text{reg}(\Theta)$.
\end{itemize}
\end{defn}
We define
$\text{reg}^\circ(\Theta)=\text{reg}(\Theta)\setminus\Theta$.
Moreover, if the subset $\Theta=\{R_1,\cdots,R_k\}$, then we
denote $\text{reg}(\Theta)=\text{reg}(R_1,\cdots,R_k)$.

Consider the region graph in Fig. \ref{fg-t-s} (a). We can find
the super region $\text{reg}(R_2,R_3,R_4)$ as follows. First, we
list all regions in a way that each region is before all of its
children. For example,
$\{S_1,S_2,S_3,R_1,R_2,R_3,R_4,Q,T_1,T_2,T_3\}$. Then we can check
all regions one by one to obtain $\text{reg}(R_2,R_3,R_4)$. In
fact, by Definition \ref{g-reg}, we have
$S_1,S_2,S_3,R_1\notin\text{reg}(R_2,R_3,R_4)$ and
$R_2,R_3,R_4\in\text{reg}(R_2,R_3,R_4).$ Note that
$\text{In}(Q)=\{R_2,R_3\}\subseteq\text{reg}(R_2,R_3,R_4)$. Then
by (2) of Definition \ref{g-reg}, we have
$Q\in\text{reg}(R_2,R_3,R_4).$ Since $T_1$ has a parent
$R_1\notin\text{reg}(R_2,R_3,R_4)$, then by Definition
\ref{g-reg}, $T_1\notin\text{reg}(R_2,R_3,R_4).$ Similarly,
$\{T_2,T_3\}\subseteq\text{reg}(R_2,R_3,R_4).$ Thus, we have
$\text{reg}(R_2,R_3,R_4)=\{R_2,R_3,R_4,Q,T_2,T_3\}$.

Consider the region graph in Fig. \ref{fg-t-s} (b). By a similar
discussion, we can check that
$\text{reg}(R_2,R_4)=\{R_2,R_4,Q_1,Q_3,T_2\}$ and
$\text{reg}(R_2,R_3,R_5)=\{R_2,R_3,R_5,Q_2,Q_4,T_3\}$.
%\end{exam}

In general, since $\text{RG}(D^{**})$ is acyclic, regions in
$D^{**}$ can be sequentially indexed as
$D^{**}=\{R_1,R_2,R_3,\cdots,R_N\}$ such that $R_i=S_i$ for
$i\in\{1,2,3\}$ and $\ell <\ell'$ if $R_{\ell}$ is a parent of
$R_{\ell'}$. By Definition \ref{g-reg}, it is easy to see that the
following Algorithm 1 output the super region $\text{reg}(\Theta)$
in time $O(|D^{**}|)$.

%\vspace{0.5cm}
\begin{center}
\setlength{\unitlength}{1mm}
\begin{picture}(90,30)(-2,-2)
\put(-1,-3){\line(1,0){87}} \put(-1,-3){\line(0,1){33}}

\put(86,30){\line(-1,0){87}} \put(86,30){\line(0,-1){33}}

\put(0,25){\textbf{Algorithm 1}: Super-Region
$(\text{RG}(D^{**}),\Theta)$:}

\put(5,20){$\mathcal R=\Theta$;}

\put(5,15){$\ell\leftarrow$ from $1$ to $N$}

\put(10,10){\textbf{If} $\text{In}(R)\subseteq\text{reg}(\Theta)$
\textbf{then}}

\put(15,5){$\mathcal R=\mathcal R\cup\{R_\ell\}$;}

\put(5,0){\textbf{Return} $\text{reg}(\Theta)=\mathcal R$;}
\end{picture}
\end{center}

\begin{rem}\label{rem-g-reg-code}
From condition (2) of Definition \ref{lnc-reg-g} and \ref{g-reg},
we can easily prove, using induction, that if
$\tilde{C}=\{d_{R}\in\mathbb F^k; R\in D^{**}\}$ is a linear code
of $\text{RG}(D^{**})$ and $\emptyset\neq\Theta\subseteq D^{**}$,
then $d_R\in\langle d_{R'};R'\in\Theta\rangle$ for all
$R\in\text{reg}(\Theta)$.
\end{rem}

Also by Definition \ref{g-reg}, it is easy to see that for any
subsets $\Theta_1,\Theta_2$ of $D^{**}$, if
$\Theta_1\subseteq\text{reg}(\Theta_2)$, then
$\text{reg}(\Theta_1)\subseteq\text{reg}(\Theta_2)$. In this
paper, we will always hold this fact as self-evident.

The following two lemmas are some other properties of super
region.

\begin{lem}\label{g-reg-cap}
Suppose $\Theta_1$ and $\Theta_2$ are two subsets of $D^{**}$.
Then
$\text{reg}(\Theta_1)\cap\text{reg}(\Theta_2)=\text{reg}(\Theta)$,
where
\begin{align}\label{eq-g-reg-cap-1}
\Theta=(\text{reg}(\Theta_1)\cap\Theta_2)
\cup(\text{reg}(\Theta_2)\cap\Theta_1).%\label{eq-g-reg-cap-1}
\end{align}
\end{lem}
\begin{proof}
By (\ref{eq-g-reg-cap-1}), $\Theta\subseteq\text{reg}(\Theta_1)$
and $\Theta\subseteq\text{reg}(\Theta_2)$. So by Definition
\ref{g-reg}, $\text{reg}(\Theta)\subseteq\text{reg}(\Theta_1)
\cap\text{reg}(\Theta_2)$.

We still need to prove that
$\text{reg}(\Theta_1)\cap\text{reg}(\Theta_2)
\subseteq\text{reg}(\Theta)$. Since $\text{RG}(D^{**})$ is
acyclic, then regions in
$\text{reg}(\Theta_1)\cap\text{reg}(\Theta_2)$ can be sequentially
indexed as $\{R_1,R_2,\cdots,R_N\}$ such that $\ell <\ell'$ if
$R_{\ell}$ is a parent of $R_{\ell'}$. For each $R_i$, if
$R_i\notin\Theta$, then by (\ref{eq-g-reg-cap-1}),
$R_i\notin\Theta_1\cup\Theta_2$. So
$R_i\in\text{reg}^\circ(\Theta_1)\cap\text{reg}^\circ(\Theta_2)$,
and by Definition \ref{g-reg},
$\text{In}(R_i)\subseteq\text{reg}(\Theta_1)\cap\text{reg}(\Theta_2)$,
which implies that
$\text{In}(R_i)\subseteq\{R_1,\cdots,R_{i-1}\}$. Thus,
$R_1\in\Theta\subseteq\text{reg}(\Theta)$ and by Definition
\ref{g-reg}, $R_i\in\text{reg}(\Theta)$ for $i\in\{2,\cdots,N\}$.
So we have $\text{reg}(\Theta_1)\cap\text{reg}(\Theta_2)
\subseteq\text{reg}^\circ(\Theta)$.

By above discussion,
$\text{reg}(\Theta_1)\cap\text{reg}(\Theta_2)=\text{reg}(\Theta)$.
\end{proof}

For any $\{Q_1,\cdots,Q_\ell,Q\}\subseteq D$, if each $Q_i$ has a
message $Y_i\in\mathbb F$ and $Q_i\rightarrow Q$, then any linear
combination $\sum_{i=1}^\ell\lambda_iY_i$ can be transmitted to
$Q$ by a linear network code. Formally, we have the following
lemma.
\begin{lem}\label{sub-g-reg-code}
Suppose $\Lambda=\{Q_1,\cdots,Q_\ell\}\subseteq D^{**}$ and $Q\in
D^{**}\backslash\Lambda$ such that for each $Q_i\in\Lambda$, there
is a path $\mathcal P_i$ from $Q_i$ to $Q$. Let $\Omega\subseteq
D^{**}$ be such that $\mathcal
P_i\backslash\Lambda\subseteq\Omega, \forall
i\in\{1,\cdots,\ell\}$. Then for any
$\{d_{Q_1},\cdots,d_{Q_\ell}\}\subseteq\mathbb F^k$ and any
$d_0\in\langle d_{Q_1},\cdots,d_{Q_\ell}\rangle$, there is a code
$\tilde{C}_{\Omega}=\{d_R\in\mathbb F^k;R\in\Omega\}$ such that
$d_Q=d_0$ and $d_R\in\langle d_{R'};R'\in\text{In}(R)\rangle$ for
all $R\in\Omega$.
\end{lem}

This lemma is obvious and we omit its proof.

\subsection{Weak Decentralized Code On Super Region}
For any set $A$, a collection $\mathcal I=\{\Delta_1, \cdots,
\Delta_K\}$ of subsets of $A$ is called a partition of $A$ if
$\Delta_1, \cdots, \Delta_K$ are mutually disjoint and
$\bigcup_{i=1}^K\Delta_i=A$. Each subset $\Delta_i$ is called an
\emph{equivalent class} of $A$. If $a\in\Delta_i$, we denote
$\Delta_i=[a]$. Thus, for each $\Delta_i$, we can pick an
arbitrary $a_i\in\Delta_i$ and denote $\mathcal I=\{\Delta_1,
\cdots, \Delta_K\}=\{[a_1],\cdots,[a_K]\}$. The element $a_i$ is
called a representative element of $\Delta_i$. Note that if
$a_i,b_i\in\Delta_i$, then we have $\Delta_i=[a_i]=[b_i]$.

For any subset $\Omega\subseteq D^{**}$ and any collection
$\tilde{C}_\Omega=\{d_R\in\mathbb F^k; R\in\Omega\}$, we call
$\tilde{C}_\Omega$ a code on $\Omega$. In the following, we give
an approach to construct a code on a super region. Such a code has
some interesting property and is the basis of our code
construction for three-source sum-network.

\begin{defn}\label{R-closed}
Suppose $\{Q_1,Q_2\}\subseteq D^{**}$. A partition $\mathcal
I=\{\Delta_1,\Delta_2,\cdots,\Delta_K\}$ of
$\text{reg}(Q_{1},Q_{2})$ is said to be R-\emph{closed} if $K\geq
2$ and $\text{reg}(\Delta_j)=\Delta_j$ for all
$j\in\{1,\cdots,K\}$.
\end{defn}

In Definition \ref{R-closed}, it must be that $Q_1,Q_2$ belongs to
different equivalent classes. This is because if
$Q_1,Q_2\in\Delta_i$, then
$\text{reg}(Q_{1},Q_{2})\subseteq\text{reg}(\Delta_i)=\Delta_i$
and we have $K=1$, which contradicts to the condition that $K\geq
2$. Thus, by proper naming, we can always assume that
$Q_1\in\Delta_1$ and $Q_2\in\Delta_2$.
\begin{defn}[Weak Decentralized Code on $\text{reg}(Q_{1},Q_{2})$]
\label{dct-code} Suppose $\mathcal
I=\{\Delta_1,\Delta_2,\cdots,\Delta_K\}$ is an R-closed partition
of $\text{reg}(Q_{1},Q_{2})$ and $d_1,d_2,\cdots,d_K\in\mathbb
F^k$ are mutually linearly independent such that $d_j\in\langle
d_1,d_2\rangle$ for all $j\geq 3$. For each $j\in\{1,2,\cdots,K\}$
and $R\in\Delta_j$, let $d_R=d_j$. The code
$\tilde{C}_{Q_1,Q_2}=\{d_R; R\in\text{reg}(Q_{1},Q_{2})\}$ is
called an $\mathcal I$-\emph{weak decentralized code} on
$\text{reg}(Q_{1},Q_{2})$.
\end{defn}

As a simple result of linear algebra, if $|\mathbb F|\geq K-1$,
then we can always find a set of vectors
$\{d_1,d_2,\cdots,d_K\}\subseteq\mathbb F^k$ satisfying the
condition of Definition \ref{dct-code}. Thus, if $|\mathbb F|\geq
K-1$, then we can always construct an $\mathcal I$-weak
decentralized code on $\text{reg}(Q_{1},Q_{2})$.

%%%%%%%%%%%%%%%%%%%%%%%%%%%%%%%%%%%%%%%%%%%
\renewcommand\figurename{Fig}
\begin{figure}[htbp]
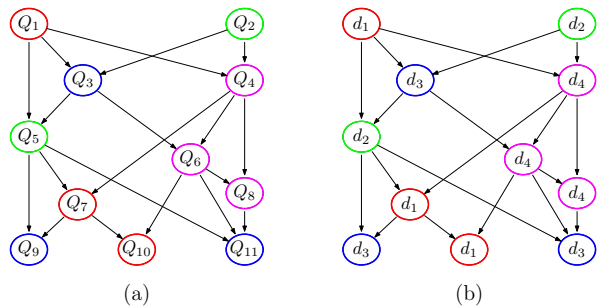

\begin{center}
\hspace{0.0cm}\includegraphics[height=4.0cm]{dct-code.1}
\hspace{0.9cm}\includegraphics[height=4.0cm]{dct-code.2}
\end{center}
\caption{An example of weak decentralized code: (a) depicts a
super region $\text{reg}(Q_{1},Q_{2})$ and (b) depicts a weak
decentralized code on $\text{reg}(Q_{1},Q_{2})$. In (a), let
$\Delta_1=\{Q_1,Q_7,Q_{10}\}, \Delta_2=\{Q_2,Q_5\},
\Delta_3=\{Q_3,Q_9,Q_{11}\}$ and $\Delta_4=\{Q_4,Q_6,Q_8\}$. Then
$\text{reg}(\Delta_j)=\Delta_j$ for all $j\in\{1,2,3,4\}$. Let
$d_1,d_2,d_3,d_4\in\mathbb F^k$ be mutually linearly independent
and $d_3,d_4\in\langle d_1,d_2\rangle$. Then the code illustrated
in Fig. \ref{fg-dct-code} (b) is a decentralized code on
$\text{reg}(Q_{1},Q_{2})$. }\label{fg-dct-code}
\end{figure}
%%%%%%%%%%%%%%%%%%%%%%%%%%%%%%%%%%%%%%%%%%%%%%

An example of weak decentralized code is given in Fig.
\ref{fg-dct-code}. Moreover, weak decentralized code has the
following property.
\begin{lem}\label{lem-weak-dec-code}
Let $\tilde{C}_{Q_1,Q_2}$ be as in Definition \ref{dct-code}. Then
$d_R\in\langle d_{R'}; R'\in\text{In}(R)\rangle$ for all
$R\in\text{reg}^\circ(Q_{1},Q_{2})$.
\end{lem}
\begin{proof}
Suppose $R\in\text{reg}^\circ(Q_{1},Q_{2})$. Then by Definition
\ref{g-reg}, $\text{In}(R)\subseteq\text{reg}(Q_{1},Q_{2})$. We
have the following two cases:

Case 1: $\text{In}(R)\subseteq\Delta_j$ for some
$j\in\{1,2,\cdots,K\}$. Then by Definition \ref{g-reg} and
\ref{dct-code}, we have $R\in\text{reg}(\Delta_j)=\Delta_j$. Again
by Definition \ref{dct-code}, we have $d_R=d_{R'}=d_j$ for all
$R'\in\text{In}(R)$. So $d_R\in\langle d_{R'};
R'\in\text{In}(R)\rangle$.

Case 2: $\text{In}(R)\nsubseteq\Delta_j$ for all
$j\in\{1,2,\cdots,K\}$. Since $R$ has at least two parents
(Definition \ref{BRD}), then we can assume
$R'_1\in\text{In}(R)\cap\Delta_{\ell_1}$ and
$R'_2\in\text{In}(R)\cap\Delta_{\ell_2}$ for some
$\{\ell_1,\ell_2\}\subseteq\{1,2,\cdots,K\}$. By Definition
\ref{dct-code}, $d_{R'_1}=d_{\ell_1}$, $d_{R'_2}=d_{\ell_2}$ and
$d_R=d_\ell$ for some $\ell\in\{1,2,\cdots,K\}$. Also by
Definition \ref{dct-code}, $d_\ell, d_{\ell_1},
d_{\ell_2}\in\langle d_1,d_2\rangle$ and $d_{\ell_1}, d_{\ell_2}$
are linearly independent. So $d_R\in\langle d_{1},
d_{2}\rangle=\langle d_{\ell_1}, d_{\ell_2}\rangle=\langle
d_{R'_1},d_{R'_2}\rangle\subseteq\text{In}(R)\rangle$.
\end{proof}

A special case of weak decentralized code is that each $\Delta_j$
contains a single element, i.e., $\text{reg}(Q_{1},Q_{2})=\{Q_1$,
$Q_2,\cdots,Q_K\}$ and $\Delta_j=\{Q_j\}, \forall
j\in\{1,2,\cdots,K\}$. In this case, $d_R$ and $d_{R'}$ are
linearly independent for all $\{R,
R'\}\subseteq\text{reg}(Q_{1},Q_{2})$. Such code is called
decentralized code \cite{Fragouli06}.

\section{Decomposition of the Basic Region Graph}
Throughout this section, we assume $G$ is a $3$s$/n$t sum-network.
By Theorem \ref{solv-eqvlt}, $G$ is solvable if and only if the
the basic region graph $\text{RG}(D^{**})$ is feasible. Thus, to
study the network coding problem of $G$, it is sufficient to
consider coding on $\text{RG}(D^{**})$. By Remark
\ref{num-s-t-reg}, $\text{RG}(D^{**})$ has exactly three source
regions and at most $n$ terminal regions. Without loss of
generality, we assume $\text{RG}(D^{**})$ has exactly $n$ terminal
regions. Recall that $S_i~(i\in\{1,2,3\})$ denote the $X_i$ source
region and $\{T_j; j\in[n]\}$ denote the set of $n$ terminal
regions.

For each $i\in\{1,2,3\}$, by Lemma \ref{g-reg-cap}, we have
\begin{align}
\text{reg}(S_{i},S_{i_1})\cap\text{reg}(S_{i},S_{i_2})=\{S_i\}
\label{eq-nt-1}
\end{align}
where $\{i_1,i_2\}=\{1,2,3\}\backslash\{i\}$. So by Definition
\ref{g-reg}, we have
\begin{align}
\text{reg}^\circ(S_{i},S_{i_1})\cap\text{reg}^\circ(S_{i},S_{i_2})
=\emptyset.\label{eq-nt-2}
\end{align}
Thus, the three subsets $\text{reg}^\circ(S_{1},S_{2})$,
$\text{reg}^\circ(S_{1},S_{3})$ and
$\text{reg}^\circ(S_{2},S_{3})$ are mutually disjoint.

To design codes on $\text{RG}(D^{**})$, we find it convenient to
decompose $\text{RG}(D^{**})$ into mutually disjoint parts
according to the connection condition of the source-terminal
pairs. In Subsection A, we will give a method to decompose
$\text{RG}(D^{**})$ and show some useful properties of such
decomposition.

\subsection{Decomposition of $\text{RG}(D^{**})$}
We first specify some subsets of $D^{**}$, which leads to a
decomposition of $\text{RG}(D^{**})$ and will play an import role
in our study.
\begin{defn}\label{lmd-omd}
We specify some subsets of $D^{**}$ as follows.
\begin{itemize}
    \item[\textbf{(1)}] $\Pi\triangleq\text{reg}(S_1,S_2)
    \cup\text{reg}(S_1,S_3)\cup\text{reg}(S_2,S_3)$.
    \item[\textbf{(2)}] For $I\subseteq[n]$,
    $\Omega_I$ is the set of all $R\in D^{**}\backslash\Pi$ such
    that $R\rightarrow T_{j}$ for all $j\in I$ and
    $R\nrightarrow T_{j'}$ for all $j'\in[n]\setminus I$.
    \item[\textbf{(3)}] $\Lambda_I$ is the set of all $Q\in\Pi$
    such that $Q$ has a child $R\in\Omega_I$.
\end{itemize}
\end{defn}
If the subset $I=\{i_1,\cdots, i_\ell\}$, we also denote
$\Omega_I=\Omega_{i_1,\cdots,i_\ell}$ and $\Lambda_I=\Lambda_{i_1,
\cdots, i_\ell}$.

\vspace{0.05in} \begin{exam}\label{eg-lmd-omd} We show some
examples of the subsets in Definition \ref{lmd-omd} for the region
graphs in Fig. \ref{fg-t-s}.

For the region graph in Fig. \ref{fg-t-s} (a). By Definition
\ref{g-reg}, we have $\text{reg}(S_1,S_2)=\{S_1,S_2,R_1,R_4\}$,
$\text{reg}(S_1,S_3)=\{S_1,S_3,R_2\}$ and
$\text{reg}(S_2,S_3)=\{S_2,S_3,R_3\}$. So by (1) of Definition
\ref{lmd-omd}, $\Pi=\{S_1,S_2,S_3,R_1,R_2,R_3,R_4\}$. By (2) of
Definition \ref{lmd-omd}, $\Omega_i=\{T_i\}$ for $i\in\{1,2,3\}$,
$\Omega_{2,3}=\{Q\}$ and
$\Omega_{1,2}=\Omega_{1,3}=\Omega_{1,2,3}=\emptyset$. By (3) of
Definition \ref{lmd-omd}, we have $\Lambda_1=\{R_1,R_3\}$,
$\Lambda_2=\{R_4\}$, $\Lambda_3=\{R_3\}$,
$\Lambda_{2,3}=\{R_2,R_3\}$.

For the region graph in Fig. \ref{fg-t-s} (b), by Definition
\ref{g-reg}, $\text{reg}(S_1,S_2)=\{S_1,S_2,R_1,R_4\}$,
$\text{reg}(S_1,S_3)=\{S_1,S_3,R_2\}$ and
$\text{reg}(S_2,S_3)=\{S_2,S_3,R_3, R_5\}$. So
$\Pi=\{S_1,S_2,S_3,R_1$, $R_2,R_3,R_4,R_5\}$. We can further check
that $\Omega_1=\{T_1\}$, $\Omega_2=\{Q_1, Q_3, T_2\}$,
$\Omega_3=\{Q_2, Q_4, T_3\}$ and
$\Omega_{1,2}=\Omega_{1,3}=\Omega_{2,3}=\Omega_{1,2,3}=\emptyset$.
By (3) of Definition \ref{lmd-omd}, we have
$\Lambda_1=\{S_1,R_5\}$, $\Lambda_2=\{R_2, R_4\}$ and
$\Lambda_3=\{R_2, R_3, R_5\}$.
\end{exam}

\begin{thm}\label{omg-intc}
The collection $\{\Pi\}\cup\{\Omega_I, I\subseteq[n]\}$ is a
partition of $D^{**}$ and it can be obtained in time
$O(|D^{**}|)$.
\end{thm}
\begin{proof}
By Definition \ref{lmd-omd}, it is easy to see that $D^{**}=
\Pi\bigcup(\bigcup_{I\subseteq[n]}\Omega_I)$ and the subsets $\Pi$
and $\Omega_I, I\subseteq[n],$ are mutually disjoint. So
$\{\Pi\}\cup\{\Omega_I, I\subseteq[n]\}$ is a partition of
$D^{**}$. We shall prove that it can be obtained in time
$O(|D^{**}|)$.

Since for each $\{i_1,i_2\}\subseteq\{1,2,3\}$, the super region
$\text{reg}(S_{i_1},S_{i_2})$ can be found by Algorithm 1, so the
subset $\Pi$ can be found by Algorithm 1 in time $O(|D^{**}|)$.

Note that $\text{RG}(D^{**})$ is acyclic and the terminal regions
have no child (Assumption 1), then regions in
$D^{**}\backslash\Pi$ can be sequentially indexed as
$D^{**}\backslash\Pi=\{R_1,R_2,R_3,\cdots,R_N\}$ such that
$R_{N-n+j}=T_j$ for $j\in[n]$ and $\ell <\ell'$ if $R_{\ell}$ is a
parent of $R_{\ell'}$. Then the collection $\{\Omega_I,
I\subseteq[n]\}$ can be found by tracing the parents back for all
terminals. See the following Algorithm 2:

\vspace{0.2cm}
\begin{center}
\setlength{\unitlength}{1mm}
\begin{picture}(100,40)(-2,-2)
\put(-1,-3){\line(1,0){97}} \put(-1,-3){\line(0,1){43}}

\put(96,40){\line(-1,0){97}} \put(96,40){\line(0,-1){43}}

\put(0,35){\textbf{Algorithm 2}: Labelling Algorithm
$(\text{RG}(D^{**}))$:}

\put(5,30){$j\leftarrow$ from $1$ to $n$}

\put(10,25){Label $R_{N-n+j}$ with $j$;}

\put(5,20){$\ell\leftarrow$ from $N-n$ to $1$}

\put(10,15){$j\leftarrow$ from $1$ to $n$}

\put(10,10){\textbf{if} $R_\ell$ has a child $R_{\ell'}$ such that
$R_{\ell'}$ is labelled}

\put(10,5){with $j$ \textbf{then}}

\put(15,0){Label $R_{\ell}$ with $j$;}
\end{picture}
\end{center}

Note that $R\rightarrow T_j$ if and only if $R$ is labelled with
$j$ by Algorithm 2.For each $R\in D^{**}\backslash\Pi$, let
$I_R=\{j\in[n]; R \text{~is labelled with~} j\}$. Then for each
$I\subseteq[n]$, we have $\Omega_I=\{R; I_R=I\}$. Clearly, the
time complexity of Algorithm 2 is $O(|D^{**}|)$.
\end{proof}

Decomposing $D^{**}$ into the subsets $\Pi$ and $\Omega_I,
I\subseteq[n],$ will be used to construct linear solution of
$\text{RG}(D^{**})$: We construct a code $\tilde{C}_\Pi=\{d_R;
R\in\Pi\}$ on $\Pi$ and a code $\tilde{C}_{\Omega_I}=\{d_R;
R\in\Omega_I\}$ on $\Omega_I$ for each $I\subseteq[n]$ such that
$\Omega_I\neq\emptyset$. Then we can potentially obtain a linear
solution $\tilde{C}$ of $\text{RG}(D^{**})$ by letting $\tilde{C}=
\tilde{C}_\Pi\textstyle\bigcup\left(\bigcup_{\emptyset\neq
I\subseteq[n]} \tilde{C}_{\Omega_I}\right)$. By Theorem
\ref{omg-intc}, $\Pi$ and all subsets $\Omega_I, I\subseteq[n],$
are mutually disjoint. So the code $\tilde{C}$ is well defined.
This method will be used to prove Theorem \ref{lmd-solv}.

\begin{lem}\label{in-omg-intc}
Let $R\in D^{**}$. Then $R\in D^{**}\backslash\Pi$ if and only if
$S_{i}\rightarrow R$ for all $i\in\{1,2,3\}$.
\end{lem}
\begin{proof}
Suppose $R\in D^{**}\backslash\Pi$. By (1) of Definition
\ref{lmd-omd}, $R\neq S_i\in\Pi$ for all $i\in\{1,2,3\}$. Let
$\{i_1,i_2\}=\{1,2,3\}\backslash\{i\}$. Again by (1) of Definition
\ref{lmd-omd}, $R\notin\text{reg}(S_{i_1},S_{i_2})\subseteq\Pi$.
By Definition \ref{g-reg}, $R$ has a parent, say $R_1$, such that
$R_1\notin\text{reg}(S_{i_1},S_{i_2})$. If $R_1\neq S_{i}$, then
similarly, $R_1$ has a parent $R_2$ such that
$R_2\notin\text{reg}(S_{i_1},S_{i_2})$. Since $\text{RG}(D^{**})$
is a finite graph, we can always find a path
$\{R_K,R_{K-1},\cdots,R_1,R\}$ such that $R_K=S_{i}$. Thus, we
have $S_{i}\rightarrow R$ for all $i\in\{1,2,3\}$.

Conversely, suppose $S_{i}\rightarrow R$ for all $i\in\{1,2,3\}$.
Then there is a path $\{S_{1},R_1,\cdots,R_{K-1},R_K=R\}$. Since
$S_1$ has no parent (Remark \ref{snp-arg}), then by Definition
\ref{g-reg}, $S_{1}\notin\text{reg}(S_{2},S_{3})$. By Definition
\ref{g-reg} and induction, $R_i\notin\text{reg}(S_{2},S_{3})$,
$i=1,\cdots,K-1,K$. So $R=R_K\notin\text{reg}(S_{2},S_{3})$.
Similarly, we can prove $R\notin\text{reg}(S_{1},S_{2})$ and
$R\notin\text{reg}(S_{1},S_{3})$. So by (1) of Definition
\ref{lmd-omd}, $R\notin\Pi$. Thus, $R\in D^{**}\backslash\Pi$.
\end{proof}

Clearly, if there is an $S_i$ and a $T_j$ such that
$S_{i}\nrightarrow T_j$, then the sum can't be transmitted to
$T_j$ and $\text{RG}(D^{**})$ is unsolvable. So we assume that
$S_{i}\rightarrow T_j$ for all $i\in\{1,2,3\}$ and $j\in[n]$. Then
by Lemma \ref{in-omg-intc}, $T_j\in D^{**}\backslash\Pi$.
Moreover, by Assumption 1, $T_j$ has no child. So $T_j\nrightarrow
T_{j'}$ for all $j'\neq j$. Thus, by (2) of Definition
\ref{lmd-omd}, we have the following remark.
\begin{rem}\label{T-in-omg}
For each $j\in[n]$, we have $T_j\in\Omega_j$.
\end{rem}

\subsection{Terminal-separable Region Graph}
In this subsection, we define a class of special region graph
called \emph{terminal-separable region graph} and prove that for
such region graph, the feasibility is determined by code on $\Pi$.

\begin{defn}[Terminal-separable Region Graph]\label{t-sep} The
region graph $\text{RG}(D^{**})$ is said to be
\emph{terminal-separable} if $\Omega_{I}=\emptyset$ for all
$I\subseteq[n]$ of size $|I|>1$.
\end{defn}

By Theorem \ref{omg-intc}, it is $O(|D^{**}|)$ time complexity to
determine whether $\text{RG}(D^{**})$ is terminal-separable.

According to Example \ref{eg-lmd-omd}, the region graph in Fig.
\ref{fg-t-s} (b) is terminal-separable. However, the region graph
in Fig. \ref{fg-t-s} (a) is not terminal-separable because
$\Omega_{2,3}=\{Q\}\neq\emptyset$.

In general, if a region graph is not terminal-separable, then it
can be viewed as a terminal-separable region graph with fewer
terminal regions. If the new one is feasible then the original one
is feasible. However, if the new one is infeasible then the
original one is not necessarily infeasible. For example, for the
graph in Fig. \ref{fg-t-s} (a), we can view $T_1$ and $Q$ as two
terminal regions and construct a linear code to transmit the sum
$\sum_{i=1}^3X_i$ to $T_1$ and $Q$. Then by Lemma \ref{P-C-fsb},
the sum can be transmitted from $Q$ to $T_2$ and $T_3$. (See
Example \ref{ex-P-C-fsb}.)

\begin{lem}\label{in-reg-lmd-t-s}
Suppose $\text{RG}(D^{**})$ is terminal-separable. Then for each
$j\in[n]$, the following hold.
\begin{itemize}
 \item [1)] $T_j\in\Omega_j\subseteq\text{reg}^\circ(\Lambda_j)$.
 \item [2)] $|\Lambda_j|\geq 2$ and $\Lambda_j\nsubseteq
 \text{reg}(S_{i_1},S_{i_2}), \forall\{i_1,i_2\}\subseteq\{1,2,3\}$.
 \item [3)] For each $Q\in\Lambda_j$, there is a path
 $\{Q, R_1, \cdots,R_K\}$ such that $\{R_1,\cdots,R_K\}\subseteq
 \Omega_j$ and $R_K=T_j$.
\end{itemize}
\end{lem}
\begin{proof}
1) By Remark \ref{T-in-omg}, we have $T_j\in\Omega_j$.

Since $\text{RG}(D^{**})$ is terminal-separable, then for any
$R\in\Omega_j$, by Definition \ref{t-sep} and (2), (3) of
Definition \ref{lmd-omd}, we have
$\text{In}(R)\subseteq\Lambda_j\cup\Omega_j$. Denote
$\Omega_j=\{R_1,R_2,\cdots,R_N\}$ such that $\ell <\ell'$ if
$R_{\ell}$ is a parent of $R_{\ell'}$. Then we have
$\text{In}(R_1)\subseteq\Lambda_j$. By Definition \ref{g-reg},
$R_i\in\text{reg}^\circ(\Lambda_j)$. Recursively, we have
$\text{In}(R_i)\subseteq\Lambda_j\cup\{R_1,\cdots,R_{i-1}\}$ and
$R_i\in\text{reg}^\circ(\Lambda_j), i=2,\cdots,N$. So
$\Omega_j\subseteq\text{reg}^\circ(\Lambda_j)$.

2) If $\Lambda_j\subseteq\text{reg}(S_{i_1},S_{i_2})$, then by 1)
and Definition \ref{g-reg},
$\Omega_j\subseteq\text{reg}^\circ(\Lambda_j)\subseteq
\text{reg}(S_{i_1},S_{i_2})$, which contradicts to (2) of
Definition \ref{lmd-omd}. So
$\Lambda_j\nsubseteq\text{reg}(S_{i_1},S_{i_2}),\forall
\{i_1,i_2\}\subseteq\{1,2,3\}$.

If $|\Lambda_j|=1$, say $\Lambda_j=\{Q\}$, then by (1), (3) of
Definition \ref{lmd-omd}, we have
$Q\in\text{reg}(S_{i_1},S_{i_2})$ for some
$\{i_1,i_2\}\subseteq\{1,2,3\}$. Thus,
$\Lambda_j=\{Q\}\subseteq\text{reg}(S_{i_1},S_{i_2})$, which
contradicts to the proved result that
$\Lambda_j\nsubseteq\text{reg}(S_{i_1},S_{i_2})$. So
$|\Lambda_j|\geq 2$.

3) For each $Q\in\Lambda_j$, by (3) of Definition \ref{lmd-omd},
$Q$ has a child, say $R_1$, such that $R_1\in\Omega_j$. By (2) of
Definition \ref{lmd-omd}, $R_1\rightarrow T_{j}$. Let
$\{R_1,\cdots,R_K\}$ be a path from $R_1$ to $T_j$ such that
$R_K=T_j$. Since $\text{RG}(D^{**})$ is terminal-separable, then
by Definition \ref{t-sep}, $\Omega_{I}=\emptyset$ for all
$I\subseteq[n]$ of size $|I|>1$. So it must be that
$\{R_1,\cdots,R_K\}\subseteq \Omega_j$.
\end{proof}

%%%%%%%%%%%%%%%%%%%%%%%%%%%%%%%%%%%%%%%%%%%
\renewcommand\figurename{Fig}
\begin{figure*}[htbp]
\begin{center}
\vspace{0.1cm}\includegraphics[height=2.8cm]{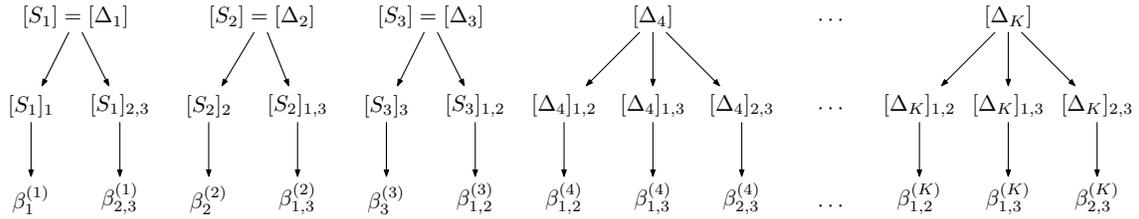}
\end{center}
\caption{Correspondence of equivalent classes, subclasses and
coding vectors: Given a partition $\mathcal
I=\{\Delta_1,\Delta_2,\Delta_3,\cdots,\Delta_K\}$ of $\Pi$ such
that $[S_i]=\Delta_i$ for $i\in\{1,2,3\}$, each equivalent class
is divided into some subclasses and each subclass corresponds to a
unique coding vector.}\label{fg-partition}
\end{figure*}
%%%%%%%%%%%%%%%%%%%%%%%%%%%%%%%%%%%%%%%%%%%%%%

For further discussion, we need the following definition.
\begin{defn}\label{sgl-sets}
A collection of vectors $\tilde{C}_\Pi=\{d_R\in\mathbb F^3;
R\in\Pi\}$ is said to be a feasible code on $\Pi$ if it satisfies
the following three conditions:
\begin{itemize}
  \item [(1)] $d_{S_i}=\alpha_i$ for each $i\in\{1,2,3\}$;
  \item [(2)] $d_R\in\langle d_{R'}; R'\in\text{In}(R)\rangle$
  for all $R\in\Pi\setminus\{S_1,S_2,S_3\}$;
  \item [(3)] $\bar{\alpha}\in\langle d_R; R\in\Lambda_j\rangle$
  for all $j\in[n]$.
\end{itemize}
\end{defn}

The following theorem shows that to determine feasibility of
$\text{RG}(D^{**})$, it is sufficient to determine existence of a
feasible code on $\Pi$.
\begin{thm}\label{lmd-solv}
Suppose $\text{RG}(D^{**})$ is terminal-separable. Then
$\text{RG}(D^{**})$ is feasible if and only if there exists a
feasible code on $\Pi$.
\end{thm}
\begin{proof}
$''\Rightarrow''$. Let $\tilde{C}=\{d_R\in\mathbb F^3; R\in
D^{**}\}$ be a linear solution of $\text{RG}(D^{**})$ and
$\tilde{C}_\Pi=\{d_R;R\in\Pi\}$ be the constraint of $\tilde{C}$
on $\Pi$. Then by Definition \ref{lnc-reg-g}, $\tilde{C}_\Pi$
satisfies conditions (1), (2) of Definition \ref{sgl-sets}.
Moreover, by 1) of Lemma \ref{in-reg-lmd-t-s},
$T_j\in\Omega_j\subseteq\text{reg}^\circ(\Lambda_j)$ for all $j\in
[n]$. Then by Remark \ref{rem-g-reg-code},
$\bar{\alpha}=d_{T_j}\in\langle d_R;R\in\Lambda_j\rangle$. So
$\tilde{C}_\Pi$ satisfies condition (3) of Definition
\ref{sgl-sets}. Thus, $\tilde{C}_\Pi$ is a feasible code on $\Pi$.

$''\Leftarrow''$. Suppose $\tilde{C}_\Pi$ is a feasible code on
$\Pi$. Then for each $j\in[n]$ by condition (3) of Definition
\ref{sgl-sets}, $\bar{\alpha}\in\langle d_R;
R\in\Lambda_j\rangle$. Moreover, for each $Q\in\Lambda_j$, by 3)
of Lemma \ref{in-reg-lmd-t-s}, there is a path $\mathcal
P_Q=\{Q,R_1,\cdots,R_K=T_j\}$ such that
$\{R_1,\cdots,R_K\}\subseteq\Omega_j$. Then by Lemma
\ref{sub-g-reg-code}, we can construct a code
$\tilde{C}_{\Omega_j}=\{d_R; R\in\Omega_j\}$ such that
$d_{T_j}=\bar{\alpha}$ and $d_R\in\langle d_{R'};
R'\in\text{In}(R)\rangle$ for all $R\in\Omega_j$. Note that
$\Omega_1, \cdots, \Omega_n$ are mutually disjoint (Theorem
\ref{omg-intc}). Then $\tilde{C}=\tilde{C}_\Pi\cup
\tilde{C}_{1}\cup\cdots \cup \tilde{C}_{n}$ is a linear solution
of $\text{RG}(D^{**})$ and $\text{RG}(D^{**})$ is feasible.
\end{proof}

Let $\tilde{C}_\Pi=\{d_R; R\in\Pi\}$ be a feasible code on $\Pi$.
Since $\text{RG}(D^{**})$ is acyclic and $d_{S_i}=\alpha_i\neq 0$
for $i\in\{1,2,3\}$. If there is an $R\in\Pi$ such that $d_R=0$,
then by tracing the parents of all $R$ such that $d_R=0$, we can
always find an $R_0\in\Pi$ such that $d_{R_0}=0$ and $d_{R'}\neq
0$ for some $R'\in\text{In}(R_0)$. We redefine $d_{R_0}$ by
letting $d_{R_0}=d_{R'}$. Then $d_{R_0}\neq 0$ and it is easy to
see that the resulted code is still a feasible code on $\Pi$. We
can perform this operation continuously until $d_R\neq 0$ for all
$R\in\Pi$ and the resulted code is still a feasible code on $\Pi$.
Thus, we have the following remark.
\begin{rem}\label{rem-sgl-sets}
Let $\tilde{C}_\Pi=\{d_R\in\mathbb F^3; R\in\Pi\}$ be a feasible
code on $\Pi$. We can always assume that $d_R\neq 0$ for all
$R\in\Pi$.
\end{rem}

\begin{lem}\label{lem-sgl-sets}
Let $\tilde{C}_\Pi=\{d_R\in\mathbb F^3; R\in\Pi\}$ be a feasible
code on $\Pi$. The following hold:
\begin{itemize}
  \item [1)] If $\Theta\subseteq\Pi$ and $R\in\text{reg}(\Theta)$,
  then $d_R\in\langle d_{R'};R'\in\Theta\rangle$. Thus, if
  $R\in\text{reg}(S_{i_1}, S_{i_2})$ for some
  $\{i_1, i_2\}\subseteq\{1, 2, 3\}$, then
  $d_R\in\langle \alpha_{i_1}, \alpha_{i_2}\rangle$.
  Moreover, If $\langle d_{R'}\rangle=\langle d_{R''}\rangle$ and
  $R\in\text{reg}(R', R'')$, then $\langle d_{R}\rangle
  =\langle d_{R'}\rangle=\langle d_{R''}\rangle$.
  \item [2)] If $\{R,R'\}=\Lambda_j$,
  then $\langle d_R,d_{R'}\rangle=\langle\bar{\alpha},
  d_{R}\rangle=\langle \bar{\alpha},d_{R'}\rangle$.
\end{itemize}
\end{lem}
\begin{proof}
1) is a direct consequence of Definition \ref{g-reg} and
conditions (1), (2) of Definition \ref{sgl-sets}.

2) By claim 1), we have $\{d_R,d_{R'}\}\subseteq\langle\alpha_{1},
\alpha_{2}\rangle\cup\langle\alpha_{1},
\alpha_{3}\rangle\cup\langle\alpha_{2}, \alpha_{3}\rangle$, which
implies that $\bar{\alpha}\notin\langle d_{R}\rangle$ and
$\bar{\alpha}\notin\langle d_{R'}\rangle$. By condition (3) of
Definition \ref{sgl-sets}, $\bar{\alpha}\in\langle d_{R},
d_{R'}\rangle$. So $\bar{\alpha}, d_{R}$ and $d_{R'}$ are mutually
linearly independent and $\langle \bar{\alpha},
d_{R}\rangle=\langle \bar{\alpha}, d_{R'}\rangle=\langle d_R,
d_{R'}\rangle$.
\end{proof}

\section{Weak Decentralized Code on ~$\Pi$}
In this section, we generalize the weak decentralized code on one
super region (Definition \ref{dct-code}) to $\Pi$, which is the
union of three super regions (Definition \ref{lmd-omd}). Our
discussions begin with a most general partition of $\Pi$ and its
refinement. Then we define week decentralized code on the so
called "R-closed partition of $\Pi$''(Definition
\ref{R-closed-Pi}.). Note that the construction of the R-closed
partition of $\Pi$ will be left to next section.

\subsection{Partition of ~$\Pi$ and Its Refinement}
Let $\mathcal I=\{\Delta_1,\cdots,\Delta_K\}$ be a partition of
$\Pi$. As mentioned before, for each $\Delta_i\in\mathcal I$, we
can choose an arbitrary $R\in\Delta_i$ as a representative element
and denote $\Delta_i=[R]$. On the other hand, for each $R\in\Pi$,
we have $R\in\Delta_i$ for some $\Delta_i\in\mathcal I$. We will
use $\Delta_i$ and $[R]$ interchangeably. We further assume that
$[S_i]\neq[S_j]$ for each pair $\{i,j\}\subseteq\{1,2,3\}$. Thus,
$K\geq 3$ and by proper naming, we can assume $\Delta_i=[S_i],
~i=1,2,3$. For each $[R]=\Delta_i\in\mathcal I$ and
$\{j_1,j_2\}\subseteq\{1,2,3\}$, we denote
\begin{align} \vspace{-0.1cm}
[R]_{j_1,j_2}=[\Delta_i]_{j_1,j_2}=\Delta_i\cap\text{reg}(S_{j_1},S_{j_2}).
\label{eq-nt-3}
\end{align}
For $i\in\{1,2,3\}$ and $\{j_1,j_2\}=\{1,2,3\}\backslash\{i\}$, we
denote
\begin{align} \vspace{-0.1cm}
[S_i]_i=[S_i]\cap(\text{reg}(S_{i},S_{j_1})\cup\text{reg}(S_{i},S_{j_2})).
\label{eq-nt-4}
\end{align} By (\ref{eq-nt-3}) and (\ref{eq-nt-4}), we have
$[S_i]_i=[S_i]_{i,j_1}\cup[S_i]_{i,j_2}$.

Given a partition $\mathcal I=\{\Delta_1,\cdots,\Delta_K\}$ of
$\Pi$ as above, we can further refine it as follows:
\begin{defn}[Subclass]\label{sub-class}
For $i\in\{1,2,3\}$, $[S_i]$ is divided into two \emph{subclasses}
$[S_i]_i$ and $[S_i]_{j_1,j_2}$, where
$\{j_1,j_2\}=\{1,2,3\}\backslash\{i\}$; For $i\geq 4$, $\Delta_i$
is divided into three \emph{subclasses} $[\Delta_i]_{1,2},
[\Delta_i]_{1,3}$ and $[\Delta_i]_{2,3}$.
\end{defn}

In what follows, an equivalent class of $\mathcal I$ refers to
$\Delta_i$ (or $[R]$) and a subclass of $\mathcal I$ refers to the
subclass of $\Delta_i$ defined above. Moreover, we usually use
$[[R]]$ to denote a subclass of $[R]$. Note that a subclass of
$[R]$ could be an empty set. The correspondence of equivalent
classes and subclasses are illustrated in Fig. \ref{fg-partition}.
By (\ref{eq-nt-1}), (\ref{eq-nt-2}), (\ref{eq-nt-3}) and
(\ref{eq-nt-4}), the collection of all subclasses of $\mathcal I$
is still a partition of $\Pi$.

%%%%%%%%%%%%%%%%%%%%%%%%%%%%%%%%%%%%%%%%%%%
\renewcommand\figurename{Fig}
\begin{figure*}[htbp]
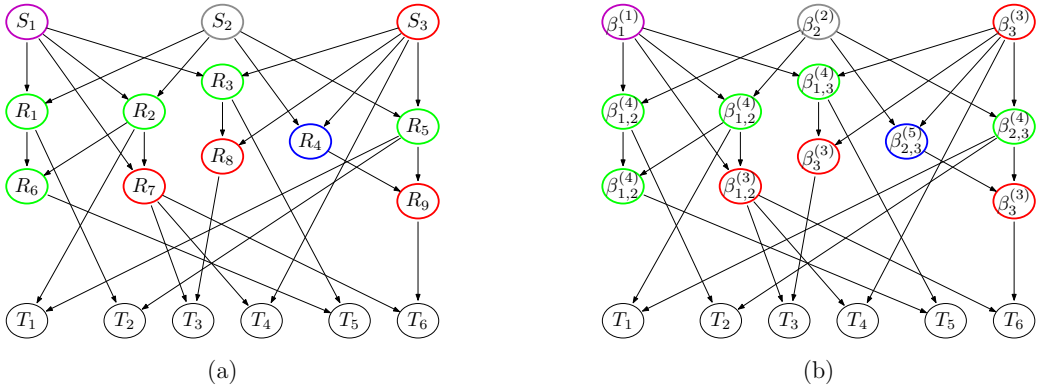

\begin{center}
\hspace{0.1cm}\includegraphics[height=5.1cm]{fg-smpl.7}
\hspace{2.0cm}\includegraphics[height=5.1cm]{fg-smpl.6}
\end{center}
\caption{An example of weak decentralized code on $\Pi$: (a) is
the example region graph, where the subset $\Pi$ is partitioned
into five equivalent classes, regions in the same equivalent class
are drawn in the same color; (b) is an illustration of the
corresponding weak decentralized code.}\label{fg-cmplx-1}
\end{figure*}
%%%%%%%%%%%%%%%%%%%%%%%%%%%%%%%%%%%%%%%%%%%%%%

\begin{exam}\label{ex-dcn-code}
Consider the region graph in Fig. \ref{fg-cmplx-1} (a). One can
check that $\text{reg}(S_1,S_2)=\{S_1,S_2,R_1,R_2,R_6$, $R_7\}$,
$\text{reg}(S_1,S_3)=\{S_1,S_3,R_3,R_8\}$ and
$\text{reg}(S_2,S_3)=\{S_2,S_3,R_4,R_5,R_9\}$. Thus,
$\Pi=\{S_1,S_2,S_3$, $R_1,\cdots,R_9\}$.

Let $\mathcal I_c=\{\Delta_1,\Delta_2,\Delta_3,\Delta_4,
\Delta_5\}$, where $\Delta_1=\{S_1\}$, $\Delta_2=\{S_2\}$,
$\Delta_3=\{S_3, R_7, R_8, R_9\}$, $\Delta_4=\{R_1, R_2, R_3, R_5,
R_6\}$ and $\Delta_5=\{R_4\}$. Then $\Delta_1$ has one nonempty
subclass $[S_1]_1=\{S_1\}$; $\Delta_2$ has one nonempty subclass
$[S_2]_2=\{S_2\}$; $\Delta_3$ has two nonempty subclasses
$[S_3]_3=\{S_3,R_8,R_9\}$ and $[S_3]_{1,2}=\{R_7\}$; $\Delta_4$
has three nonempty subclasses $[R_1]_{1,2}=\{R_1,R_2,R_6\}$,
$[R_1]_{1,3}=\{R_3\}$ and $[R_1]_{2,3}=\{R_5\}$; $\Delta_5$ has
one nonempty subclass $[R_4]_{2,3}=\{R_4\}$. All these subclasses
form the desired refinement of $\mathcal I_c$.
\end{exam}

Let $\mathcal I_0=\{[R]; R\in\Pi\}$, where $[R]=\{R\}$ for all
$R\in\Pi$. We call $\mathcal I_0$ the \emph{trivial partition} of
$\Pi$. Clearly, each equivalent class $[R]\in\mathcal I_0$ has
only one non-empty subclass, i.e., itself.

\subsection{Weak Decentralized Code on ~$\Pi$}
In this subsection, we construct weak decentralized code on $\Pi$.
All coding vectors will be taken from $\mathbb F^3$, where
$\mathbb F$ is a sufficient large field. Also note that for $k=3$,
the vectors $\alpha_1=(1,0,0), \alpha_2=(0,1,0), \alpha_3=(0,0,1)$
and $\bar{\alpha}=(1,1,1)$.

First, we give a lemma for constructing coding vectors.
\begin{lem}\label{genc-code}
Let $\mathbb F$ be a sufficiently large field. Then for any $K\geq
3$, there exist $K$ sets of vectors $\mathcal
B_1=\{\alpha_1,\alpha_2+\alpha_3\}$, $\mathcal
B_2=\{\alpha_2,\alpha_1+\alpha_3\}$, $\mathcal
B_3=\{\alpha_3,\alpha_1+\alpha_2\}$, $\mathcal
B_4=\{\beta^{(4)}_{1,2},\beta^{(4)}_{1,3},\beta^{(4)}_{2,3}\}$,
$\cdots$, $\mathcal
B_K=\{\beta^{(K)}_{1,2},\beta^{(K)}_{1,3},\beta^{(K)}_{2,3}\}\subseteq\mathbb
F^3$ and the following properties are satisfied:
\begin{itemize}
  \item [(1)] For any $\ell\in\{4,\cdots,K\}$ and $\{i_1,i_2\}
  \subseteq\{1,2,3\}$, $\beta^{(\ell)}_{i_1,i_2}\in\langle\alpha_{i_1},
  \alpha_{i_2}\rangle$;
  \item [(2)] If $\{\gamma, \gamma'\}\subseteq\mathcal B_\ell$ for
  some $\ell\in\{1,\cdots,K\}$, then $\bar{\alpha}\in\langle\gamma,
  \gamma'\rangle$;
  \item [(3)] For any pair $\{\gamma, \gamma'\}\subseteq
  \cup_{\ell=1}^K\mathcal
  B_\ell$, $\gamma$ and $\gamma'$ are linearly independent.
  \item [(4)] If $\{\gamma, \gamma',\gamma''\}\subseteq
  \bigcup_{\ell=1}^K\mathcal B_\ell$ such that $\{\gamma, \gamma',\gamma''\}
  \nsubseteq\langle\alpha_{i_1},\alpha_{i_2}\rangle$ for all $\{i_1,i_2\}
  \subseteq\{1,2,3\}$ and $\{\gamma,\gamma',\gamma''\}
  \neq\{\beta^{(\ell)}_{1,2},\beta^{(\ell)}_{1,3},\beta^{(\ell)}_{2,3}\}$
  for all $\ell\in\{4,\cdots,K\}$, then $\gamma, \gamma'$ and $\gamma''$
  are linearly independent;
\end{itemize}
\end{lem}
\begin{proof}
The proof is given in Appendix A.
\end{proof}

We give an example of sets of vectors satisfying properties
(1)$-$(4) of Lemma \ref{genc-code}. For simplicity, we assume that
$\mathbb F=GF(p)$ for a sufficiently large prime $p$.
\begin{exam}\label{ex-genc-code}
Let $\mathcal B_1=\{\alpha_1, \alpha_2+\alpha_3\}, \mathcal
B_2=\{\alpha_2, \alpha_1+\alpha_3\}, \mathcal B_3=\{\alpha_3,
\alpha_1+\alpha_2\}$, $\mathcal B_4=\{\alpha_1+3\alpha_2,
2\alpha_1+3\alpha_3$, $2\alpha_2-\alpha_3\}$ and $\mathcal
B_5=\{2\alpha_1+3\alpha_2, \alpha_1+3\alpha_3,
\alpha_2-2\alpha_3\}$. Then $\{\mathcal B_1,\mathcal B_2,\mathcal
B_3,\mathcal B_4,\mathcal B_5\}$ satisfies all conditions of Lemma
\ref{genc-code}.
\end{exam}

\vspace{0.1cm}\textbf{Convention}: To unify the notations, we also
denote $\beta^{(1)}_{1}=\alpha_1$ and $\beta^{(1)}_{2,
3}=\alpha_{2}+\alpha_{3}$. Similarly, we denote
$\beta^{(2)}_{2}=\alpha_2$, $\beta^{(2)}_{1,
3}=\alpha_{1}+\alpha_{3}$, $\beta^{(3)}_{3}=\alpha_3$ and
$\beta^{(3)}_{1, 2}=\alpha_{1}+\alpha_{2}$.

\vspace{0.1cm}\begin{defn}\label{R-closed-Pi} A partition
$\mathcal I=\{\Delta_1,\Delta_2,\Delta_3,\cdots,\Delta_K\}$ of
$\Pi$ is said to be R-\emph{closed} if $[S_{i_1}]\neq[S_{i_2}]$
and $\text{reg}([\Delta_j]_{i_1,i_2})=[\Delta_j]_{i_1,i_2}$ for
all $\{i_1,i_2\}\subseteq\{1,2,3\}$ and $j\in\{1,\cdots,K\}$.
\end{defn}

For example, for the region graph in Fig. \ref{fg-cmplx-1} (a),
the partition $\mathcal I_c$ in Example \ref{ex-dcn-code} is an
R-closed partition of $\Pi$.

Let $\{\mathcal B_{1},\mathcal B_{2},\mathcal
B_{3},\cdots,\mathcal B_{K}\}$ be constructed as in Lemma
\ref{genc-code} and $\mathcal
I=\{\Delta_1,\Delta_2,\Delta_3,\cdots,\Delta_K\}$ be an R-closed
partition of $\Pi$. By proper naming, we can let $\Delta_i=[S_i],
i=1,2,3$. Then all subclasses of $\mathcal I$ are in one-to-one
correspondence with all vectors in $\bigcup_{\ell=1}^K\mathcal
B_\ell$, where subclasses are defined as in Definition
\ref{sub-class} (See Fig. \ref{fg-partition}.). Let
$\tilde{C}_\Pi=\{d_R; R\in\Pi\}$ be constructed by assigning each
vector in $\bigcup_{\ell=1}^K\mathcal B_\ell$ to all regions in
the corresponding subclass. Specifically, let
\begin{itemize}
  \item $d_R=\beta^{(i)}_{i}$ for each $i\in\{1,2,3\}$
  and $R\in[S_i]_i$;
  \vspace{0.1cm}\item $d_R=\beta^{(i)}_{j_1,j_2}$
  for each $i\in\{1,2,3\}$ and $R\in[S_i]_{j_1,j_2}$,
  where $\{j_1,j_2\}=\{1,2,3\}\backslash\{i\}$;
  \vspace{0.1cm}\item $d_{R}=\beta^{(i)}_{j_1,j_2}$ for each
  $i\in\{4,\cdots,K\}$, each subset $\{j_1,j_2\}\subseteq\{1,2,3\}$
  and each $R\in[\Delta_{i}]_{j_1,j_2}$.
\end{itemize}

\begin{defn}\label{prtn-dctr-code-def}
The code $\tilde{C}_\Pi$ constructed as above is called an
$\mathcal I$-\emph{weak decentralized code} on $\Pi$.
\end{defn}

Fig. \ref{fg-cmplx-1} (b) illustrates an $\mathcal I_c$-weak
decentralized code for the region graph in Fig. \ref{fg-cmplx-1}
(a), where $\mathcal I_c$ is as in Example \ref{ex-dcn-code}.

To discuss the property of weak decentralized code on $\Pi$, we
need the following conception.
\begin{defn}[Independent Set]\label{ind-set-ptn}
Let $\mathcal I$ be an R-closed partition of $\Pi$. A subset
$\{Q,Q',Q''\}\subseteq\Pi$ is called an $\mathcal
I$-\emph{independent set} if the following three conditions hold:
\begin{itemize}
  \item [(1)] $|\{Q,Q',Q''\}\bigcap[[R]]|\leq 1$ for any equivalent
  class $[R]\in\mathcal I$ and any subclass $[[R]]$ of
  $[R]$, i.e., the three regions $Q,Q',Q''$ belongs to
  three different subclasses of $\mathcal I$;
  \item [(2)] $\{Q,Q',Q''\}\nsubseteq[R]$ for any equivalent
  class $R\in\mathcal I$;
  \item [(3)] $\{Q,Q',Q''\}\nsubseteq[S_{i_1}]_{i_1}\bigcup[S_{i_2}]_{i_2}
  \bigcup\text{reg}(S_{i_1}, S_{i_2})$ for any pair $\{i_1,i_2\}
  \subseteq\{1,2,3\}$.
\end{itemize}
\end{defn}

\begin{exam}
Consider the partition $\mathcal I_c=\{\Delta_1, \Delta_2,
\Delta_3$, $\Delta_4, \Delta_5\}$ in Example \ref{ex-dcn-code}. We
can check that $\{S_1, R_2, R_3\}$ is an $\mathcal
I_c$-independent set. We can also check that $\{R_1, R_2, R_4\}$
is not an $\mathcal I_c$-independent set because $|\{R_1, R_2,
R_4\}\cap[\Delta_4]_{1,2}|=|\{R_1, R_2\}|\geq 2$, violating
condition (1) of Definition \ref{ind-set-ptn}; $\{R_1, R_3, R_5\}$
is not an $\mathcal I_c$-independent set because $\{R_1, R_3,
R_5\}\subseteq\Delta_4$, violating condition (2) of Definition
\ref{ind-set-ptn}; $\{S_1, R_3, R_9\}$ is not an $\mathcal
I_c$-independent set because $\{S_1, R_3,
R_9\}\subseteq[S_{3}]_3\bigcup\text{reg}(S_{1}, S_{3})$, violating
condition (3) of Definition \ref{ind-set-ptn}.
\end{exam}

The following theorem is an important property of weak
decentralized code on $\Pi$.
\begin{thm}\label{prtn-dctr-code}
Let $\mathcal I=\{\Delta_1,\Delta_2,\Delta_3,\cdots,\Delta_K\}$ be
an R-closed partition of $\Pi$ and $\tilde{C}_\Pi=\{d_R\in\mathbb
F^3; R\in\Pi\}$ be an $\mathcal I$-weak decentralized code on
$\Pi$. The following hold:
\begin{itemize}
  \item [1)] $d_{S_i}=\alpha_i,i=1,2,3$.
  \item [2)] $d_R\in\langle d_{R'}; R'\in\text{In}(R)\rangle$ for
  all $R\in\Pi\backslash\{S_1,S_2,S_3\}$.
  \item [3)] If $Q$ and $Q'$ belong to two different subclasses of some
  $[R_j]\in\mathcal I$, then $\bar{\alpha}\in\langle d_{Q},d_{Q'}\rangle$.
  \item [4)] If $\{Q,Q',Q''\}$ is an $\mathcal I$-independent
  set, then $d_{Q},d_{Q'},d_{Q''}$ are linearly independent.
  Hence, $\bar{\alpha}\in\langle d_{Q},d_{Q'},d_{Q''}\rangle$.
\end{itemize}
\end{thm}
\begin{proof}
1), 3) are obvious by the construction of $\tilde{C}_\Pi$.

To prove 2), we fix $\{i_1,i_2\}\subseteq\{1,2,3\}$. For each
$j\in\{1, \cdots, K\}$, let $\Delta'_j=[\Delta_j]_{i_1,i_2}$.
Since $\mathcal I$ is an R-closed partition of $\Pi$, then
$\mathcal I'=\{\Delta'_1,\Delta'_2,\cdots,\Delta'_K\}$ is an
R-closed partition of $\text{reg}(S_{i_1},S_{i_2})$. By the
construction of $\tilde{C}_\Pi$, the subcode
$\tilde{C}_{S_{i_1},S_{i_2}}=\{d_R;
R\in\text{reg}(S_{i_1},S_{i_2})\}$ is an $\mathcal I'$-weak
decentralized code on $\text{reg}(S_{i_1},S_{i_2})$. By Lemma
\ref{lem-weak-dec-code}, $d_R\in\langle
d_{R'};R'\in\text{In}(R)\rangle$ for all
$R\in\text{reg}^\circ(S_{i_1},S_{i_2})$. Note that $\{i_1,i_2\}$
can be arbitrary subset of $\{1,2,3\}$. Then we have
$d_R\in\langle d_{R'}; R'\in\text{In}(R)\rangle$ for all
$R\in\text{reg}^\circ(S_{1},S_{2})\cup\text{reg}^\circ(S_{1},S_{3})
\cup\text{reg}^\circ(S_{2},S_{3})=\Pi\setminus\{S_1,S_2,S_3\}$.

4) By Definition \ref{prtn-dctr-code-def} and \ref{ind-set-ptn},
the three vectors $d_{Q}, d_{Q'}$ and $d_{Q''}$ satisfy condition
(4) of Lemma \ref{genc-code}, hence are linearly independent. So
$\bar{\alpha}\in\mathbb F^3=\langle d_{Q},d_{Q'},d_{Q''}\rangle$.
\end{proof}

Note that in the construction of $\mathcal I$-weak decentralized
code, $\Lambda_1,\Lambda_2,\cdots,\Lambda_n$ are not taken into
consideration. So in general, an $\mathcal I$-weak decentralized
code is not necessarily a feasible code on $\Pi$. In the next
section, we will construct a partition $\mathcal I_c$ of $\Pi$,
called a normal partition of $\Pi$, such that $\text{RG}(D^{**})$
is feasible if and only if the $\mathcal I_c$-weak decentralized
code on $\Pi$ is a feasible code on $\Pi$.

\section{Network Coding for $3$-source $n$-terminal
Terminal-separable Sum-network}

In this section, we always assume $\text{RG}(D^{**})$ is a
terminal-separable region graph with $3$ source regions and $n$
terminal regions, where $n$ is an positive integer. We will
characterize the feasibility of $\text{RG}(D^{**})$ and show that
it can be determined in polynomial time. Recall that by Theorem
\ref{lmd-solv}, we just need to determine the existence of a
feasible code on $\Pi$.

\subsection{Some Examples}
In this subsection, we use three examples to show the basic idea
of determining feasibility of $\text{RG}(D^{**})$.

\begin{exam}\label{ex-chr-ptn-1}
Let $\text{RG}(D^{**})$ be the region graph in Fig.
\ref{fg-cmplx-1} (a). In Example \ref{ex-dcn-code}, we have seen
that $\Pi=\{S_1,S_2,S_3$, $R_1,\cdots,R_9\}$. By (2) of Definition
\ref{lmd-omd}, $\Omega_j=\{T_j\}$ for $j\in\{1,\cdots,6\}$ and
$\Omega_I=\emptyset$ for all $I\subseteq\{1,\cdots,6\}$ of size
$|I|\geq 2$. So $\text{RG}(D^{**})$ is terminal-separable.
Moreover, by (3) of Definition \ref{lmd-omd}, $\Lambda_1=\{R_2,
R_5\}$, $\Lambda_2=\{R_1, R_5\}$, $\Lambda_3=\{R_7, R_8\}$,
$\Lambda_4=\{S_3, R_7\}$, $\Lambda_5=\{R_3, R_6\}$,
$\Lambda_6=\{R_7, R_9\}$.
\end{exam}

Let $\tilde{C}_\Pi=\{d_R\in\mathbb F^3; R\in\Pi\}$ be an arbitrary
feasible code on $\Pi$. By Remark \ref{rem-sgl-sets}, we assume
$d_R\neq 0$ for all $R\in\Pi$. The following process yields a
partition of $\Pi$.

First, we consider $S_3$. Since $\{S_3, R_7\}=\Lambda_4$, then by
2) of Lemma \ref{lem-sgl-sets}, $\langle d_{R_7},
\bar{\alpha}\rangle =\langle d_{S_3}, \bar{\alpha}\rangle.$
Similarly, $\{R_7, R_8\}=\Lambda_3$ and $\{R_7, R_9\}=\Lambda_6$
imply that $\langle d_{R_8}, \bar{\alpha}\rangle=\langle d_{R_7},
\bar{\alpha}\rangle=\langle d_{R_9}, \bar{\alpha}\rangle.$ So we
have $$\langle d_{S_3}, \bar{\alpha}\rangle=\langle d_{R_7},
\bar{\alpha}\rangle=\langle d_{R_8}, \bar{\alpha}\rangle=\langle
d_{R_9}, \bar{\alpha}\rangle.$$ For the coding vectors
$d_{S_3},d_{R_7},d_{R_8}$ and $d_{R_9}$, note that if we know one
of them, then we can use the above relation to obtain all of them.
As an example, let us see how $d_{S_3}$ determine $d_{R_7}$: Since
$R_7\in\text{reg}(S_1,S_2)$, then by 1) of Lemma
\ref{lem-sgl-sets}, $d_{R_7}\in\langle\alpha_{1},
\alpha_{2}\rangle$. So $\langle d_{R_7}\rangle=\langle d_{S_3},
\bar{\alpha}\rangle \cap\langle\alpha_1, \alpha_2\rangle=\langle
\alpha_1+\alpha_2\rangle.$

Second, consider $R_1$. Since $\{R_1, R_5\}=\Lambda_2$ and $\{R_2,
R_5\}=\Lambda_1$, then by 2) of Lemma \ref{lem-sgl-sets}, we have
$$\langle d_{R_1}, \bar{\alpha}\rangle=\langle d_{R_5},
\bar{\alpha}\rangle=\langle d_{R_2}, \bar{\alpha}\rangle.$$ Note
that $\{R_1,R_2\}\subseteq\text{reg}(S_1,S_2)$. Then by 1) of
Lemma \ref{lem-sgl-sets}, $d_{R_1}, d_{R_2}\in\langle\alpha_1,
\alpha_2\rangle.$ So $d_{R_2}\in\langle d_{R_1},
\bar{\alpha}\rangle\cap\langle\alpha_1, \alpha_2\rangle=\langle
d_{R_1}\rangle$, which implies that $\langle
d_{R_1}\rangle=\langle d_{R_2}\rangle.$ Moreover, by Fig.
\ref{fg-cmplx-1} (a), $R_6\in\text{reg}(R_1, R_2)$. Then by 1) of
Lemma \ref{lem-sgl-sets}, $\langle d_{R_6}\rangle=\langle
d_{R_2}\rangle=\langle d_{R_1}\rangle.$ So $\langle d_{R_6},
\bar{\alpha}\rangle=\langle d_{R_1}, \bar{\alpha}\rangle=\langle
d_{R_5}, \bar{\alpha}\rangle=\langle d_{R_2},
\bar{\alpha}\rangle.$ Again by the condition $\{R_3,
R_6\}=\Lambda_5$ and 2) of Lemma \ref{lem-sgl-sets}, we have
$\langle d_{R_3}, \bar{\alpha}\rangle=\langle d_{R_6},
\bar{\alpha}\rangle.$ Thus,
$$\langle d_{R_3},
\bar{\alpha}\rangle=\langle d_{R_6}, \bar{\alpha}\rangle=\langle
d_{R_1}, \bar{\alpha}\rangle=\langle d_{R_5},
\bar{\alpha}\rangle=\langle d_{R_2}, \bar{\alpha}\rangle.$$ %So
%$d_{R_1}, d_{R_2}, d_{R_3}, d_{R_5}$ and $d_{R_6}$ are determined
%by each other.

Let $\Delta_1=\{S_1\},\Delta_2=\{S_2\}$, $\Delta_3=\{S_3, R_7,
R_8, R_9\}$, $\Delta_4=\{R_1, R_2, R_3, R_5, R_6\}$,
$\Delta_5=\{R_4\}$ and $\mathcal I_c=\{\Delta_1, \Delta_2,
\Delta_3, \Delta_4, \Delta_5\}$. (See Fig. \ref{fg-cmplx-1} (a).)
Then for any $\Delta_i$, if we know the coding vector of one
region in $\Delta_i$, then we can obtain the coding vectors of all
other regions in $\Delta_i$.

Let $\tilde{C}_\Pi$ be the $\mathcal I_c$-weak decentralized code
constructed in Fig. \ref{fg-cmplx-1} (b). It is easy to check that
$\tilde{C}_\Pi$ is a feasible code on $\Pi$. So
$\text{RG}(D^{**})$ is feasible.

%%%%%%%%%%%%%%%%%%%%%%%%%%%%%%%%%%%%%%%%%%%
\renewcommand\figurename{Fig}
\begin{figure}[htbp]
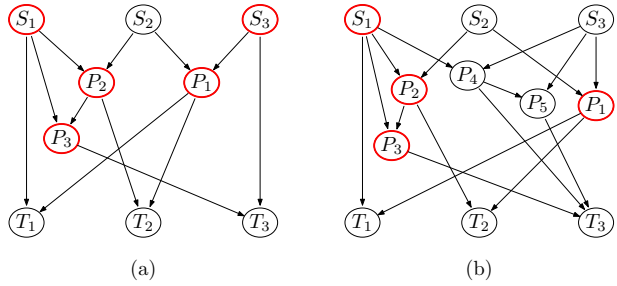

\begin{center}
\vspace{0.2cm}\includegraphics[height=3.7cm]{infsb.3}
\hspace{0.7cm}\includegraphics[height=3.7cm]{infsb.2}
\end{center}
\caption{Two examples of infeasible region graph.
}\label{fg-infsb}
\end{figure}
%%%%%%%%%%%%%%%%%%%%%%%%%%%%%%%%%%%%%%%%%%%%%%

\begin{exam}\label{ex-chr-ptn-2}
Let $\text{RG}(D^{**})$ be the region graph in Fig. \ref{fg-infsb}
(a). By Definition \ref{g-reg},
$\text{reg}(S_1,S_2)=\{S_1,S_2,P_2,P_3\}$,
$\text{reg}(S_1,S_3)=\{S_1,S_3\}$ and
$\text{reg}(S_2,S_3)=\{S_2,S_3,P_1\}$. By (2) of Definition
\ref{lmd-omd}, $\Omega_j=\{T_j\}$ for $j\in\{1,2,3\}$ and
$\Omega_I=\emptyset$ for all $I\subseteq\{1,2,3\}$ of size
$|I|\geq 2$. So $\text{RG}(D^{**})$ is terminal-separable. By (3)
of Definition \ref{lmd-omd}, $\Lambda_1=\{S_1, P_1\}$,
$\Lambda_2=\{P_1, P_2\}$ and $\Lambda_3=\{S_3, P_3\}$.
\end{exam}

Let $\tilde{C}_\Pi=\{d_R\in\mathbb F^3; R\in\Pi\}$ be an arbitrary
feasible code on $\Pi$.

Consider $S_1$. Since $\{S_1, P_1\}=\Lambda_1$ and $\{P_1,
P_2\}=\Lambda_2$, then by 2) of Lemma \ref{lem-sgl-sets}, we have
$$\langle d_{P_2}, \bar{\alpha}\rangle=\langle d_{P_1},
\bar{\alpha}\rangle=\langle d_{S_1}, \bar{\alpha}\rangle.$$ So
$d_{P_2}\in\langle d_{S_1}, \bar{\alpha}\rangle.$ Note that
$P_2\in\text{reg}(S_1,S_2)$. Then by 1) of Lemma
\ref{lem-sgl-sets}, we have
$d_{P_2}\in\langle\alpha_1,\alpha_2\rangle.$ Thus, $$\langle
d_{P_2}\rangle=\langle d_{S_1},
\bar{\alpha}\rangle\cap\langle\alpha_1,\alpha_2\rangle=\langle
d_{S_1}\rangle.$$ Moreover, since $P_3\in\text{reg}(S_1,P_2)$ (See
Fig. \ref{fg-infsb} (a).), then by 1) of Lemma \ref{lem-sgl-sets},
$\langle d_{P_3}\rangle=\langle d_{P_2}\rangle=\langle
d_{S_1}\rangle$, which implies that
$$\langle d_{P_3}, \bar{\alpha}\rangle=\langle d_{P_2},
\bar{\alpha}\rangle=\langle d_{P_1}, \bar{\alpha}\rangle=\langle
d_{S_1}, \bar{\alpha}\rangle.$$ Again by $\{S_3, P_3\}=\Lambda_3$
and 2) of Lemma \ref{lem-sgl-sets}, we have
\begin{align}\label{ex-chr-ptn-2-eq-1}
\langle d_{S_3}, \bar{\alpha}\rangle=\langle d_{P_3},
\bar{\alpha}\rangle=\langle d_{P_2}, \bar{\alpha}\rangle=\langle
d_{P_1}, \bar{\alpha}\rangle=\langle d_{S_1}, \bar{\alpha}\rangle.
\end{align} So we obtain a subset $[S_1]=\{S_1,P_1,P_2,P_3,S_3\}$.
Let $\mathcal I_c=\{[S_1],[S_2]\}$, where $[S_2]=\{S_2\}$. Note
that $[S_1]=[S_3]$.

We remark that (\ref{ex-chr-ptn-2-eq-1}) can not be satisfied
because we can check that $d_{S_1}=\alpha_1\notin\langle \alpha_3,
\bar{\alpha}\rangle=\langle d_{S_3}, \bar{\alpha}\rangle$. So
$\langle d_{S_3}, \bar{\alpha}\rangle\neq\langle d_{S_1},
\bar{\alpha}\rangle$. Thus, $\tilde{C}_\Pi$ can not be well
constructed. By Theorem \ref{lmd-solv}, $\text{RG}(D^{**})$ is
infeasible.

\begin{exam}\label{ex-chr-ptn-3}
Let $\text{RG}(D^{**})$ be the region graph in Fig. \ref{fg-infsb}
(b). By Definition \ref{g-reg},
$\text{reg}(S_1,S_2)=\{S_1,S_2,P_2,P_3\}$,
$\text{reg}(S_1,S_3)=\{S_1,S_3,P_4,P_5\}$ and
$\text{reg}(S_2,S_3)=\{S_2,S_3$, $P_1\}$. By (2) of Definition
\ref{lmd-omd}, $\Omega_j=\{T_j\}$ for $j\in\{1,2,3\}$ and
$\Omega_I=\emptyset$ for all $I\subseteq\{1,2,3\}$ of size
$|I|\geq 2$. So $\text{RG}(D^{**})$ is terminal-separable.
Moreover, by (3) of Definition \ref{lmd-omd}, $\Lambda_1=\{S_1,
P_1\}$, $\Lambda_2=\{P_1, P_2\}$ and $\Lambda_3=\{P_3, P_4,
P_5\}$.
\end{exam}

Let $\tilde{C}_\Pi=\{d_R\in\mathbb F^3; R\in\Pi\}$ be an arbitrary
feasible code on $\Pi$. Consider $S_1\in\Pi$. By the same
discussion as in Example \ref{ex-chr-ptn-2}, we can find a subset
$[S_1]=\{S_1,P_1,P_2,P_3\}$ such that
\begin{align*}%\label{ex-chr-ptn-3-eq-1}
\langle d_{P_3}, \bar{\alpha}\rangle=\langle
d_{P_2}, \bar{\alpha}\rangle=\langle d_{P_1},
\bar{\alpha}\rangle=\langle d_{S_1}, \bar{\alpha}\rangle.
\end{align*} Let $\mathcal
I_c=\{[S_1],[S_2],[S_3],[P_4],[P_5]\}$, where $[S_2]=\{S_2\}$,
$[S_3]=\{S_3\}$, $[P_4]=\{P_4\}$ and $[P_5]=\{P_5\}$. Note that
$[S_1]_1=\{S_1,P_2,P_3\}$ and $\Lambda_3=\{P_3, P_4,
P_5\}\subseteq[S_1]_1\bigcup\text{reg}(S_{1},S_{3})$.

We now derive a contradiction. Firstly, as in Example
\ref{ex-chr-ptn-2}, we have $\langle d_{P_3}\rangle=\langle
d_{S_1}\rangle.$ Secondly, since $P_4, P_5\in\text{reg}(S_1,
S_3)$, by 1) of Lemma \ref{lem-sgl-sets}, we have $d_{P_4},
d_{P_5}\in\langle \alpha_1, \alpha_3\rangle$. Then by condition
(3) of Definition \ref{sgl-sets},
$\bar{\alpha}=\alpha_1+\alpha_2+\alpha_3\in\langle d_R;
R\in\Lambda_3\rangle=\langle
d_{P_3},d_{P_4},d_{P_5}\rangle\subseteq\langle \alpha_1,
\alpha_3\rangle$, a contradiction.

Thus, there exists no feasible code on $\Pi$ and we conclude that
$\text{RG}(D^{**})$ is infeasible.

In general, if $\text{RG}(D^{**})$ is terminal-separable, then we
can always find a partition $\mathcal I_c$ of $\Pi$ such that : If
$[S_i]=[S_j]$ for some $\{i,j\}\subseteq\{1,2,3\}$, then
$\text{RG}(D^{**})$ is infeasible (Example \ref{ex-chr-ptn-2}); If
$\Lambda_j\subseteq[S_{i_1}]_{i_1}\bigcup[S_{i_2}]_{i_2}
\bigcup\text{reg}(S_{i_1}, S_{i_2})$ for some $j\in[n]$ and
$\{i_1,i_2\}\subseteq\{1,2,3\}$, then $\text{RG}(D^{**})$ is
infeasible (Example \ref{ex-chr-ptn-3}); Otherwise, the $\mathcal
I_c$-weak decentralized code on $\Pi$ is a feasible code on $\Pi$
and $\text{RG}(D^{**})$ is feasible (Example \ref{ex-chr-ptn-1}).
---- In this sense, we call $\mathcal I_c$ a normal partition of $\Pi$.

\subsection{Determining Feasibility of $\text{RG}(D^{**})$}
In this subsection we formally describe our method, which is a
generalization of the idea of last subsection.

Firstly, we need a definitions.
\begin{defn}\label{cntd}
Let $\mathcal I=\{[S_1],[S_2],[S_3],\cdots,[R_K]\}$ be a partition
of $\Pi$ and $\{[R'],[R'']\}\subseteq\mathcal I$. We say that
$[R']$ and $[R'']$ are \emph{connected} if one of the following
conditions hold:
\begin{itemize}
  \item [(1)] There is a subclass $[[R']]$ of $[R']$ and a subclass
  $[[R'']]$ of $[R'']$ such that $\Lambda_j\subseteq[[R']]\cup[[R'']]$
  for some $j\in[n]$;
  \item [(2)] There is a subset $\{i_1,i_2\}\subseteq\{1,2,3\}$ such
  that $\text{reg}([R']_{i_1,i_2})\cap\text{reg}([R'']_{i_1,i_2})
  \neq\emptyset$.
\end{itemize}
\end{defn}

Suppose $\mathcal I$ is a partition of $\Pi$ and $\{[R'],
[R'']\}\subseteq\mathcal I$. By combining $[R']$ and $[R'']$ into
one equivalent class $[R']\cup[R'']$, we obtain a partition
$\mathcal I'=\mathcal I\cup\{[R']\cup[R'']\}\setminus\{[R'],
[R'']\}$ of $\Pi$. We call $\mathcal I'$ a \emph{contraction} of
$\mathcal I$ by combining $[R']$ and $[R'']$.

\begin{exam}\label{eg-chpt}
Consider again Example \ref{ex-chr-ptn-3}. We can check that
$\mathcal I_c$ is obtained from $\mathcal I_0$ by a series of
contraction by combining two connected equivalent classes, where
$\mathcal I_0=\{[R];R\in\Pi\}$ is the trivial partition of $\Pi$.
In fact, by condition (1) of Definition \ref{cntd},
$[S_1]=\{S_1\}$ and $[P_1]=\{P_1\}$ are connected in $\mathcal
I_0$. So by combining $[S_1]$ and $[P_1]$ into
$[S_1]=\{S_1,P_1\}$, we obtain a partition $\mathcal I_1$ of
$\Pi$. Similarly, $[S_1]=\{S_1,P_1\}$ and $[P_2]=\{P_2\}$ are
connected in $\mathcal I_1$ and we can obtain a partition
$\mathcal I_2$ of $\Pi$ by combining $[S_1]$ and $[P_2]$ into
$[S_1]=\{S_1,P_1,P_2\}$. Observing Fig. \ref{fg-infsb} (b), we
have $P_3\in\text{reg}(S_1,P_2)=\text{reg}([S_1]_{1,2})$. So by
condition (2) of Definition \ref{cntd}, $[S_1]=\{S_1,P_1,P_2\}$
and $[P_3]=\{P_3\}$ are connected in $\mathcal I_2$ and we can
obtain a partition $\mathcal I_3$ of $\Pi$ by combining
$[S_1]=\{S_1,P_1,P_2\}$ and $[P_3]=\{P_3\}$ into
$[S_1]=\{S_1,P_1,P_2,P_3\}$. We can check that no pair of
equivalent classes of $\mathcal I_3$ are connected.
\end{exam}

The above example also gives an illustration of constructing
normal partition of $\Pi$. Let $\mathcal I_0=\{[R]; R\in\Pi\}$ be
the trivial partition of $\Pi$. A normal partition of $\Pi$ can be
obtained from $\mathcal I_0$ by a series of contraction.
Specifically, we have the following definition.

\begin{defn}[Normal Partition]\label{chpt}
Let $\mathcal I_0,\mathcal I_1,\cdots,\mathcal I_L=\mathcal I_c$
be a sequence of partitions of $\Pi$ such that for each
$\ell\in\{1,\cdots,L\}$ and $\{i,j\}\subseteq\{1,2,3\}$,
$[S_i]\neq[S_j]$ in $\mathcal I_{\ell-1}$ and $\mathcal I_\ell$ is
a contraction of $\mathcal I_{\ell-1}$ by combining two connected
equivalent classes. The partition $\mathcal I_c$ is called a
\emph{normal partition} of $\Pi$ if one of the following
conditions hold:
\begin{itemize}
  \item [(1)] $[S_i]=[S_j]$ for some $\{i,j\}\subseteq\{1,2,3\}$;
  \item [(2)] No pair of equivalent classes in $\mathcal I_c$ are
  connected.
\end{itemize}
\end{defn}

By previous discussion, the partition $\mathcal I_c$ in Example
\ref{ex-chr-ptn-3} is a normal partition of $\Pi$. We can also
check that the partition $\mathcal I_c$ in Example
\ref{ex-chr-ptn-1} is a normal partition of $\Pi$.

Consider the partition $\mathcal I_c=\{[S_1],[S_2]\}$ in Example
\ref{ex-chr-ptn-2}, where $[S_1]=\{S_1,P_1,P_2,P_3,S_3\}$ and
$[S_2]=\{S_2\}$. It is easy to check that $\mathcal I_c$ is
obtained from $\mathcal I_0$ by a series of contraction by
combining two connected equivalent classes. In this example, we
have $[S_1]=\{S_1,P_1,P_2,P_3,S_3\}=[S_3]$. So $\mathcal I_c$ is a
normal partition of $\Pi$.

As we have seen in the last subsection, for a feasible code on
$\Pi$, the coding vectors of regions in two connected equivalent
classes are determined by each other. Thus, we can combine such
two equivalent classes together. Specifically, we have the
following lemma.
\begin{lem}\label{chpt-lem}
Let $ \tilde{C}_\Pi=\{d_R; R\in\Pi\}$ be a feasible code on $\Pi$
and $\mathcal I_c$ be a normal partition of $\Pi$. Then
$d_{Q'}\in\langle d_{Q}\rangle$ for any $[R]\in\mathcal I_c$, any
$\{i_1,i_2\}\subseteq\{1,2,3\}$ and any
$\{Q,Q'\}\subseteq[R]_{i_1,i_2}$.
\end{lem}
\begin{proof}
The proof is given in Appendix B.
\end{proof}

By Definition \ref{chpt}, it is easy to check that the following
Algorithm 3 outputs a normal partition of $\Pi$.

\vspace{0.2cm}
\begin{center}
\setlength{\unitlength}{1mm}
\begin{picture}(110,64)(-2,-62)
\put(-1,4){\line(1,0){107}} \put(-1,4){\line(0,-1){66}}

\put(106,-62){\line(-1,0){107}} \put(106,-62){\line(0,1){66}}

\put(0,0){\textbf{Algorithm 3}: Partitioning algorithm $(\Pi)$:}

\put(3,-5){$L=0$;}

\put(3,-10){\textbf{While} there are $R',R''\in\mathcal I_L$ which
are connected \textbf{do}}

\put(6,-15){Let $~\mathcal I_{L+1}~$ be a contraction of
$~\mathcal I_{L}~$ by combining $R'$}

\put(6,-20){and $R''$;}

\put(6,-25){\textbf{If} $[S_i]=[S_j]$ for some
$\{i,j\}\subseteq\{1,2,3\}$ \textbf{then}}

\put(9,-30){$\mathcal I_c=\mathcal I_L$;}

\put(9,-35){\textbf{return} $\mathcal I_c$;}

\put(9,-40){\textbf{stop};}

\put(6,-45){\textbf{else}}

\put(9,-50){$L=L+1$;}

\put(3,-55){$\mathcal I_c=\mathcal I_L$;}

\put(3,-60){\textbf{return} $\mathcal I_c$;}

\end{picture}
\end{center}

Clearly, the While-loop of Algorithm 3 has at most $|\mathcal
I_0|=|\Pi|$ rounds. In each round, we need to determine wether
there are two connected equivalent classes, which can be done in
time $O(n)$ by Definition \ref{cntd}. So Algorithm 3 can output
$\mathcal I_c$ in $\{|\Pi|,n\}$-polynomial time.

%%%%%%%%%%%%%%%%%%%%%%%%%%%%%%%%%%%%%%%%%%%
\renewcommand\figurename{Fig}
\begin{figure}[htbp]
\begin{center}
\includegraphics[height=3.6cm]{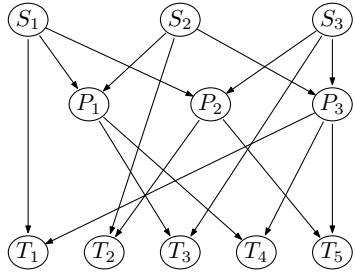}
\end{center}
\caption{An example of infeasible region graph.}\label{fg-m-chtp}
\end{figure}
%%%%%%%%%%%%%%%%%%%%%%%%%%%%%%%%%%%%%%%%%%%%%%

We remark that for a given region graph, there could be several
normal partitions of $\Pi$. Consider the region graph in Fig.
\ref{fg-m-chtp}. Let $\mathcal I_c=\{[S_1],[S_3],[P_1]\}$ such
that $[S_1]=\{S_1,P_3$, $P_2,S_2\},[S_3]=\{S_3\},[P_1]=\{P_1\}$
and $\mathcal I_c'=\{[S_1],[S_2],[P_2]\}$ such that
$[S_1]=\{S_1,P_3,P_1,S_3\},[S_2]=\{S_2\},[P_2]=\{P_2\}$. Then
$\mathcal I_c$ and $\mathcal I_c'$ are both normal partitions of
$\Pi$. However, we will show that to determine whether
$\text{RG}(D^{**})$ is feasible, it is sufficient to construct one
normal partition of $\Pi$ (as Algorithm 3 does).

\begin{defn}[Compatibility]\label{cmptl}
Suppose $\mathcal I=\{[R_1],[R_2]$, $[R_3],\cdots,[R_K]\}$ is a
partition of $\Pi$. We say that $\mathcal I$ is \emph{compatible}
if the following three conditions hold:
\begin{itemize}
  \item [(1)] $[S_i]\neq[S_j]$ for all pair $\{i,j\}\subseteq\{1,2,3\}$;
  \item [(2)] No pair of equivalent classes of $\mathcal I$
  are connected;
  \item [(3)] $\Lambda_j\nsubseteq[S_{i_1}]_{i_1}\bigcup[S_{i_2}]_{i_2}
  \bigcup\text{reg}(S_{i_1},S_{i_2})$
  for all $j\in[n]$ and $\{i_1,i_2\} \subseteq\{1,2,3\}$.
\end{itemize}
\end{defn}

The following theorem is one of the main results of this section.
\begin{thm}\label{regular-cmptl}
Let $\mathcal I_c$ be a normal partition of $\Pi$. Then
$\text{RG}(D^{**})$ is feasible if and only if $\mathcal I_c$ is
compatible. Moreover, it is $\{|\Pi|,n\}$-polynomial time
complexity to determine whether $\text{RG}(D^{**})$ is feasible.
\end{thm}

By Definition \ref{chpt} and Theorem \ref{regular-cmptl}, to
determine feasibility of $\text{RG}(D^{**})$, we start with the
trivial partition $\mathcal I_0$ of $\Pi$ and combine the pair of
connected equivalent classes step by step. If $[S_i]=[S_j]$ for
some $\{i,j\}\subseteq\{1,2,3\}$ in some step, then
$\text{RG}(D^{**})$ is infeasible. Else, by at most $|\Pi|$ steps,
we can obtain a partition $\mathcal I_c$ of $\Pi$ satisfying
conditions (1), (2) of Definition \ref{cmptl}. If $\mathcal I_c$
further satisfies condition (3) of Definition \ref{cmptl}, then
$\text{RG}(D^{**})$ is feasible. Otherwise, $\text{RG}(D^{**})$ is
infeasible.

Before proving Theorem \ref{regular-cmptl}, we first give some
lemmas.

\begin{lem}\label{cmpt-wd-code}
If $\mathcal I$ is a compatible partition of $\Pi$, then we can
construct an $\mathcal I$-weak decentralized code on $\Pi$.
\end{lem}
\begin{proof}
Let $\mathcal I=\{[R_1],[R_2],[R_3],\cdots,[R_K]\}$. Then we have
$[R_j]_{i_1,i_2}=\text{reg}([R_j]_{i_1,i_2})$ for each $j\in\{1,
2, \cdots, K\}$ and $\{i_1,i_2\}\subseteq\{1,2,3\}$. Otherwise,
there exists an
$R'\in\text{reg}([R_\ell]_{i_1,i_2})\backslash[R_\ell]_{i_1,i_2}$.
By condition (2) of Definition \ref{cntd}, $[R_\ell]$ and $[R']$
are connected, which contradicts to condition (2) of Definition
\ref{cmptl}. Thus, $\mathcal I$ is R-closed and we can construct
an $\mathcal I$-weak decentralized code on $\Pi$.
\end{proof}

\begin{lem}\label{comp-code}
If $\mathcal I$ is a compatible partition of $\Pi$, then the
$\mathcal I$-weak decentralized code is a feasible code on $\Pi$
and $\text{RG}(D^{**})$ is feasible.
\end{lem}
\begin{proof}
Let $\tilde{C}_\Pi=\{d_R; R\in\Pi\}$ be an $\mathcal I$-weak
decentralized code on $\Pi$. We need to prove that $\tilde{C}_\Pi$
is a feasible code on $\Pi$. Note that 1), 2) of Theorem
\ref{prtn-dctr-code} corresponds to (1), (2) of Definition
\ref{sgl-sets} respectively. So we only need to prove that
$\tilde{C}_\Pi$ satisfies condition (3) of Definition
\ref{sgl-sets}, i.e., $\bar{\alpha}\in\langle d_R;
R\in\Lambda_j\rangle$ for all $j\in[n]$. We have the following two
cases:

Case 1: There is an $[R_\ell]\in\mathcal I_c$ such that
$\Lambda_j$ intersects with at least two different subclasses of
$[R_\ell]$. Suppose $\{Q,Q'\}\subseteq\Lambda_j$ and $Q,Q'$ belong
to two different subclasses of $[R_\ell]$. By 3) of Theorem
\ref{prtn-dctr-code}, $\bar{\alpha}\in\langle
d_{Q},d_{Q'}\rangle\subseteq\langle d_R; R\in\Lambda_j\rangle$.

Case 2: For each $[R_\ell]\in\mathcal I$, $\Lambda_j$ intersects
with at most one subclass of $[R_\ell]$. Then by condition (2) of
Definition \ref{cmptl} and condition (1) of Definition \ref{cntd},
$\Lambda_j$ intersects with at least three equivalent classes in
$\mathcal I$. Moreover, by condition (3) of  Definition
\ref{cmptl},
$\Lambda_j\nsubseteq[S_{i_1}]_{i_1}\bigcup[S_{i_2}]_{i_2}
\bigcup\text{reg}(S_{i_1},S_{i_2})$ for all $\{i_1,i_2\}
\subseteq\{1,2,3\}$. Then by Definition \ref{ind-set-ptn}, we can
find a subset $\{Q,Q',Q''\}\subseteq \Lambda_j$ such that
$\{Q,Q',Q''\}$ is an $\mathcal I$-independent set. By 4) of
Theorem \ref{prtn-dctr-code}, $\bar{\alpha}\in\langle
d_{Q},d_{Q'},d_{Q''}\rangle\subseteq\langle d_R;
R\in\Lambda_j\rangle$.

By the above discussion, $\tilde{C}_\Pi$ satisfies conditions
(1)-(3) of Definition \ref{sgl-sets}. So $\tilde{C}_\Pi$ is a
feasible code on $\Pi$.
\end{proof}

Now we can prove Theorem \ref{regular-cmptl}.
\begin{proof}[Proof of Theorem \ref{regular-cmptl}]
If $\mathcal I_c$ is compatible, then by Lemma \ref{comp-code},
$\text{RG}(D^{**})$ is feasible.

Conversely, suppose $\text{RG}(D^{**})$ is feasible. By Theorem
\ref{lmd-solv}, there is a feasible code
$\tilde{C}_\Pi=\{d_R\in\mathbb F^3; R\in\Pi\}$ on $\Pi$. We shall
prove $\mathcal I_c$ satisfies conditions (1)-(3) of Definition
\ref{cmptl}.

For $\{i_1,i_2\}\subseteq\{1,2,3\}$, if $[S_{i_1}]=[S_{i_2}]$,
then $S_{i_1}\in[S_{i_2}]$. By Definition \ref{lnc-reg-g} and
Lemma \ref{chpt-lem}, $d_{S_{i_1}}=\alpha_{i_1}\in\langle
d_{S_{i_2}}\rangle=\langle\alpha_{i_2}\rangle$, a contradiction.
So $[S_{i_1}]\neq[S_{i_2}]$ for all
$\{i_1,i_2\}\subseteq\{1,2,3\}$. Moreover, by Definition
\ref{chpt}, no pair of equivalent classes of $\mathcal I_c$ are
connected. Thus, $\mathcal I_c$ satisfies conditions (1), (2) of
Definition \ref{cmptl}.

We can prove condition (3) by contradiction. Suppose
$\Lambda_j\subseteq[S_{i_1}]_{i_1}\bigcup[S_{i_2}]_{i_2}
\bigcup\text{reg}(S_{i_1},S_{i_2})$ for some $j\in[n]$ and
$\{i_1,i_2\}\subseteq\{1,2,3\}$. By Lemma \ref{chpt-lem} and
(\ref{eq-nt-4}), we have
$$d_{Q}\in\langle d_{S_{i_1}}\rangle=\langle\alpha_{i_1}\rangle, \forall
Q\in[S_{i_1}]_{i_1}$$ and
$$d_{Q}\in\langle d_{S_{i_2}}\rangle=\langle\alpha_{i_2}\rangle, \forall
Q\in[S_{i_2}]_{i_2}.$$ Moreover, since $\tilde{C}_\Pi$ is a
feasible code on $\Pi$, by 1) of Lemma \ref{lem-sgl-sets}, we have
$d_Q\in\langle\alpha_{i_1},\alpha_{i_2}\rangle,\forall
Q\in\text{reg}(S_{i_1},S_{i_2})$. Since we assume
$\Lambda_j\subseteq[S_{i_1}]_{i_1}\bigcup[S_{i_2}]_{i_2}
\bigcup\text{reg}(S_{i_1},S_{i_2})$, then
$d_Q\in\langle\alpha_{i_1},\alpha_{i_2}\rangle$ for all
$Q\in\Lambda_j$ and by condition (3) of Definition \ref{sgl-sets},
$\bar{\alpha}=\alpha_1+\alpha_2+\alpha_3\in\langle d_Q;
Q\in\Lambda_j\rangle=\langle\alpha_{i_1},\alpha_{i_2}\rangle$, a
contradiction. So
$\Lambda_j\nsubseteq[S_{i_1}]_{i_1}\bigcup[S_{i_2}]_{i_2}
\bigcup\text{reg}(S_{i_1},S_{i_2})$ and $\mathcal I_c$ satisfies
condition (3) of Definition \ref{cmptl}.

By the above discussion, $\mathcal I_c$ is compatible. Finally, we
prove the time complexity. We have seen that Algorithm 3 can
output $\mathcal I_c$ in $\{|\Pi|,n\}$-polynomial time. Moreover,
it is easy to see that conditions (1)-(3) of Definition
\ref{cmptl} can be checked in time $O(n)$. Thus, it is
$\{|\Pi|,n\}$-polynomial time complexity to determine whether of
$\text{RG}(D^{**})$ is feasible.
\end{proof}

\subsection{Some Simple Cases of Three-source Sum-network}
The following theorem gives some families of terminal-separable
region graph that is feasible.
\begin{thm}\label{simple-case}
Suppose $\text{RG}(D^{**})$ is terminal-separable. Then it is
feasible if one of the following conditions hold:
\begin{itemize}
  \item [(1)] $\Lambda_{j_1}\cap\Lambda_{j_2}=\emptyset$ for all
  pair $\{j_1,j_2\}\subseteq[n]$ such that
  $|\Lambda_{j_1}|=|\Lambda_{j_2}|=2$.
  \item [(2)] $\Lambda_j\subseteq\Pi\backslash\{S_1,S_2,S_3\}$
  for all $j\in[n]$.
  \item [(3)] There is an $\{\ell',\ell''\}\subseteq\{1,2,3\}$ such
  that for all $j\in[n]$,
  $\Lambda_j\cap\text{reg}^\circ(S_{\ell'},S_{\ell''})\neq\emptyset$.
  \item [(4)] $n\leq 2$.
\end{itemize}
\end{thm}
\begin{proof}
1) Suppose condition (1) holds. Let $A$ be the subset of $[n]$
such that $|\Lambda_j|=2$ for all $j\in A$ and $|\Lambda_j|>2$ for
all $j\in[n]\backslash A$. Let $\mathcal I=\{\Lambda_{j};j\in
A\}\cup\{[R];R\in\Pi\backslash(\cup_{j\in A}\Lambda_{j})\}$, where
$[R]=\{R\}$ for all $R\in\Pi\backslash(\cup_{j\in A}\Lambda_{j})$.
Then for each $j\in A$, $\Lambda_j$ is an equivalent class and for
each $R\in\Pi\backslash(\cup_{j\in A}\Lambda_{j})$, $\{R\}$ is an
equivalent class. By Definition \ref{cmptl}, $\mathcal I$ is
compatible. By Lemma \ref{comp-code}, $\text{RG}(D^{**})$ is
feasible. Fig. \ref{fg-smpl-case2} is an example of such region
graph and feasible code.

2) Suppose condition (2) holds. Let $\mathcal
I=\{[S_1],[S_2],[S_3]$, $[R]\}$, where $[S_i]=\{S_i\}$ for
$i\in\{1,2,3\}$ and $[R]=\Pi\backslash\{S_1,S_2,S_3\}$. Then by
Definition \ref{cmptl}, $\mathcal I$ is compatible. By Lemma
\ref{comp-code}, $\text{RG}(D^{**})$ is feasible. Fig.
\ref{fg-smpl-case3} (a) is an example of such region graph and
feasible code.

3) Suppose condition (3) holds. Without loss of generality, assume
$\ell=1, \ell'=2$ and $\ell''=3$. Let $\mathcal
I=\{[S_1],[S_2],[S_3]\}$, where $[S_1]=\Pi\backslash\{S_2,S_3\}$,
$[S_2]=\{S_2\}$ and $[S_3]=\{S_3\}$. Then by Definition
\ref{cmptl}, $\mathcal I$ is compatible. So by Lemma
\ref{comp-code}, $\text{RG}(D^{**})$ is feasible. Fig.
\ref{fg-smpl-case3} (b) is an example of such region graph and
feasible code.

4) Suppose condition (4) holds. If $n=1$, the conclusion is
trivial. So we assume $n=2$. We have the following three cases:

Case 1: $|\Lambda_1|>2$ or $|\Lambda_2|>2$. In this case, the
condition (1) holds and by proved result, $\text{RG}(D^{**})$ is
feasible.

Case 2: $|\Lambda_1|=|\Lambda_2|=2$ and
$\Lambda_1\cap\Lambda_2=\emptyset$. In this case, the condition
(1) holds and by proved result, $\text{RG}(D^{**})$ is feasible.

Case 3: $|\Lambda_1|=|\Lambda_2|=2$ and
$\Lambda_1\cap\Lambda_2\neq\emptyset$. If $\Lambda_1=\Lambda_2$,
then the conclusion is trivial. So we assume
$\Lambda_1\neq\Lambda_2$. Thus, we can assume
$\Lambda_1=\{Q_1,Q_2\}$ and $\Lambda_2=\{Q_1,Q_3\}$. We have the
following two subcases:

Case 3.1: $Q_1=S_i$ for some $i\in\{1,2,3\}$. Without loss of
generality, assume $Q_1=S_1$. By 2) of Lemma \ref{in-reg-lmd-t-s},
$\{Q_1,Q_2\}=\Lambda_1\nsubseteq\text{reg}(S_{1},S_{2})$ and
$\{Q_1,Q_2\}=\Lambda_1\nsubseteq\text{reg}(S_{1},S_{3})$. Then we
have $Q_2\in\text{reg}^\circ(S_{2},S_{3})$. Similarly, we have
$Q_3\in\text{reg}^\circ(S_{2},S_{3})$. So condition (3) holds and
by the proved result, $\text{RG}(D^{**})$ is feasible.

Case 3.2: $Q_1\neq S_i$ for all $i\in\{1,2,3\}$. Then
$Q_1\in\text{reg}^\circ(S_{i_1},S_{i_2})$ for some
$\{i_1,i_2\}\subseteq\{1,2,3\}$. So condition (3) holds and by the
proven result, $\text{RG}(D^{**})$ is feasible.
\end{proof}

\begin{rem}\label{rem-num-leq-2}
We remark that if the number of different subsets in
$\{\Lambda_1,\Lambda_2,\cdots,\Lambda_n\}$ is at most $2$, then we
can still construct a feasible code on $\Pi$. So
$\text{RG}(D^{**})$ is still feasible.
\end{rem}

For example, if $\{\Lambda_1\neq\Lambda_2=\cdots=\Lambda_n\}$.
Then consider the subgraph $\text{RG}(D^{**})'$ formed by
$\Pi\cup\Omega_1\cup\Omega_2$. By (4) of Theorem
\ref{simple-case}, $\text{RG}(D^{**})'$ is feasible. So by Theorem
\ref{lmd-solv}, there exists a feasible code on $\Pi$. Again by
Theorem \ref{lmd-solv}, $\text{RG}(D^{**})$ is feasible.

%%%%%%%%%%%%%%%%%%%%%%%%%%%%%%%%%%%%%%%%%%%
\renewcommand\figurename{Fig}
\begin{figure}[htbp]
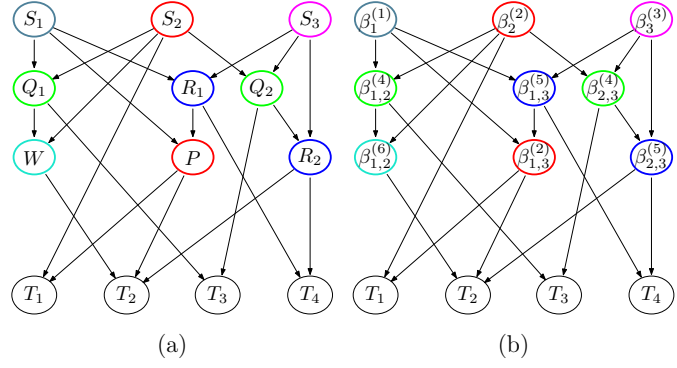

\begin{center}
\vspace{0.2cm}\includegraphics[height=4.8cm]{fg-smpl.3}
\hspace{0.1cm}\includegraphics[height=4.8cm]{fg-smpl.4}
\end{center}
\caption{An example of weak decentralized code: (a) is a region
graph with $\Pi=\{S_1, S_2, S_3, Q_1, R_1, Q_2, W, P, R_2\}$ and
$\Lambda_1=\{\{S_2, P\}, \Lambda_2=\{W, P, R_2\}$,
$\Lambda_3=\{Q_1, Q_2\}, \Lambda_4=\{R_1, R_2\}$. Let $\mathcal
I=\{[S_1], [S_2], [S_3], [Q_1], [R_1], [W]\}$, where
$[S_1]=\{S_1\}$, $[S_2]=\{S_2, P\}$, $[S_3]=\{S_3\}$,
$[Q_1]=\{Q_1, Q_2\}$, $[R_1]=\{R_1, R_2\}$ and $[W]=\{W\}$. Then
$\mathcal I$ is compatible. (b) illustrates the $\mathcal I$-weak
decentralized code, where $\mathcal
B_1=\{\beta^{(1)}_{1},\beta^{(1)}_{2,3}\},\cdots,\mathcal
B_6=\{\beta^{(6)}_{1,2},\beta^{(6)}_{1,3},\beta^{(6)}_{2,3}\}$ are
as in Lemma \ref{genc-code}. }\label{fg-smpl-case2}
\end{figure}
%%%%%%%%%%%%%%%%%%%%%%%%%%%%%%%%%%%%%%%%%%%%%%

%%%%%%%%%%%%%%%%%%%%%%%%%%%%%%%%%%%%%%%%%%%
\renewcommand\figurename{Fig}
\begin{figure}[htbp]
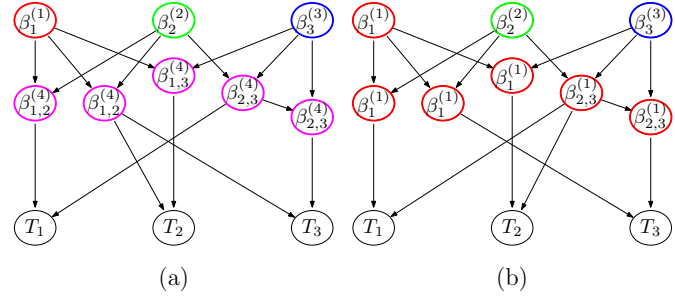

\begin{center}
\vspace{0.2cm}\includegraphics[height=3.9cm]{fg-smpl.8}
\hspace{0.1cm}\includegraphics[height=3.9cm]{fg-smpl.9}
\end{center}
\caption{Examples of weak decentralized code: In (a),
$\Lambda_j\subseteq\Pi\backslash\{S_1,S_2,S_3\}$ for
$j\in\{1,2,3\}$; In (b),
$\Lambda_j\cap\text{reg}^\circ(S_2,S_3)\neq\emptyset,
j=1,2,3$.}\label{fg-smpl-case3}
\end{figure}
%%%%%%%%%%%%%%%%%%%%%%%%%%%%%%%%%%%%%%%%%%%%%%

As a simple corollary of (4) of Theorem \ref{simple-case}, we can
characterize solvability of $3$s$/2$t sum-network (not necessarily
terminal-separable) as follows.
\begin{cor}\label{3s-2t}
If $G$ is a $3$s$/2$t sum-network, then the basic region graph
$\text{RG}(D^{**})$ is feasible.
\end{cor}
\begin{proof}
Since $G$ is a $3$s$/2$t sum-network, then $\text{RG}(D^{**})$ has
at most two terminal regions. If $\text{RG}(D^{**})$ has only one
terminal region, then clearly, it is feasible. So we assume that
$\text{RG}(D^{**})$ has two terminal regions, say $T_1$ and $T_2$.
Then we have the following two cases:

Case 1: $\Omega_{1,2}\neq\emptyset$. Pick a $P\in\Omega_{1,2}$. By
(2) of Definition \ref{lmd-omd}, $P\in D^{**}\backslash\Pi$. Then
by Lemma \ref{in-omg-intc}, $S_i\rightarrow P$ for all
$i\in\{1,2,3\}$ and by Lemma \ref{sub-g-reg-code}, the sum
$\sum_{i=1}^3X_i$ can be transmitted from $\{S_1,S_2,S_3\}$ to
$P$. Moreover, by (2) of Definition \ref{lmd-omd}, we have
$P\rightarrow T_j, j=1,2$. So by Lemma \ref{P-C-fsb},
$\text{RG}(D^{**})$ is feasible.

Case 2: $\Omega_{1,2}=\emptyset$. Then $\text{RG}(D^{**})$ is
terminal separable and by (4) of Theorem \ref{simple-case},
$\text{RG}(D^{**})$ is feasible.

By the above discussion, $\text{RG}(D^{**})$ is feasible.
\end{proof}

Recall that in this paper we assumed $S_i\rightarrow T_j$ for all
$i\in\{1,2,3\}$ and $j\in[n]$ (Assumption 2). On the other hand,
if $S_i\nrightarrow T_j$ for some $i\in\{1,2,3\}$ and $j\in[n]$,
then clearly $\text{RG}(D^{**})$ is infeasible. So by Corollary
\ref{3s-2t}, a $3$s$/2$t sum-network is feasible if and only if
each source-terminal pair is connected, which is a special case of
the result of \cite{Rama08}.

When $n=3$, we can give a graph theoretic characterization for
feasibility of $\text{RG}(D^{**})$.
\begin{thm}\label{sgl-3t}
Suppose $\text{RG}(D^{**})$ is terminal separable and has three
terminal regions. Then $\text{RG}(D^{**})$ is infeasible if and
only if by proper naming, the following condition (C-IR) hold: \\
\textbf{(C-IR)} There is a $P_1\in\text{reg}^\circ(S_{2},S_{3})$
and a $P_2\in\text{reg}^\circ(S_{1},S_{2})$ such that
$\Lambda_{1}=\{S_{1},P_1\},\Lambda_{2}=\{P_1,P_2\}$ and
$\Lambda_{3}\subseteq\text{reg}(S_{1},P_2)\cup\text{reg}(S_{1},
S_{3})$.
\end{thm}
\begin{proof}
The proof is given in Appendix C.
\end{proof}

A new characterization on the solvability of a $3$s$/3$t
sum-network $G$ (not necessarily terminal-separable) can be
derived as a corollary of Theorem \ref{sgl-3t}.
\begin{cor}\label{3s-3t}
For a $3$s$/3$t sum-network $G$, $\text{RG}(D^{**})$ is infeasible
if and only if it is terminal separable and satisfies condition
(C-IR).
\end{cor}
\begin{proof}
Since $G$ is a $3$s$/3$t sum-network, then $\text{RG}(D^{**})$ has
at most $3$ terminal regions. If $\text{RG}(D^{**})$ has at most
two terminal regions, then by Corollary \ref{3s-2t}, it is
feasible. So we assume $\text{RG}(D^{**})$ has three terminal
regions. Then one of the following three cases holds:

Case 1: $\Omega_{1,2,3}\neq\emptyset$. Pick a $P\in
\Omega_{1,2,3}$. As in the proof of Corollary \ref{3s-2t}, the sum
$\sum_{i=1}^3X_i$ can be transmitted from $\{S_1,S_2,S_3\}$ to
$P$. Moreover, by (2) of Definition \ref{lmd-omd}, we have
$P\rightarrow T_j, j=1,2,3$. So by Lemma \ref{P-C-fsb},
$\text{RG}(D^{**})$ is feasible.

Case 2: $\Omega_{1,2,3}=\emptyset$ and
$\Omega_{i_1,i_2}\neq\emptyset$ for some
$\{i_1,i_2\}\subseteq\{1,2,3\}$. Pick a $Q\in\Omega_{i_1,i_2}$. By
Corollary \ref{3s-2t}, the sum $\sum_{i=1}^3X_i$ can be
transmitted from $\{S_1,S_2,S_3\}$ to $\{Q,T_{i_3}\}$
simultaneously, where $\{i_3\}=\{1,2,3\}\backslash\{i_1,i_2\}$.
Moreover, by (2) of Definition \ref{lmd-omd}, $P\rightarrow T_j,
j=i_1,i_2$. So by Lemma \ref{P-C-fsb}, $\text{RG}(D^{**})$ is
feasible.

Case 3: $\text{RG}(D^{**})$ is terminal separable. By Theorem
\ref{sgl-3t}, $\text{RG}(D^{**})$ is infeasible if and only if the
condition (C-IR) holds.

By the above discussion, $\text{RG}(D^{**})$ is infeasible if and
only if it is terminal separable and satisfies condition (C-IR).
\end{proof}

Fig. \ref{fg-infsb} are two examples of infeasible region graph
with three source regions and three terminal regions.

The first necessary and sufficient condition for solvability of
$3$s$/3$t sum-network is given by Shenvi and Dey [18, Th.1]. We
remark that the conditions in \cite{Shenvi10} can be easily
derived from Corollary \ref{3s-3t}. In fact, by interchanging the
name of $S_2$ and $S_3$ and replacing the name of $T_1,T_2,T_3$ by
$T_2,T_3,T_1$ respectively, condition (C-IR) can be restated as
\\ (\textbf{C-IR}$\boldmath{'}$) There is a
$P_1\in\text{reg}^\circ(S_{2},S_{3})$ and a
$P_2\in\text{reg}^\circ(S_{1},S_{3})$ such that
$\Lambda_{2}=\{S_{1},P_1\},\Lambda_{3}=\{P_1,P_2\}$ and
$\Lambda_{1}\subseteq\text{reg}(S_{1},P_2)\cup\text{reg}(S_{1},
S_{2})$.\\
Let $e_1=\text{lead}(P_2)$ and $e_2=\text{lead}(P_1)$. Then we can
check that $e_1,e_2$ satisfy the conditions 1)$-$6) of [18, Th.1].

For example, for the network in Fig. \ref{fg-unsolv-net} (a), its
basic region graph is in Fig. \ref{fg-unsolv-reg} (a) and
satisfies the condition (C-IR$'$). Let
$e_1=\text{lead}(P_2)=(v_1,v_3)$ and
$e_2=\text{lead}(P_1)=(v_2,v_4)$. Then $e_1,e_2$ satisfy the
conditions 1)$-$6) of [18, Th.1].

Similarly, for the network in Fig. \ref{fg-unsolv-net} (b), its
basic region graph is in Fig. \ref{fg-unsolv-reg} (b) and
satisfies the condition (C-IR$'$). Let
$e_1=\text{lead}(P_2)=(v_1,v_4)$ and
$e_2=\text{lead}(P_1)=(v_3,v_5)$. Then $e_1,e_2$ satisfy the
conditions 1)$-$6) of [18, Th.1].

Note that by Theorem \ref{b-reg-unq}, the basic region graph
$\text{RG}(D^{**})$ can be obtained in time $O(|E|)$. Moreover, by
Theorem \ref{omg-intc}, it is $O(|D^{**}|)$ time complexity to
determine whether $\text{RG}(D^{**})$ is terminal-separable. So by
Corollary \ref{3s-3t}, it is $O(|E|)$ time complexity to determine
solvability of a $3$s$/3$t sum-network, where $E$ is the link set.
However, by [18, Th.1], it needs $O(|E|^3)$ time complexity. Thus,
our result gives a faster method to determine solvability of $G$.

\section{Conclusions and Discussions}
We investigated the network coding problem of $3$s$/n$t
sum-network. We obtained a computationally simple sufficient and
necessary condition for solvability of a class of $3$s$/n$t
sum-network by developing the region decomposition method in
\cite{Wentu11, Wentu12} and generalizing the decentralized coding
method in \cite{Fragouli06}. The condition was characterized by
some simple structural properties of some certain partitions on a
part of the region graph and also can be judged using a very
simple polynomial time algorithm. As a result, the solvability of
$3$s$/3$t sum-networks was characterized by using a single
forbidden structure. Our method can further develop a completely
characterization on the solvability of $3$s$/4$t sum-networks.
Limited by the space, we leave it to a future paper.

\appendices

\section{Proof of Lemma \ref{genc-code}}
\begin{proof}[Proof of Lemma \ref{genc-code}]
Clearly, when $K=3$, the sets $\mathcal B_1, \mathcal B_2,
\mathcal B_3$ satisfy properties (1)$-$(4) of Lemma
\ref{genc-code}.

Now suppose $K>3$ and the sets $\mathcal B_1, \cdots, \mathcal
B_{K-1}$ satisfy properties (1)$-$(4). We want to construct a
subset $\mathcal
B_{K}=\{\beta^{(K)}_{1,2},\beta^{(K)}_{1,3},\beta^{(K)}_{2,3}\}$
such that the sets $\mathcal B_1, \cdots, \mathcal B_{K-1},
\mathcal B_{K}$ satisfy properties (1)$-$(4). The key is to
carefully choose a vector $\beta^{(K)}\in\mathbb
F^3\backslash\langle\bar{\alpha}\rangle$ and let
$0\neq\beta^{(K)}_{i_1,i_2}\in\langle\beta^{(K)},
\bar{\alpha}\rangle\cap\langle\alpha_{i_1}, \alpha_{i_2}\rangle$
for each $\{i_1,i_2\}\subseteq\{1,2,3\}$.

Let $\Phi_{K-1}$ be the set of all pairs
$\{\gamma',\gamma''\}\subseteq\bigcup_{\ell=1}^{K-1}\mathcal
B_\ell$ such that
$\{\gamma',\gamma''\}\nsubseteq\langle\alpha_{i_1},\alpha_{i_2}\rangle$
for all $\{i_1,i_2\}\subseteq\{1,2,3\}$. Note that $\mathcal B_1,
\cdots, \mathcal B_{K-1}$ satisfy property (3) of Lemma
\ref{genc-code}. Then $\gamma'$ and $\gamma''$ are linearly
independent. So $\langle\gamma',\gamma''\rangle$ is an
$2$-dimensional subspace of $\mathbb F^3$ and
$\langle\gamma',\gamma''\rangle
\neq\langle\alpha_{i_1},\alpha_{i_2}\rangle$. Thus,
$\langle\gamma',\gamma''\rangle
\cap\langle\alpha_{i_1},\alpha_{i_2}\rangle$ is an $1$-dimensional
subspace of $\mathbb F^3$. Let
$$\langle\gamma',\gamma''\rangle_{i_1,i_2}=\langle\gamma',\gamma''\rangle
\cap\langle\alpha_{i_1},\alpha_{i_2}\rangle$$ and let
\begin{align}
\Psi_{K-1}=\bigcup_{\{\gamma',\gamma''\}\in\Phi_{K-1}}\{\langle\gamma',
\gamma''\rangle_{1,2}, \langle\gamma',\gamma''\rangle_{1,3},
\langle\gamma',\gamma''\rangle_{2,3}\}. \notag %\label{eq-genc-code-0}
\end{align}
Then $\Psi_{K-1}\subseteq\langle\alpha_{1},
\alpha_{2}\rangle\cup\langle\alpha_{1},
\alpha_{3}\rangle\cup\langle\alpha_{2}, \alpha_{3}\rangle$. Since
$\mathbb F$ is sufficiently large, there exists a
$\beta^{(K)}\in\mathbb F^3$ such that
$\beta^{(K)}\notin\langle\bar{\alpha}, \gamma\rangle$ for all
$\gamma\in\Psi_{K-1}$. Then we have
\begin{align}
\gamma\notin\langle\bar{\alpha}, \beta^{(K)}\rangle,~
\forall\gamma\in\Psi_{K-1}.\label{eq-genc-code-1}
\end{align}
For each $\{i_1,i_2\}\subseteq\{1,2,3\}$, let
\begin{align}
0\neq\beta^{(K)}_{i_1,i_2}\in\langle\beta^{(K)},
\bar{\alpha}\rangle\cap\langle\alpha_{i_1}, \alpha_{i_2}\rangle.
\label{eq-genc-code-2}
\end{align}
Let $\mathcal B_{K}=\{\beta^{(K)}_{1,2}$, $\beta^{(K)}_{1,3}$,
$\beta^{(K)}_{2,3}\}$. We shall prove $\mathcal B_1, \cdots,
\mathcal B_{K-1}, \mathcal B_{K}$ satisfy properties (1)$-$(4) of
Lemma \ref{genc-code}.

By (\ref{eq-genc-code-2}),
$\beta^{(K)}_{i_1,i_2}\in\langle\alpha_{i_1}, \alpha_{i_2}\rangle$
for all $\{i_1,i_2\}\subseteq\{1,2,3\}$. So $\mathcal B_1, \cdots,
\mathcal B_{K-1}, \mathcal B_{K}$ satisfy property (1).

By assumption, $\mathcal B_1, \cdots, \mathcal B_{K-1}$ satisfy
property (2) of Lemma \ref{genc-code}, then for any
$\ell\in\{1,\cdots,K-1\}$ and
$\{\gamma',\gamma''\}\subseteq\mathcal B_\ell$, the pair
$\{\gamma',\gamma''\}$ is in $\Phi_{K-1}$. Moreover, since
$\mathcal B_1, \cdots, \mathcal B_{K-1}$ satisfy property (1) of
Lemma \ref{genc-code}, then $\{\gamma',\gamma''\}\subseteq\mathcal
B_\ell\subseteq\langle\alpha_{1},\alpha_{2}\rangle\cup\langle\alpha_{1},
\alpha_{3}\rangle\cup\langle\alpha_{2},\alpha_{3}\rangle$. So
$\{\gamma',\gamma''\}\subseteq\{\langle\gamma',
\gamma''\rangle_{1,2}$, $\langle\gamma',\gamma''\rangle_{1,3}$,
$\langle\gamma',\gamma''\rangle_{2,3}\}$. Thus,
$\bigcup_{\ell=1}^{K-1}\mathcal B_\ell\subseteq\Psi_{K-1}$ and by
(\ref{eq-genc-code-1}),
\begin{align}
\gamma\notin\langle\beta^{(K)},\bar{\alpha}\rangle, ~
\forall\gamma\in\textstyle{\bigcup_{\ell=1}^{K-1}\mathcal
B_\ell}.\label{eq-genc-code-4}
\end{align}
Note that $\alpha_i\in\mathcal B_i, i=1,2,3$. Then
$\alpha_{i}\notin\langle\beta^{(K)},\bar{\alpha}\rangle, i=1,2,3$.
So by (\ref{eq-genc-code-2}), $\beta^{(K)}_{1,2}$,
$\beta^{(K)}_{1,3}$ and $\beta^{(K)}_{2,3}$ are mutually linearly
independent and $\bar{\alpha}\in\langle\gamma,\gamma'\rangle$ for
all $\{\gamma, \gamma'\}\subseteq\mathcal B_K$. Thus, $\mathcal
B_1, \cdots, \mathcal B_{K-1}, \mathcal B_{K}$ satisfy property
(2).

Now, we prove that $\mathcal B_1, \cdots, \mathcal B_{K-1},
\mathcal B_{K}$ satisfy property (4). Suppose
$\{\gamma,\gamma',\gamma''\}\subseteq\bigcup_{\ell=1}^{K}\mathcal
B_\ell$ such that $\{\gamma, \gamma',\gamma''\}
\nsubseteq\langle\alpha_{i_1},\alpha_{i_2}\rangle$ for all
$\{i_1,i_2\}\subseteq\{1,2,3\}$ and $\{\gamma, \gamma', \gamma''\}
\neq\{\beta^{(\ell)}_{1,2}$, $\beta^{(\ell)}_{1,3}$,
$\beta^{(\ell)}_{2,3}\}$ for all $\ell\in\{4,\cdots,K\}$. We have
the following three cases:

Case 1:
$\{\gamma,\gamma',\gamma''\}\subseteq\bigcup_{\ell=1}^{K-1}\mathcal
B_\ell$. By induction assumption, $\gamma, \gamma'$ and $\gamma''$
are linearly independent.

Case 2:
$\{\gamma',\gamma''\}\subseteq\bigcup_{\ell=1}^{K-1}\mathcal
B_\ell$ and $\gamma\in\mathcal B_K$. We have the following two
subcases:

Case 2.1: $\{\gamma',\gamma''\}\subseteq\langle\alpha_{i_1},
\alpha_{i_2}\rangle$ for some $\{i_1,i_2\}\subseteq\{1,2,3\}$. By
assumption of $\{\gamma, \gamma',\gamma''\}$, we have
$\gamma\notin\langle\alpha_{i_1}, \alpha_{i_2}\rangle$. So
$\gamma, \gamma'$ and $\gamma''$ are linearly independent.

Case 2.2: $\{\gamma',\gamma''\}\nsubseteq\langle\alpha_{i_1},
\alpha_{i_2}\rangle$ for all $\{i_1,i_2\}\subseteq\{1,2,3\}$. Then
the pair $\{\gamma',\gamma''\}$ is in the set $\Phi_{K-1}$ and
$\gamma', \gamma''$ are linearly independent. Note that
$\gamma\in\mathcal B_K=\{\beta^{(K)}_{1,2}$, $\beta^{(K)}_{1,3},
\beta^{(K)}_{2,3}\}$. Then by (\ref{eq-genc-code-2}),
$\gamma\in\langle\alpha_{1},\alpha_{2}\rangle\cup\langle\alpha_{1},
\alpha_{3}\rangle\cup\langle\alpha_{2},\alpha_{3}\rangle$. So
$\gamma, \gamma'$ and $\gamma''$ are linearly independent.
Otherwise, we have $\gamma\in\langle\gamma', \gamma''\rangle$ and
$\gamma\in\{\langle\gamma', \gamma''\rangle_{1,2},
\langle\gamma',\gamma''\rangle_{1,3},
\langle\gamma',\gamma''\rangle_{2,3}\}\subseteq\Psi_{K-1}$. By
(\ref{eq-genc-code-2}), $\gamma\in\langle\bar{\alpha},
\beta^{(K)}\rangle$, which contradicts to (\ref{eq-genc-code-1}).

Case 3: $\gamma''\in\bigcup_{\ell=1}^{K-1}\mathcal B_\ell$ and
$\{\gamma,\gamma'\}\subseteq\mathcal B_K$. Then by the proven
result, $\gamma$ and $\gamma'$ are linearly independent. So by
(\ref{eq-genc-code-2}), $\langle\beta^{(K)},\bar{\alpha}\rangle=
\langle\gamma',\gamma''\rangle$. From (\ref{eq-genc-code-4}), we
have $\gamma\notin\langle\beta^{(K)},\bar{\alpha}\rangle=
\langle\gamma',\gamma''\rangle$. Thus, $\gamma, \gamma'$ and
$\gamma''$ are linearly independent.

In all cases, $\gamma, \gamma'$ and $\gamma''$ are linearly
independent. Thus, $\bigcup_{\ell=1}^{K}\mathcal B_\ell$ satisfies
condition (4).

Clearly, if $\mathcal B_1, \cdots, \mathcal B_{K-1}, \mathcal
B_{K}$ satisfy property (4), then for any $\{\gamma,
\gamma'\}\subseteq\bigcup_{\ell=1}^K\mathcal B_\ell$, we can
always find a vector $\gamma''\in\bigcup_{\ell=1}^K\mathcal
B_\ell$ such that $\gamma, \gamma'$ and $\gamma''$ are linearly
independent. So $\gamma$ and $\gamma'$ are linearly independent
and $\mathcal B_1, \cdots, \mathcal B_{K-1}, \mathcal B_{K}$
satisfy property (3).

Thus, we can always find $K$ sets $\mathcal B_1, \cdots, \mathcal
B_{K-1}, \mathcal B_{K}$ satisfying properties (1)$-$(4) of Lemma
\ref{genc-code} for all $K\geq 3$.
\end{proof}

We give an example of the above construction. For simplicity, we
let $\mathbb F=GF(p)$ for a sufficiently large prime $p$.

When $K=4$, we have $\bigcup_{\ell=1}^{K-1}\mathcal
B_\ell=\{\alpha_1, \alpha_2+\alpha_3, \alpha_2, \alpha_1+\alpha_3,
\alpha_3, \alpha_1+\alpha_2\}$. So $\Phi_3=\{\{\alpha_1,
\alpha_2+\alpha_3\}, \{\alpha_2, \alpha_1+\alpha_3\}, \{\alpha_3,
\alpha_1+\alpha_2\}, \{\alpha_1+\alpha_2, \alpha_1+\alpha_3\},
\{\alpha_1+\alpha_2, \alpha_2+\alpha_3\}, \{\alpha_1+\alpha_3,
\alpha_2+\alpha_3\}\}$. Correspondingly, we have
$\Psi_3=\{\langle\alpha_1\rangle,
\langle\alpha_2+\alpha_3\rangle\}\cup\{\langle\alpha_2\rangle,
\langle\alpha_1+\alpha_3\rangle\}\cup\{\langle\alpha_3\rangle,
\langle\alpha_1+\alpha_2\rangle\}\cup\{\langle\alpha_1+\alpha_2\rangle,
\langle\alpha_1+\alpha_3\rangle,
\langle\alpha_2-\alpha_3\rangle\},
\{\langle\alpha_1+\alpha_2\rangle,
\langle\alpha_2+\alpha_3\rangle,
\langle\alpha_1-\alpha_3\rangle\},
\{\langle\alpha_1+\alpha_3\rangle,
\langle\alpha_2+\alpha_3\rangle,
\langle\alpha_1-\alpha_2\rangle\}\}$. Let
$\beta^{(4)}=\alpha_1+3\alpha_2$. Then $\beta^{(4)}$ satisfies
(\ref{eq-genc-code-1}). By (\ref{eq-genc-code-2}), we have
$\beta^{(4)}_{1,2}=\alpha_1+3\alpha_2$,
$\beta^{(4)}_{1,3}=2\alpha_1+3\alpha_3$ and
$\beta^{(4)}_{2,3}=2\alpha_2-\alpha_3$. So we can obtain $\mathcal
B_4=\{\alpha_1+3\alpha_2, 2\alpha_1+3\alpha_3,
2\alpha_2-\alpha_3\}$. Similarly, we can construct $\mathcal
B_5=\{2\alpha_1+3\alpha_2, \alpha_1+3\alpha_3,
\alpha_2-2\alpha_3\}$ (See Example \ref{ex-genc-code}.), and so
on.

\section{Proof of Lemma \ref{chpt-lem}}
In what follows, we suppose $\tilde{C}_\Pi=\{d_R\in\mathbb F^3;
R\in\Pi\}$ is a feasible code on $\Pi$. By Remark
\ref{rem-sgl-sets}, we can assume that $d_R\neq 0$ for all
$R\in\Pi$.

To prove Lemma \ref{chpt-lem}, the key is to prove that all
equivalent class $[R]\in\mathcal I_c$ satisfies the following
property:\\
\textbf{(P1)} : For any pair $\{Q,Q'\}\subseteq[R]$,
$d_{Q'}\in\langle d_{Q}, \bar{\alpha}\rangle$.

To prove this, we first prove three lemmas.

\begin{lem}\label{claim 1}
Let $\mathcal I$ be a partition of $\Pi$ and $[R]\in\mathcal I$
satisfying (P1). Then for any $\{i_1,i_2\}\subseteq\{1,2,3\}$ and
any pair $\{Q,Q'\}\subseteq[R]_{i_1,i_2}$, $d_{Q'}\in\langle
d_{Q}\rangle$. Moreover, for any subclass $[[R]]$ of $[R]$ and any
pair $\{Q,Q'\}\subseteq[[R]]$, $d_{Q'}\in\langle d_{Q}\rangle$.
\end{lem}
\begin{proof}
By (\ref{eq-nt-3}),
$\{Q,Q'\}\subseteq[R]_{i_1,i_2}\subseteq\text{reg}(S_{i_1},S_{i_2})$.
So by 1) of Lemma \ref{lem-sgl-sets}, $d_Q,
d_{Q'}\in\langle\alpha_{i_1}, \alpha_{i_2}\rangle$. Meanwhile,
since $[R]\in\mathcal I$ satisfies (P1), then $d_{Q'}\in\langle
d_{Q}, \bar{\alpha}\rangle$. So
$d_{Q'}\in\langle\alpha_{i_1},\alpha_{i_2}\rangle\cap\langle
d_{Q},\bar{\alpha}\rangle=\langle d_{Q}\rangle$ and the first
claim is true.

Now, we prove the second claim. Suppose $\{Q,Q'\}\subseteq[[R]]$.
If $[R]\neq[S_i]$ for all $i\in\{1,2,3\}$, then by Definition
\ref{sub-class}, $[[R]]=[R]_{i_1,i_2}$ for some
$\{i_1,i_2\}\subseteq\{1,2,3\}$ and by the proven result,
$d_{Q'}\in\langle d_{Q}\rangle$. If $[R]=[S_i]$ for some
$i\in\{1,2,3\}$, then by Definition \ref{sub-class}, we have the
following two cases:

Case 1: $[[R]]=[S_i]_i$. By (\ref{eq-nt-3}) and (\ref{eq-nt-4}),
$[[R]]=[S_i]_{i,j_1}\cup[S_i]_{i,j_2}$, where
$\{j_1,j_2\}=\{1,2,3\}\backslash\{i\}$. By the proven result, we
have $\alpha_i=d_{S_i}\in\langle d_Q\rangle$ and $d_{Q'}\in\langle
d_{S_i}\rangle$. So $d_{Q'}\in\langle d_{Q}\rangle$.

Case 2: $[[R]]=[S_i]_{j_1,j_2}$, where
$\{j_1,j_2\}=\{1,2,3\}\backslash\{i\}$. By the proven result, we
have $d_{Q'}\in\langle d_{Q}\rangle$.

In both cases, we have $d_{Q'}\in\langle d_{Q}\rangle$, which
proves the second claim.
\end{proof}

\begin{lem}\label{claim 2}
Let $\mathcal I$ be a partition of $\Pi$ and
$\{[R'],[R'']\}\subseteq\mathcal I$ such that $[R']$ and $[R'']$
satisfy (P1). If there is a $j\in[n]$ such that
$\Lambda_j\subseteq[[R']]\cup[[R'']]$, where $[[R']]~($resp.
$[[R'']])$ is a subclass of $[R']~($resp. $[R''])$. Then
$\bar{\alpha}\in\langle d_{P'},d_{P''}\rangle$ for any
$P'\in[[R']]$ and $P''\in[[R'']]$.
\end{lem}
\begin{proof}
By Lemma \ref{claim 1}, $d_{Q'}\in\langle d_{P'}\rangle$ for all
$Q'\in[[R']]$ and $d_{Q''}\in\langle d_{P''}\rangle$ for all
$Q''\in[[R'']]$. Then we have $\langle d_Q;
Q\in[[R']]\cup[[R'']]\rangle=\langle d_{P'},d_{P''}\rangle$. Note
that by assumption, $\Lambda_j\subseteq[[R']]\cup[[R'']]$. Then
$\langle d_Q; Q\in\Lambda_j\rangle\subseteq\langle d_Q;
Q\in[[R']]\cup[[R'']\rangle=\langle d_{P'},d_{P''}\rangle$. Now,
since $\tilde{C}_\Pi$ is a feasible code on $\Pi$, then we have
$\bar{\alpha}\in\langle d_Q; Q\in\Lambda_j\rangle\subseteq\langle
d_{P'},d_{P''}\rangle$.
\end{proof}

\begin{lem}\label{claim 3}
Let $\mathcal I$ be a partition of $\Pi$ and $\mathcal I'$ be a
contraction of $\mathcal I$ by combining two connected equivalent
classes in $\mathcal I$. If all equivalent classes in $\mathcal I$
satisfy (P1), then all equivalent classes in $\mathcal I'$ satisfy
(P1).
\end{lem}
\begin{proof}
Suppose $[R'],[R'']\in\mathcal I$ are connected and $\mathcal I'$
is obtained by combining $[R']$ and $[R'']$.

Suppose $[R]\in\mathcal I'$. If $[R]\neq[R']\cup[R'']$, then
$[R]\in\mathcal I$, and by assumption, $[R]$ satisfies (P1). So we
only need to prove that $[R]=[R']\cup[R'']$ satisfies (P1), i.e.,
$d_{Q'}\in\langle d_{Q}, \bar{\alpha}\rangle$ for any pair
$\{Q,Q'\}\subseteq[R']\cup[R'']$. Note that by assumption, $[R'],
[R'']\in\mathcal I$ satisfy (P1). So if $\{Q,Q'\}\subseteq[R']$ or
$\{Q,Q'\}\subseteq[R'']$, then $d_{Q'}\in\langle d_{Q},
\bar{\alpha}\rangle$. Thus, we can assume $\{Q,Q'\}\nsubseteq[R']$
and $\{Q,Q'\}\nsubseteq[R'']$. By proper naming, we can assume
$Q\in[R']$ and $Q'\in[R'']$. Since, $[R']$ and $[R'']$ are
connected, by Definition \ref{cntd}, we have two cases:

Case 1: There is a $\Lambda_j\subseteq[[R']]\cup[[R'']]$, where
$[[R']]~($resp. $[[R'']])$ is a subclass of $[R']~($resp.
$[R''])$. By Lemma \ref{claim 2}, $\bar{\alpha}\in\langle
d_{P'},d_{P''}\rangle$, where $P'\in[[R']]$ and $P''\in[[R'']]$.
Similar to 2) of Lemma \ref{lem-sgl-sets}, we have $\langle
d_{P'}, \bar{\alpha}\rangle=\langle d_{P''},\bar{\alpha}\rangle$.
Since $Q', P''\in[R'']$ and $[R'']$ satisfy (P1), then
$d_{Q'}\in\langle d_{P''},\bar{\alpha}\rangle=\langle
d_{P'},\bar{\alpha}\rangle$. Similarly, since $Q, P'\in[R']$ and
$[R']$ satisfy (P1), then $d_{P'}\in\langle
d_{Q},\bar{\alpha}\rangle$. Thus, $d_{Q'}\in\langle
d_{P'},\bar{\alpha}\rangle\subseteq\langle
d_{Q},\bar{\alpha}\rangle$.

Case 2: There is a subset $\{i_1,i_2\}\subseteq\{1,2,3\}$ such
that
$\text{reg}([R']_{i_1,i_2})\cap\text{reg}([R'']_{i_1,i_2})\neq\emptyset$.
Suppose
$Q_0\in\text{reg}([R']_{i_1,i_2})\cap\text{reg}([R'']_{i_1,i_2})$.
Pick a $P_1\in\text{reg}([R']_{i_1,i_2})$. By Lemma \ref{claim 1},
$d_{P'}\in\langle d_{P_1}\rangle, \forall P'\in[R']_{i_1,i_2}$.
Then by 1) of Lemma \ref{lem-sgl-sets}, $d_{Q_0}\in\langle d_{P'};
P'\in[R']_{i_1,i_2}\rangle=\langle d_{P_1}\rangle$, which implies
that $\langle d_{P_1}\rangle=\langle d_{Q_0}\rangle$. Note that by
assumption, $[R']$ satisfies (P1). Then $d_{P_1}\in\langle
d_{Q},\bar{\alpha}\rangle$. By 1) of Lemma \ref{lem-sgl-sets},
$d_{P_1},d_{Q}\notin\langle\bar{\alpha}\rangle$. So we have
$\langle d_{Q}, \bar{\alpha}\rangle=\langle d_{P_1},
\bar{\alpha}\rangle=\langle d_{Q_0}, \bar{\alpha}\rangle$.
Similarly, Pick a $P_2\in\text{reg}([R'']_{i_1,i_2})$ and we have
$\langle d_{Q'}, \bar{\alpha}\rangle=\langle d_{P_2},
\bar{\alpha}\rangle=\langle d_{Q_0}, \bar{\alpha}\rangle$. So
$d_{Q'}\in\langle d_{Q_0},\bar{\alpha}\rangle=\langle
d_{Q},\bar{\alpha}\rangle$.

In both cases, $[R]=[R']\cup[R'']$ satisfies (P1). Thus, all
equivalent classes in $\mathcal I'$ satisfy (P1).
\end{proof}

Now we can prove Lemma \ref{chpt-lem}.
\begin{proof}[Proof of Lemma \ref{chpt-lem}]
First, for the trivial partition $\mathcal I_0$ of $\Pi$, since
$[R]=\{R\}, \forall [R]\in\mathcal I_0$, so $[R]$ naturally
satisfies property (P1) and $[S_i]\neq[S_j], \forall
\{i,j\}\subseteq\{1,2,3\}$.

By Definition \ref{chpt}, $\mathcal I_c=\mathcal I_L$, where
$\mathcal I_0,\mathcal I_1,\cdots,\mathcal I_L$ is a sequence of
partitions of $\Pi$ such that for each $\ell\in\{1,\cdots,L\}$ and
$\{i,j\}\subseteq\{1,2,3\}$, $[S_i]\neq[S_j]$ in $\mathcal
I_{\ell-1}$ and $\mathcal I_\ell$ is a contraction of $\mathcal
I_{\ell-1}$ by combining two connected equivalent classes. So by
Lemma \ref{claim 3}, all equivalent classes in $\mathcal I_\ell$
satisfy (P1). In particular, all equivalent classes in $\mathcal
I_c=\mathcal I_L$ satisfies (P1). Then Lemma \ref{chpt-lem} is
derived from Lemma \ref{claim 1}.
\end{proof}

\section{Proof of Theorem \ref{sgl-3t}}

In what follows, we always assume $G$ has three terminal regions
$T_1,T_2$ and $T_3$. If $\Lambda_{j_1}=\Lambda_{j_2}$ for some
$\{j_1,j_2\}\subseteq\{1,2,3\}$, then by Remark
\ref{rem-num-leq-2}, $\text{RG}(D^{**})$ is feasible. So we always
assume that $\Lambda_1,\Lambda_2,\Lambda_3$ are mutually
different.

We first give some lemmas, which will help to prove Theorem
\ref{sgl-3t}.

\begin{lem}\label{s2-ps}
For any $j\in\{1,2,3\}$, the following hold.
\begin{itemize}
  \item [1)] If $\Lambda_j=\{S_{j_1},P\}$ for some $j_1\in\{1,2,3\}$,
  then $P\in\text{reg}^\circ(S_{j_2},S_{j_3})$, where
  $\{j_2,j_3\}=\{1,2,3\}\backslash\{j_1\}$.
  \item [2)] If $\Lambda_j=\{P',P''\}\subseteq\Pi\setminus\{S_1,S_2,S_3\}$
  and $P'\in\text{reg}^\circ(S_{j_1},S_{j_2})$ for some
  $\{j_1,j_2\}\subseteq\{1,2,3\}$, then
  $P''\in\text{reg}^\circ(S_{j_1},S_{j_3})$ or
  $P''\in\text{reg}^\circ(S_{j_2},S_{j_3})$, where
  $\{j_3\}=\{1,2,3\}\backslash\{j_1,j_2\}$.
\end{itemize}
\end{lem}
\begin{proof}
1) By 2) of Lemma \ref{in-reg-lmd-t-s},
$\Lambda_j\nsubseteq\text{reg}(S_{j_1},S_{j_2})$ and
$\Lambda_j\nsubseteq\text{reg}(S_{j_1},S_{j_3})$. So if
$\Lambda_j=\{S_{j_1},P\}$, then
$P\notin\text{reg}(S_{j_1},S_{j_2})\cup\text{reg}(S_{j_1},S_{j_3})$.
Thus, we have $P\in\Pi\setminus(\text{reg}(S_{j_1},S_{j_2})\cup
\text{reg}(S_{j_1},S_{j_2}))=\text{reg}^\circ(S_{j_2},S_{j_3})$.

2) By 2) of Lemma \ref{in-reg-lmd-t-s},
$\Lambda_j\nsubseteq\text{reg}(S_{j_1},S_{j_2})$. Moreover, by
assumption of this lemma, $P''\notin\{S_1,S_2,S_3\}$. So if
$P'\in\text{reg}^\circ(S_{j_1},S_{j_2})$, then
$P''\in\text{reg}^\circ(S_{j_1},S_{j_3})\cup\text{reg}^\circ(S_{j_2},S_{j_3})$.
\end{proof}

For example, if $\Lambda_j=\{S_1,P\}$, then
$P\in\text{reg}^\circ(S_{2},S_{3})$; If
$\Lambda_j=\{P',P''\}\subseteq\Pi\backslash\{S_1,S_2,S_3\}$ and
$P'\in\text{reg}^\circ(S_{1},S_{2})$, then we have
$P''\in\text{reg}^\circ(S_{1},S_{3})$ or
$P''\in\text{reg}^\circ(S_{2},S_{3})$.

\begin{lem}\label{case1-1}
Suppose $P_1\in\text{reg}^\circ(S_1,S_2)$ and
$P_2,P_3\in\text{reg}^\circ(S_2,S_3)$ such that
$\Lambda_1=\{P_1,P_2\}$ and $\Lambda_2=\{P_1,P_3\}$. Then
$\text{RG}(D^{**})$ is feasible.
\end{lem}
\begin{proof}
If $\Lambda_3\cap\text{reg}^\circ(S_1,S_2)\neq\emptyset$ or
$\Lambda_3\cap\text{reg}^\circ(S_2,S_3)\neq\emptyset$, then by (3)
of Theorem \ref{simple-case}, $\text{RG}(D^{**})$ is feasible. So
we assume
$\Lambda_3\cap(\text{reg}^\circ(S_1,S_2)\cup\text{reg}^\circ(S_2,S_3))
=\emptyset$. Then we have
$\Lambda_3\subseteq\Pi\setminus(\text{reg}^\circ(S_1,S_2)
\cup\text{reg}^\circ(S_2,S_3))=\{S_2\}\cup\text{reg}(S_1,S_3)$.
Moreover, since by 2) of Lemma \ref{in-reg-lmd-t-s},
$\Lambda_3\nsubseteq\text{reg}(S_1,S_3)$, then $S_2\in\Lambda_3$.
Thus, we have
$$S_2\in\Lambda_3\subseteq\{S_2\}\cup\text{reg}(S_1,S_3).$$ Again,
by 2) of Lemma \ref{in-reg-lmd-t-s},
$\Lambda_3\nsubseteq\text{reg}(S_1,S_2)$ and
$\Lambda_3\nsubseteq\text{reg}(S_2,S_3)$. Then either
$\Lambda_3=\{S_1,S_2,S_3\}$ or $\{S_2,P\}\subseteq\Lambda_3$ for
some $P\in\text{reg}^\circ(S_1,S_3)$.

Let $\mathcal I=\{[S_{1}],[S_{2}],[S_{3}],[P_1]\}$, where
$[S_{1}]=\{S_{1}\}$,
$[S_{2}]=\{S_{2}\}\cup\text{reg}^\circ(S_{1},S_{3})$,
$[S_{3}]=\{S_{3}\}$ and $[P_1]=\text{reg}^\circ(S_{1},
S_{2})\cup\text{reg}^\circ(S_{2},S_{3})$. By Definition
\ref{cmptl}, it is easy to check that $\mathcal I$ is a compatible
partition of $\Pi$. So by Lemma \ref{comp-code},
$\text{RG}(D^{**})$ is feasible. An illustration is given in Fig.
\ref{fg-3t-1}.
\end{proof}

%%%%%%%%%%%%%%%%%%%%%%%%%%%%%%%%%%%%%%%%%%%
\renewcommand\figurename{Fig}
\begin{figure}[htbp]
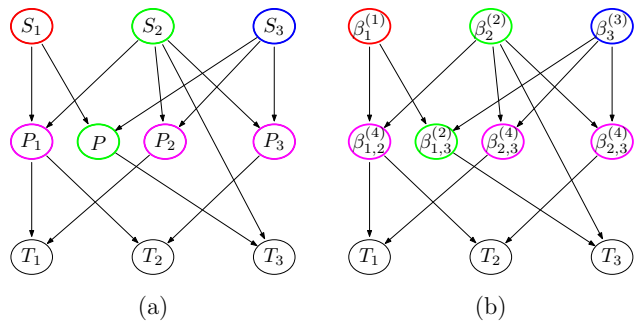

\begin{center}
\includegraphics[height=4.2cm]{fg-3t.1}
\hspace{0.55cm}\includegraphics[height=4.2cm]{fg-3t.2}
\end{center} \caption{An example of code construction: (a) is a region
graph satisfies conditions of Lemma \ref{case1-1} with
$\Lambda_3=\{S_2,P\}$; (b) illustrates a code, where the sets of
coding vectors $\mathcal
B_1=\{\beta^{(1)}_{1},\beta^{(1)}_{2,3}\},\cdots,\mathcal
B_4=\{\beta^{(4)}_{1,2},\beta^{(4)}_{1,3},\beta^{(4)}_{2,3}\}$ are
as in Lemma \ref{genc-code}. We can check that it is still a
feasible code if $\Lambda_3=\{S_1,S_2,S_3\}$.}\label{fg-3t-1}
\end{figure}
%%%%%%%%%%%%%%%%%%%%%%%%%%%%%%%%%%%%%%%%%%%%%%

\begin{lem}\label{case1-2}
Suppose $\Lambda_1=\{P_1,P_2\}$ and $\Lambda_2=\{P_1,P_3\}$ for
some $P_1\in\text{reg}^\circ(S_1,S_2)$,
$P_2\in\text{reg}^\circ(S_2,S_3)$ and
$P_3\in\text{reg}^\circ(S_1,S_3)$. If $\text{RG}(D^{**})$ is
infeasible, then the condition (C-IR) holds.
\end{lem}
\begin{proof}
If $\Lambda_3\subseteq\Pi\setminus\{S_1,S_2,S_3\}$, then by (2) of
Theorem \ref{simple-case}, $\text{RG}(D^{**})$ is feasible. So we
assume that $\Lambda_3\cap\{S_1,S_2,S_3\}\neq\emptyset$. Moreover,
if $|\Lambda_3|\geq 3$, then it is easy to see that $\mathcal
I=\{[P_1]\}\cup\{[R]; R\in\Pi\backslash[P_1]\}$ is a compatible
partition of $\Pi$, where $[P_1]=\{P_1,P_2,P_3\}$ and $[R]=\{R\}$
for all $R\in\Pi\backslash[P_1]$. By Lemma \ref{comp-code},
$\text{RG}(D^{**})$ is feasible. An illustration is given in Fig.
\ref{fg-3t-2} (a). So we further assume that $|\Lambda_3|=2$.
Thus, by 1) of Lemma \ref{s2-ps}, we have the following three
cases:

Case 1: $\Lambda_{3}=\{S_3,P\}$. By 1) of Lemma \ref{s2-ps},
$P\in\text{reg}^\circ(S_1,S_2)$. By (3) of Theorem
\ref{simple-case}, $\text{RG}(D^{**})$ is feasible.

Case 2: $\Lambda_{3}=\{S_2,P\}$. By 1) of Lemma \ref{s2-ps},
$P\in\text{reg}^\circ(S_1,S_3)$. We assert that
$\text{RG}(D^{**})$ is feasible if $P\neq P_3$. In fact, let
$\mathcal I=\{[S_2],[P_1]\}\cup\{[R];
R\in\Pi\backslash([S_2]\cup[P_1])\}$, where $[S_2]=\{S_2,P\},
[P_1]=\{P_1,P_2,P_3\}$ and $[R]=\{R\}$ for all
$R\in\Pi\backslash([S_2]\cup[P_1])$. By Definition \ref{cmptl}, we
can check that $\mathcal I$ is a compatible partition of $\Pi$. By
Lemma \ref{comp-code}, $\text{RG}(D^{**})$ is feasible. An
illustration is given in Fig. \ref{fg-3t-2} (b).

Case 3: $\Lambda_{3}=\{S_1,P\}$ and
$P\in\text{reg}^\circ(S_2,S_3)$. Similar to Case 2, we can prove
that $\text{RG}(D^{**})$ is feasible if $P\neq P_2$.

By the above discussion, if $\text{RG}(D^{**})$ is infeasible,
then either $\Lambda_{3}=\{S_1,P_2\}$ or
$\Lambda_{3}=\{S_2,P_3\}$. Suppose $\Lambda_{3}=\{S_1,P_2\}$. (See
Fig. \ref{fg-3t-3} (a).) Then by renaming regions, condition
(C-IR) holds. (See Fig. \ref{fg-3t-3} (b).) Similarly, if
$\Lambda_{3}=\{S_2,P_3\}$, then by renaming regions, condition
(C-IR) holds.

Thus, we proved that if $\text{RG}(D^{**})$ is infeasible, then
the condition (C-IR) holds.
\end{proof}

%%%%%%%%%%%%%%%%%%%%%%%%%%%%%%%%%%%%%%%%%%%
\renewcommand\figurename{Fig}
\begin{figure}[htbp]
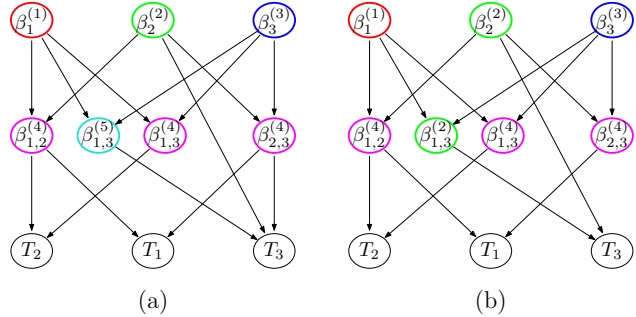

\begin{center}
\vspace{0.0cm}\includegraphics[height=4.2cm]{fg-3t.3}
\hspace{0.55cm}\includegraphics[height=4.2cm]{fg-3t.4}
\end{center} \vspace{-0.3cm}\caption{Illustrations of
code construction for Lemma \ref{case1-2}: (a) illustrates a code
for $|\Lambda_3|\geq 3$; (b) illustrates a code for
$\Lambda_3=\{S_2,P\}$. }\label{fg-3t-2}
\end{figure}
%%%%%%%%%%%%%%%%%%%%%%%%%%%%%%%%%%%%%%%%%%%%%%

%%%%%%%%%%%%%%%%%%%%%%%%%%%%%%%%%%%%%%%%%%%
\renewcommand\figurename{Fig}
\begin{figure}[htbp]
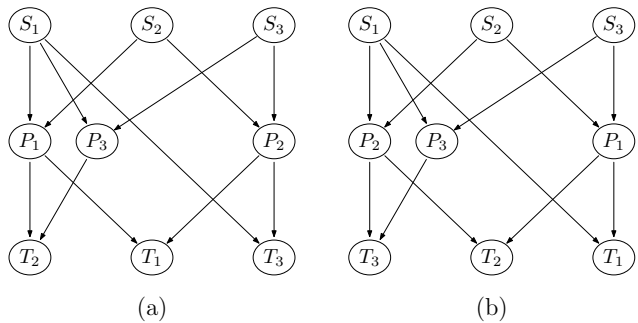

\begin{center}
\vspace{0.2cm}\includegraphics[height=4.2cm]{fg-3t.5}
\hspace{0.55cm}\includegraphics[height=4.2cm]{fg-3t.6}
\end{center} \vspace{-0.3cm}\caption{Illustration of region renaming:
The region graph (b) is obtain from (a) by renaming regions. We
can check that (b) satisfies condition (C-IR). }\label{fg-3t-3}
\end{figure}
%%%%%%%%%%%%%%%%%%%%%%%%%%%%%%%%%%%%%%%%%%%%%%

\begin{lem}\label{case2-1}
Suppose $P_2,P_3\in\text{reg}^\circ(S_2,S_3)$ such that
$\Lambda_1=\{S_1,P_2\}$ and $\Lambda_2=\{S_1,P_3\}$. Then
$\text{RG}(D^{**})$ is feasible.
\end{lem}
\begin{proof}
If $\Lambda_{3}\cap\text{reg}^\circ(S_{2},S_{3})\neq\emptyset$,
then by (3) of Theorem \ref{simple-case}, $\text{RG}(D^{**})$ is
feasible. So we assume
$\Lambda_{3}\cap\text{reg}^\circ(S_{2},S_{_3})=\emptyset$. Then we
have
\begin{align}
\Lambda_{3}\subseteq\Pi\setminus\text{reg}^\circ(S_{2},S_{_3})
=\text{reg}(S_{1},S_{2})\cup\text{reg}(S_{1},S_{3}).
\label{eq-sgl-3t-8}
\end{align}
By enumerating, we have the following two cases:

Case 1: $\Lambda_{3}\cap(\text{reg}^\circ(S_{1},S_{2})\cup
\text{reg}^\circ(S_{1},S_{3}))\neq\emptyset$. Without loss of
generality, assume
$Q_1\in\Lambda_{3}\cap\text{reg}^\circ(S_{1},S_{2})$. Since by 2)
of Lemma \ref{in-reg-lmd-t-s},
$\Lambda_{3}\nsubseteq\text{reg}(S_{1},S_{2})$, then by
(\ref{eq-sgl-3t-8}), there is a
$Q_2\in\Lambda_{3}\cap\text{reg}(S_{1},S_{3})\backslash\{S_{1}\}$.
Let $\mathcal I=\{[S_{1}], [S_{2}], [S_{3}]\}$, where
$[S_{1}]=\{S_{1}\}\cup\text{reg}^\circ(S_{2}, S_{3})$,
$[S_{2}]=\{S_{2}\}$ and $[S_{3}]=\text{reg}^\circ(S_{1},
S_{2})\cup\text{reg}^\circ(S_{1}, S_{3})\cup\{S_{3}\}$. Then by
Definition \ref{cmptl}, we can check that $\mathcal I$ is
compatible. By Lemma \ref{comp-code}, $\text{RG}(D^{**})$ is
feasible. An illustration is given in Fig. \ref{fg-3t-4} (a).

Case 2: $\Lambda_{3}\cap(\text{reg}^\circ(S_{1},S_{2})\cup
\text{reg}^\circ(S_{1},S_{3}))=\emptyset$. By (\ref{eq-sgl-3t-8}),
$\Lambda_{j_3}\subseteq\{S_1,S_2,S_3\}$. On the other hand, by 2)
of Lemma \ref{in-reg-lmd-t-s},
$\Lambda_{3}\nsubseteq\text{reg}(S_{i_1},S_{i_2})$ for all
$\{i_1,i_2\}\subseteq\{1,2,3\}$. Then
$\Lambda_{3}=\{S_1,S_2,S_3\}$. Let $\mathcal
I=\{[S_{1}]\}\cup\{[R]; R\in\Pi\backslash[S_{1}]\}$, where
$[S_{1}]=\{S_{1}\}\cup\text{reg}^\circ(S_{2},S_{3})$ and
$[R]=\{R\}, \forall R\in\Pi\backslash[S_{1}]$. We can check that
$\mathcal I$ is compatible. By Lemma \ref{comp-code},
$\text{RG}(D^{**})$ is feasible. An illustration is given in Fig.
\ref{fg-3t-4} (b).

In both cases, $\text{RG}(D^{**})$ is feasible, which completes
the proof.
\end{proof}

%%%%%%%%%%%%%%%%%%%%%%%%%%%%%%%%%%%%%%%%%%%
\renewcommand\figurename{Fig}
\begin{figure}[htbp]
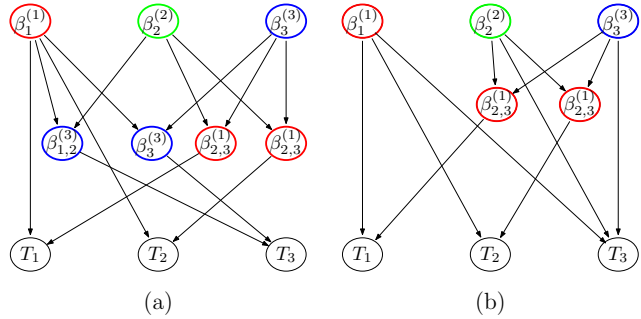

\begin{center}
\vspace{0.0cm}\includegraphics[height=4.2cm]{fg-3t.8}
\hspace{0.3cm}\includegraphics[height=4.2cm]{fg-3t.7}
\end{center} \vspace{-0.3cm}\caption{Illustrations of
code construction for the proof of Lemma \ref{case2-1}: (a)
illustrates a code for Case 1 and (b) illustrates a code for Case
2. }\label{fg-3t-4}
\end{figure}
%%%%%%%%%%%%%%%%%%%%%%%%%%%%%%%%%%%%%%%%%%%%%%

\begin{lem}\label{case2-2}
Suppose $\Lambda_1=\{S_1,P_1\}$ and $\Lambda_2=\{P_1,P_2\}$ for
some $P_1\in\text{reg}^\circ(S_2,S_3)$ and
$P_2\in\text{reg}^\circ(S_1,S_2)$. If $\text{RG}(D^{**})$ is
infeasible, then condition (C-IR) holds.
\end{lem}
\begin{proof}
If $\Lambda_{3}\cap\text{reg}^\circ(S_{2},S_{3})\neq\emptyset$,
then by (3) of Theorem \ref{simple-case}, $\text{RG}(D^{**})$ is
feasible. So we assume
$\Lambda_{3}\cap\text{reg}^\circ(S_{2},S_{3})=\emptyset$. Then we
have
\begin{align}
\Lambda_{3}\subseteq\Pi\setminus\text{reg}^\circ(S_{2},S_{3})
=\text{reg}(S_{1},S_{2})\cup \text{reg}(S_{1},S_{3}).
\label{eq-sgl-3t-2}
\end{align}
Now, suppose
\begin{align}
\Lambda_{3}\nsubseteq\text{reg}(S_{1},P_2)\cup\text{reg}(S_{1},
S_{3}). \label{eq-sgl-3t-3}
\end{align}
By enumerating, we have the following three cases:

Case 1: $\Lambda_{3}\cap\text{reg}(S_{1},P_2)\neq\emptyset$. We
can assume $$Q_1\in\Lambda_{3}\cap\text{reg}(S_{1},P_2).$$ Since
by 2) of Lemma \ref{in-reg-lmd-t-s},
$\Lambda_{3}\nsubseteq\text{reg}(S_{1},S_{2})$, then by
(\ref{eq-sgl-3t-2}), $\Lambda_{3}\cap(\text{reg}S_{1},S_{3})
\backslash\{S_{1}\}\neq\emptyset.$ Assume
$$Q_2\in\Lambda_{3}\cap\text{reg}(S_{1},S_{3})
\backslash\{S_{1}\}.$$ Moreover, by (\ref{eq-sgl-3t-2}) and
(\ref{eq-sgl-3t-3}), $\Lambda_{3}\cap\text{reg}(S_{1},S_{2})
\backslash\text{reg}(S_{1},P_2)\neq\emptyset$. Then we can assume
$$Q_3\in\Lambda_{3}\cap\text{reg}(S_{1},S_{2})
\backslash\text{reg}(S_{1},P_2).$$ Let $\mathcal
I=\{[S_{1}]\}\cup\{[R]; R\in\Pi\backslash[S_{1}]\}$, where
$[S_{1}]=\text{reg}(S_{1},P_2)\cup\{P_1\}$ and $[R]=\{R\}$ for all
$R\in\Pi\backslash[S_{1}]\}$. We can check that $\mathcal I$ is
compatible. By Lemma \ref{comp-code}, $\text{RG}(D^{**})$ is
feasible. An illustration is given in Fig. \ref{fg-3t-5}.

Case 2: $\Lambda_{3}\cap\text{reg}(S_{1},P_2)=\emptyset$ and
$|\Lambda_{3}|\geq 3.$ Similar to Case 1, we can prove that
$\text{RG}(D^{**})$ is feasible.

Case 3: $\Lambda_{3}\cap\text{reg}(S_{1},P_2)=\emptyset$ and
$|\Lambda_{3}|=2.$ Since by 2) of Lemma \ref{in-reg-lmd-t-s},
$\Lambda_{3}\nsubseteq\text{reg}(S_{1},S_{2})$ and
$\Lambda_{3}\nsubseteq\text{reg}(S_{1},S_{3})$, then by
(\ref{eq-sgl-3t-2}) and (\ref{eq-sgl-3t-3}),
$\Lambda_{3}=\{P_3,P_4\}$ for some $P_3\in\text{reg}(S_{1},S_{2})
\backslash\text{reg}(S_{1},P_2)$ and
$P_4\in\text{reg}(S_{1},S_{3})\backslash\{S_{1}\}$. Also by 2) of
Lemma \ref{in-reg-lmd-t-s},
$\Lambda_{3}\nsubseteq\text{reg}(S_{2},S_{3})$, then
$\Lambda_{3}=\{P_3,P_4\}\neq\{S_2,S_3\}$. Let $\mathcal
I=\{[S_{1}],[P_3]\}\cup\{[R];
R\in\Pi\backslash[S_{1}]\cup[P_3]\}$, where
$[S_{1}]=\text{reg}(S_{1},P_2)\cup\{P_1\}$, $[P_3]=\{P_3,P_4\}$
and $[R]=\{R\}$ for all $R\in\Pi\backslash[S_{1}]\cup[P_3]\}$. By
Definition \ref{cmptl}, we can check that $\mathcal I$ is
compatible. By Lemma \ref{comp-code}, $\text{RG}(D^{**})$ is
feasible. An illustration is given in Fig. \ref{fg-3t-6}.

Thus, if $\text{RG}(D^{**})$ is infeasible, then
(\ref{eq-sgl-3t-3}) is violated. So
$\Lambda_{3}\subseteq\text{reg}(S_{1},P_2)\cup\text{reg}(S_{1},
S_{3})$. Combining assumptions of this lemma, we derived condition
(C-IR).
\end{proof}

%%%%%%%%%%%%%%%%%%%%%%%%%%%%%%%%%%%%%%%%%%%
\renewcommand\figurename{Fig}
\begin{figure}[htbp]
\begin{center}
\vspace{0.0cm}\includegraphics[height=4.2cm]{fg-3t.9}
\hspace{0.3cm}\includegraphics[height=4.2cm]{fg-3t.10}
\end{center} \vspace{-0.3cm}\caption{Illustration of
code construction for case 1 in the proof of Lemma \ref{case2-2}:
(a) is the region graph and (b) is a code on the graph.
}\label{fg-3t-5}
\end{figure}
%%%%%%%%%%%%%%%%%%%%%%%%%%%%%%%%%%%%%%%%%%%%%%

%%%%%%%%%%%%%%%%%%%%%%%%%%%%%%%%%%%%%%%%%%%
\renewcommand\figurename{Fig}
\begin{figure}[htbp]
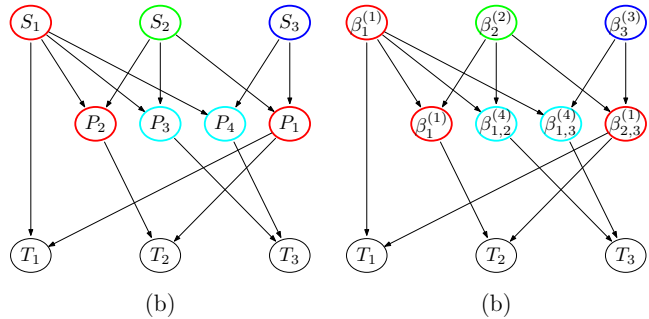

\begin{center}
\vspace{0.0cm}\includegraphics[height=4.2cm]{fg-3t.11}
\hspace{0.3cm}\includegraphics[height=4.2cm]{fg-3t.12}
\end{center} \vspace{-0.3cm}\caption{Illustration of
code construction for case 3 in the proof of Lemma \ref{case2-2}:
(a) is the region graph and (b) is a code on the graph.
}\label{fg-3t-6}
\end{figure}
%%%%%%%%%%%%%%%%%%%%%%%%%%%%%%%%%%%%%%%%%%%%%%

Now we can prove Theorem \ref{sgl-3t}. We will prove the necessity
and sufficiency separately.
\begin{proof}[Proof of Necessity]
Suppose $\text{RG}(D^{**})$ is infeasible. We need to prove the
condition (C-IR). We can prove this by enumerating.

If $\Lambda_{j_1}\cap\Lambda_{j_2}=\emptyset$ for all
$j_1,j_2\subseteq\{1,2,3\}$ such that
$|\Lambda_{j_1}|=|\Lambda_{j_2}|=2$, then by (1) of Theorem
\ref{simple-case}, $\text{RG}(D^{**})$ is feasible. So we assume
that the collection $\{\Lambda_1,\Lambda_2,\Lambda_3\}$ contains
two subsets which have size $2$ and have a non-empty intersection.
By proper naming, we can assume
\begin{align}
\Lambda_{1}=\{P_1,P_2\} ~\text{and}~
\Lambda_{2}=\{P_1,P_3\}.\label{eq-sgl-3t-1}
\end{align}
Then by enumerating, we have the following two cases:

Case 1: $\{P_1,P_2,P_3\}\subseteq\Pi\backslash\{S_1,S_2,S_3\}
=\text{reg}^\circ(S_1,S_2)\cup
\text{reg}^\circ(S_1,S_3)\cup\text{reg}^\circ(S_2,S_3).$ Note that
by (\ref{eq-sgl-3t-1}), $\Lambda_{1}=\{P_1,P_2\}$. Then by 2) of
Lemma \ref{s2-ps} and by proper naming,
\begin{align*}
P_1\in\text{reg}^\circ(S_{1},S_{2}) ~\text{and}~
P_2\in\text{reg}^\circ(S_{2},S_{3}).%\label{eq-sgl-3t-4}
\end{align*}
Again by (\ref{eq-sgl-3t-1}), $\Lambda_{2}=\{P_1,P_3\}$. Then by
2) of Lemma \ref{s2-ps}, $P_3\in\text{reg}^\circ(S_{2},S_{3})$ or
$P_3\in\text{reg}^\circ(S_{1},S_{3})$. If
$P_3\in\text{reg}^\circ(S_{2},S_{3})$, then by Lemma
\ref{case1-1}, $\text{RG}(D^{**})$ is feasible; Otherwise, by
Lemma \ref{case1-2}, the condition (C-IR) holds.

Case 2: $\{P_1,P_2,P_3\}\cap\{S_1,S_2,S_3\}\neq\emptyset$. By
proper naming, we can assume $S_{1}\in\{P_1,P_2,P_3\}$. Then we
can enumerating the following two subcases:

Case 2.1: $S_1=P_1$. Then by (\ref{eq-sgl-3t-1}) and 1) of Lemma
\ref{s2-ps}, $P_2,P_3\in\text{reg}^\circ(S_{2},S_{3})$. By Lemma
\ref{case2-1}, $\text{RG}(D^{**})$ is feasible.

Case 2.2: $S_1=P_2$ or $S_1=P_3$. By proper naming, we can assume
$S_1=P_3$. Since by (\ref{eq-sgl-3t-1}),
$\Lambda_{2}=\{P_1,P_3\}=\{P_1,S_1\}$, then by 1) of Lemma
\ref{s2-ps}, $P_1\in\text{reg}^\circ(S_{2},S_{3})$. Moreover,
since by (\ref{eq-sgl-3t-1}), $\Lambda_{1}=\{P_1,P_2\}$, then by
2) of Lemma \ref{s2-ps} and by proper naming, we have
$P_2\in\text{reg}^\circ(S_{1},S_{2})$. Now, we interchange the
name of $T_1$ and $T_2$. Then by Lemma \ref{case2-2}, the
condition (C-IR) holds.

Combining the above discussions, we can conclude that if
$\text{RG}(D^{**})$ is infeasible, then the condition (C-IR)
holds.
\end{proof}

\begin{proof}[Proof of Sufficiency]
Suppose the condition (C-IR) holds, we need to prove that
$\text{RG}(D^{**})$ is infeasible. We can prove this by
contradiction. Suppose $\tilde{C}_\Pi=\{d_R\in\mathbb F^3;
R\in\Pi\}$ is an arbitrary feasible code on $\Pi$. Since by the
condition (C-IR), $P_2\in\text{reg}^\circ(S_{1},S_{2})$, then by
1) of Lemma \ref{lem-sgl-sets},
\begin{align*}
d_{P_2}\in\langle\alpha_{1},\alpha_{2}\rangle.
\end{align*}
Moreover, since $\Lambda_{1}=\{S_{1},P_1\}$ and
$\Lambda_{2}=\{P_{1},P_2\}$, then by 2) of Lemma
\ref{lem-sgl-sets}, $\langle\bar{\alpha},
d_{S_{1}}\rangle=\langle\bar{\alpha},
d_{P_1}\rangle=\langle\bar{\alpha},d_{P_2}\rangle$. So
\begin{align}\label{eq-sf-sgl-3t-2}
d_{P_2}\in\langle\bar{\alpha},
d_{S_{1}}\rangle\cap\langle\alpha_{1},
\alpha_{2}\rangle=\langle\bar{\alpha},
\alpha_{1}\rangle\cap\langle\alpha_{1},
\alpha_{2}\rangle=\langle\alpha_{1}\rangle.
\end{align}
By (\ref{eq-sf-sgl-3t-2}) and 1) of Lemma \ref{lem-sgl-sets}, we
have $d_R\in\langle\alpha_{1}\rangle$ for all
$R\in\text{reg}(S_{1},P_2)$ and $d_R\in\langle\alpha_{1},
\alpha_{3}\rangle$ for all $R\in\text{reg}(S_{1},S_{3})$. Since by
the condition (C-IR), $\Lambda_{3}\subseteq\text{reg}(S_{1},P_2)
\cup\text{reg}(S_{1},S_{3})$, then by condition (3) of Definition
\ref{sgl-sets} and (\ref{eq-sf-sgl-3t-2}), we have
$$\bar{\alpha}\in\langle d_R; R\in\Lambda_{3}\rangle
\subseteq\langle\alpha_{1}, d_{P_2},
\alpha_{3}\rangle=\langle\alpha_{1}, \alpha_{3}\rangle,$$ which is
a contradiction. Thus, there exists no feasible code on $\Pi$ and
by Theorem \ref{lmd-solv}, $\text{RG}(D^{**})$ is infeasible.
\end{proof}

\begin{thebibliography}{1}
\bibitem{Ahlswede00}
R. Ahlswede, N. Cai, S.-Y. R. Li, and R. W. Yeung, ``Network
information flow,'' \emph{IEEE Trans. Inf. Theory}, vol. 46, no.
4, pp. 1204-1216, Jul. 2000.

\bibitem{Li03}
S.-Y. R. Li, R. W. Yeung, and N. Cai, ``Linear network coding,''
\emph{IEEE Trans. Inf. Theory}, vol. 49, no. 2, pp. 371-381, Feb.
2003.

\bibitem{Korner79}
J. Korner and K. Marton, ``How to encode the modulo-two sum of
binary sources,'' \emph{IEEE Trans. Inf. Theory}, vol. 25, no. 2,
pp. 219-221, 1979.

\bibitem{Gallager88}
R. G. Gallager, ``Finding parity in a simple broadcast network,''
\emph{IEEE Trans. Inf. Theory}, vol. 34, pp. 176-180, 1988.

\bibitem{Orlitsky01}
A. Orlitsky and J. R. Roche, ``Coding for computing,'' \emph{IEEE
Trans. Inf. Theory}, vol. 47, no. 3, pp. 903-917, 2001.

\bibitem{Feng04}
H. Feng, M. Effros and S. A. Savari, ``Functional source coding
for networks with receiver side information,'' in \emph{Proc.
Allerton Conf. Commun., Contr., Comput.}, Sep. 2004.

\bibitem{Giridhar}
A. Giridhar and P. R. Kumar, ``Computing and communicating
functions over sensor networks,'' \emph{IEEE J. Select. Areas
Commun.}, vol. 23, no. 4, pp. 755-764, 2005.

\bibitem{Kanoria}
Y. Kanoria and D. Manjunath, ``On distributed computation in noisy
random planar networks,'' in \emph{Proc. ISIT}, Nice, France,
2008.

\bibitem{Appuswamy09}
R. Appuswamy, M. Franceschetti, N. Karamchandani and K. Zeger,
``Network computing capacity for the reverse butterfly network,''
in \emph{Proc ISIT}, Seoul, Korea, 2009.

\bibitem{Appuswamy11}
R. Appuswamy, M. Franceschetti, N. Karamchandani and K. Zeger,
``Network coding for computing: cut-set bounds,'' \emph{IEEE
Trans. Inf. Theory}, vol. 57, no. 2, pp. 1015-1030, Feb. 2011.

\bibitem{Kannan}
S. Kannan and P. Viswanath, ``Multi-session function computation
and multicasting in undirected graphs,'' \emph{IEEE J. Select.
Areas Commun.}, vol. 31, no. 4, pp. 702-713, 2013.

\bibitem{Rama08}
A. Ramamoorthy, ``Communicating the sum of sources over a
network,'' in \emph{Proc ISIT}, Toronto, Canada, July 06-11, pp.
1646-1650, 2008.

\bibitem{Rai091}
B. K. Rai, B. K. Dey, and A. Karandikar, ``Some results on
communicating the sum of sources over a network,'' in \emph{Proc
NetCod} 2009.

\bibitem{Rai092}
B. K. Rai and B. K. Dey, ``Feasible alphabets for communicating
the sum of sources over a network,'' in \emph{Proc ISIT}, Seoul,
Korea, 2009.

\bibitem{Langberg09}
M. Langberg and A. Ramamoorthy, ``Communicating the sum of sources
in a 3-sources/3-terminals network,'' in \emph{Proc ISIT}, Seoul,
Korea, 2009.

\bibitem{Langberg10}
M. Langberg and A. Ramamoorthy, ``Communicating the sum of sources
in a 3-sources/3-terminals network; revisited,'' in \emph{Proc
ISIT}, Austin, Texas, U.S.A, 2010.

\bibitem{Rai10}
B. K. Rai, B. K. Dey, and S. Shenvi, ``Some bounds on the capacity
of communicating the sum of sources,'' in \emph{Proc IEEE
Information Theory Workshop}, 2010.

\bibitem{Shenvi10}
S. Shenvi and B. K. Dey, ``A necessary and sufficient condition
for solvability of a 3s/3t sum-network,''  in \emph{Proc ISIT},
Texas, U.S.A, 2010.

\bibitem{Rai12}
B. K. Rai and B. K. Dey, ``On Network Coding for Sum-Networks,''
\emph{IEEE Trans. Inf. Theory}, vol. 58, no. 1, pp. 50-63, Jan.
2012.

\bibitem{Rai13}
B. K. Rai and N. Das, ``Sum-Networks: Min-Cut=2 Does Not Guarantee
Solvability,'' \emph{IEEE Communications Letters}, vol. 17, no.
11, pp. 2144-2147, Nov. 2013.

\bibitem{Fragouli06}
C. Fragouli and E. Soljanin,``Information flow decomposition for
network coding,'' \emph{IEEE Trans. Inf. Theory}, vol. 52, no. 3,
pp. 829-848, Mar. 2006.

\bibitem{Wentu11}
W. Song, K. Cai, R. Feng and C. Yuen, ``The Complexity of Network
Coding With Two Unit-Rate Multicast Sessions,'' \emph{IEEE Trans.
Inf. Theory}, vol. 59, no. 9, pp. 5692-5707, Sept. 2013.

\bibitem{Wentu12}
W. Song, R. Feng, K. Cai and J. Zhang, ``Network Coding for
2-unicast with Rate (1,2)'' in \emph{Proc. ISIT}, Boston, 2012.
\end{thebibliography}
\end{document}